\documentclass[letterpaper,onecolumn,11pt,accepted=2023-02-20]{quantumarticle}
\pdfoutput=1
\usepackage[margin=0.4in]{geometry}
\usepackage[]{graphicx}
\usepackage{amssymb}
\usepackage{enumitem}
\usepackage{hyperref}
\usepackage{amsmath}
\usepackage{amsthm}
\usepackage{multicol}
\usepackage{qcircuit}
\usepackage{braket}
\usepackage{caption}
\usepackage{subcaption}
\usepackage{xcolor}
\usepackage{svg}
\usepackage[htt]{hyphenat}

\newcommand{\eps}{\varepsilon}
\newcommand{\erf}{\text{erf}}
\newcommand{\sgn}{\text{sgn}}
\newcommand{\rect}{\text{rect}}
\newcommand{\bigoh}[1]{\mathcal{O}(#1)}

\newtheorem{theorem}{Theorem}
\newtheorem*{theorem*}{Theorem}

\newtheorem{definition}[theorem]{Definition}
\newtheorem{lemma}[theorem]{Lemma}
\newtheorem{remark}[theorem]{Remark}

\newtheorem{proposition}[theorem]{Proposition}
\newtheorem{fact}[theorem]{Fact}
\newtheorem{empiricalclaim}[theorem]{Empirical Claim}
\newtheorem*{empiricalclaim*}{Empirical Claim}
\begin{document}

\title{Amplitude Estimation from Quantum Signal Processing}

\author{Patrick Rall}
\email{patrickjrall@ibm.com}
\affiliation{IBM Quantum, MIT-IBM Watson AI Lab, Cambridge, Massachusetts 02142, USA}
\author{Bryce Fuller}
\email{Bryce.Fuller@ibm.com}
\affiliation{IBM Quantum, Thomas J Watson Research Center, Yorktown Heights, New York 10598, USA}

\maketitle

\begin{abstract} Amplitude estimation algorithms are based on Grover's algorithm: alternating reflections about the input state and the desired outcome. But what if we are given the ability to perform arbitrary rotations, instead of just reflections? In this situation, we find that quantum signal processing lets us estimate the amplitude in a more flexible way. We leverage this technique to give improved and simplified algorithms for many amplitude estimation tasks: we perform non-destructive estimation without any assumptions on the amplitude, develop an algorithm with improved performance in practice, present a new method for unbiased amplitude estimation, and finally give a simpler method for trading quantum circuit depth for more repetitions of short circuits.
\end{abstract}

\newcommand{\maindocument}{}

\ifdefined\maindocument
\else
    \documentclass[11pt]{article}
    \usepackage[margin=0.4in]{geometry}
    \usepackage[margin=0.4in]{geometry}
\usepackage[]{graphicx}
\usepackage{amssymb}
\usepackage{enumitem}
\usepackage{hyperref}
\usepackage{amsmath}
\usepackage{amsthm}
\usepackage{multicol}
\usepackage{qcircuit}
\usepackage{braket}
\usepackage{caption}
\usepackage{subcaption}
\usepackage{xcolor}
\usepackage{svg}
\usepackage[htt]{hyphenat}

\newcommand{\eps}{\varepsilon}
\newcommand{\erf}{\text{erf}}
\newcommand{\sgn}{\text{sgn}}
\newcommand{\rect}{\text{rect}}
\newcommand{\bigoh}[1]{\mathcal{O}(#1)}

\newtheorem{theorem}{Theorem}
\newtheorem*{theorem*}{Theorem}
\newtheorem{corollary}[theorem]{Corollary}
\newtheorem{definition}[theorem]{Definition}
\newtheorem{lemma}[theorem]{Lemma}
\newtheorem{remark}[theorem]{Remark}
\newtheorem{example}[theorem]{Example}
\newtheorem{proposition}[theorem]{Proposition}
\newtheorem{fact}[theorem]{Fact}
\newtheorem{empiricalclaim}[theorem]{Empirical Claim}
\newtheorem*{empiricalclaim*}{Empirical Claim}
    \begin{document}
    \begin{center}
        {\Large Introduction}
    \end{center}
\fi

    Amplitude estimation is a fundamental subroutine in quantum algorithms. It directly yields quantum speedups that improve the accuracy of Monte Carlo estimation \cite{1504.06987}, which has applications in many areas of science and technology. But it also forms a fundamental component of more complicated quantum algorithms like those presented in \cite{1907.09965, 2009.11270}. A general description of the problem is as follows: given a state $\ket{\psi}$ and a projector $\Pi$, estimate $a := |\Pi \ket{\psi}|$. There exist algorithms that can obtain an estimate $\hat a$ satisfying $|a - \hat a| \leq \eps$ with high probability in $O(\eps^{-1})$ queries to oracles for $\ket{\psi}$ and $\Pi$. 

The first algorithm for this task was given by \cite{0005055}, and it is based on applying phase estimation to the Grover operator which consists of the reflections $Z_\psi Z_\Pi$. This method has since become widely used in the literature. But, starting in 2019, there have been several efforts to improve the simplicity, constant-factor performance, and versatility of amplitude estimation. 
    
    One line of work studies algorithms that do not require the phase estimation step from the \cite{0005055} method. These algorithms simply perform Grover's algorithm: they apply alternating reflections $Z_\Pi$ and $Z_\psi$ to a copy of $\ket{\psi}$, and then measure the $\Pi$ observable. Via repeated measurements we can extract information about $a$. \cite{1904.10246} gave an algorithm based on maximum-likelihood estimation that has improved constant-factor performance over \cite{0005055} and is non-adaptive, but there is no proof that the algorithm is correct.  \cite{1908.10846} gave an algorithm for relative-error amplitude estimation with a proof of correctness, although the constant-factor performance is poor. \cite{1912.05559} presented an additive-error estimation algorithm which is provably correct, and empirically outperforms \cite{0005055} and roughly matches the performance of \cite{1904.10246}. \cite{2010.04370} gave a relative-error estimation algorithm that is also non-adaptive. Finally, \cite{2012.03348} present a family of estimation algorithms with trade-offs between the circuit depth and the total number of samples. Together, these techniques make amplitude estimation simpler, faster, and more versatile.
    
    A separate line of work studies quantum algorithms for approximating partition functions that leverage amplitude estimation as a subroutine. In an effort to improve an algorithm from \cite{1504.06987},  \cite{1907.09965} modify the \cite{0005055} method in order to be `non-destructive'. In this situation, $\ket{\psi}$ is much harder to prepare than to reflect about. Thus, we only want to use one copy of $\ket{\psi}$ for the estimation algorithm, and return it undamaged when done. The non-destructive method by  \cite{1907.09965} was later used again in \cite{2009.11270}.

    Recent literature has presented many different quantum algorithms, all of which improve over \cite{0005055} in some particular way. In this manuscript we build a framework that unifies the underlying quantum subroutine of some of these methods. This allows us to inherit several of these improvements simultaneously, whenever sensible. This `unification' of amplitude estimation algorithms is achieved through Quantum Signal Processing (QSP). 
    
    QSP is a powerful technique for building quantum algorithms, capable of unifying many other existing results \cite{1806.01838, 2105.02859}. While existing techniques only consider reflections $Z_\Pi$ and $Z_\psi$ around $\Pi$ and $\ket{\psi}$, we consider arbitrary rotations $e^{i \phi Z_\Pi}$ and $e^{i \phi Z_\psi}$. We find that in situations where Grover's algorithm is implemented, it is generally possible to implement these arbitrary rotations at no additional cost.  The dynamics of such quantum circuits are well studied \cite{1409.3305, 1603.03996, 1806.01838, 2105.02859}: we find that careful choice of phases $\phi$ lets us implement a family of polynomials $P(a)$ in the amplitude. This yields our first result:

\begin{theorem*} (Section~\ref{sec:polysamp}, Theorem~\ref{thm:polysamp}) \textbf{Sampling from polynomials.} Say we are given access to one copy of $\ket{\psi}$ as well as the oracles  $e^{i \phi Z_\psi}$ and $e^{i \phi Z_\Pi}$. Then it is possible to repeatedly toss coins with $\text{Pr}[\text{heads}] = |P(a)|^2$ where $P(a)$ can be freely chosen from a certain family of polynomials. The number of oracle queries per coin toss is $O(\text{deg}(P))$.
\end{theorem*}
 
    We find that all the algorithms mentioned above \emph{except} for those presented in \cite{0005055, 1907.09965} can be formulated as repeated invocations of special cases of this theorem. This already brings many previous results into a single framework. By considering the full range of polynomials, we find this tool for amplitude estimation can improve or simplify many of the existing algorithms. 

    In Section~\ref{sec:repair} we consider the setting where the state $\ket{\psi}$ is expensive to prepare, but the reflection $Z_\psi$ is relatively cheap. This appears in algorithms for estimating observables on ground states such as \cite{2002.12508}, where the manner by which $\ket{\psi}$ is prepared in the first place is by constructing a reflection about the ground space $Z_\psi$ and then employing amplitude amplification. It also appears in the partition function estimation algorithms in \cite{1504.06987, 1907.09965, 2009.11270}, where a sequence of qsamples of thermal states is constructed through amplitude amplification. Here, we can reflect about any of these states using Szegedy walk unitaries, but re-preparing the states corresponding to low temperatures is expensive. To aid in settings like these, \cite{1907.09965} gave a method for non-destructive amplitude estimation that restores a copy of $\ket{\psi}$ after the algorithm completes.  That way, each qsample state can be re-used.

    However, we find that the non-destructive estimation algorithm presented in \cite{1907.09965} makes a subtle implicit assumption: that some $\kappa$ is known such that $\kappa < a < \sqrt{1-\kappa^2}$.  But since the whole point of amplitude estimation is to determine the value of $a$, such a bound may not be known beforehand. We resolve this issue by developing a procedure that requires no such bound. Furthermore, our procedure can be combined with any algorithm in the polynomial sampling framework above. However, in the process we incur a polynomial slowdown in the failure probability $\delta$, which only appears polylogarithmically in the overall complexity when a bound $\kappa$ is known.

    \begin{theorem*} (Section~\ref{sec:repair}, Theorem~\ref{thm:repair}) \textbf{State repair.} Say we just sampled from several polynomials, and the sum of the degrees of these polynomials is $D$. Then we can restore the quantum state to a copy of $\ket{\psi}$ with success probability $\geq 1-\delta$ by sampling from one additional polynomial with degree $O(\delta^{-1/2} D)$.
    \end{theorem*}

    We remark that there was also a non-destructive amplitude estimation algorithm presented in \cite{2103.09717}, but this algorithm estimates the probability $a^2$, not $a$, which is weaker\footnote{When converting an additive-error estimate of $a^2$ to an additive-error estimate of $a$, the error is blown up by roughly a factor of $1/a$. This can be expensive when $a$ is small. Converting between $a$ and the Grover angle $\theta := \arcsin(a)$ is not so expensive.}. The algorithms in \cite{1907.09965, 2009.11270} rely on methods for mean estimation from \cite{1504.06987} which in turn rely on estimates of $\arcsin(a) \sim a$, and it is not clear how to generalize them to rely on $a^2$ instead. 
    
    \textbf{Parallel work.} \cite{2207.08643} simultaneously presented a different method for achieving non-destructive amplitude estimation algorithm that also improves over that of \cite{1907.09965}.

    Next, in Section~\ref{sec:performance} we ask if the polynomial sampling framework can reduce the constant factors in the number of queries required. Since we want amplitude estimation to be useful in practice, it is essential to study which method can minimize the number of queries, even though the asymptotic performance is already determined to be $\Theta(\eps^{-1})$ \cite{9804066}. The prior state of the art for amplitude estimation with high empirical performance is Iterative Quantum Amplitude Estimation (IQAE), the algorithm presented by \cite{1912.05559}, which, according to our experiments, requires about $\approx 9.93/\eps$ queries to the $Z_\Pi$ oracle to achieve a success probability of $\geq 95\%$.

    IQAE relies on the oracles $Z_\Pi$ and $Z_\psi$, which can be viewed as special cases of the $e^{i\phi Z_\Pi}$ and $e^{i\phi Z_\psi}$ oracles when $\phi = \pi/2$. In this case, the polynomial $P(a)$ is a Chebyshev polynomial. While in this sense IQAE is restricted to odd Chebyshev polynomials, the polynomial sampling framework admits Chebyshev polynomials of both even and odd degree. But except for this modest increase in flexibility, we find that it is not particularly useful to deviate from $\phi = \pi/2$ in order to improve constant-factor performance. This makes intuitive sense: Chebyshev polynomials exhibit the maximum rate of change among polynomials bounded by $|P(a)| \leq 1$ for $a \in [0,1]$, and we even rely on this fact in Section~\ref{sec:repair}. Thus, as a rule of thumb, deviating from Chebyshev polynomials should waste resources.

    However, we still find that it is possible to improve over the performance of IQAE. IQAE features a `no-overshooting' condition that prevents the algorithm from taking too many samples from the most expensive polynomials near the end of execution. We find that this no-overshooting condition can be tuned in order to significantly improve performance. We combine this optimization with the ability to estimate the amplitude $a$ directly, rather than the Grover angle $\theta := \arcsin(a)$, and the ability to sample from both odd and even Chebyshev polynomials, into a new algorithm we call `ChebAE'.
    
    \begin{empiricalclaim*} (Section~\ref{sec:performance}, Empirical~Claim~\ref{empericalclaim:chebae}) \textbf{Chebyshev amplitude estimation `ChebAE'.} There exists an amplitude estimation algorithm that samples from Chebyshev polynomials and returns an estimate of the amplitude with error $\eps$ and success probability $\geq 95\%$ in about $\approx 4.66/\eps$ queries to $Z_\Pi$.
\end{empiricalclaim*}

   Depending on the details of the analysis, ChebAE only requires about $45\%$ to $65\%$ of the queries of IQAE depending on the analysis methodology. To support this claim, we performed a comprehensive numerical study where we benchmark the query complexity of the relevant algorithms in the literature.

    In Section~\ref{sec:unbiased}, we study a problem inspired by recent work \cite{2109.10215, 2207.08800}: unbiased amplitude amplitude estimation. An amplitude estimation algorithm effectively samples from a random variable $\mathbf{\hat{a}}$ that satisfies $| \mathbf{\hat{a}} - a |\leq \eps$ with high probability. For unbiased estimation, we additionally demand that $\mathbb{E}[\mathbf{\hat{a}}] \approx a$. 

    This is in contrast to the unbiased \emph{phase} estimation problem considered by \cite{2109.10215, 2207.08800}, which, given an oracle for $e^{i \theta}$, samples from an unbiased estimator $\mathbf{\hat{\boldsymbol{\theta}}}$ of $\theta$. Their solution to the problem is based on applying a random shift to the angle $\theta$ before estimation, and accounting for it afterward. It is not possible to use \cite{0005055} to convert unbiased phase estimation to unbiased amplitude estimation since the relationship $a = \sin\theta $ is not linear. Instead, we invoke Jackson's theorem \cite{rivlin} to build a polynomial that guarantees $\mathbb{E}[\mathbf{\hat{a}}] \approx a$ directly.

    \begin{theorem*} (Section~\ref{sec:unbiased}, Theorem~\ref{thm:unbiasedamplitudeestimation}) \textbf{Unbiased amplitude estimation.} For any $\eps,\eta,\delta$, there exists an amplitude estimation algorithm that samples from a random variable $\mathbf{\hat{a}}$ that satisfies $\text{Pr}[ | \mathbf{\hat{a}} - a |\geq \eps] \leq \delta$ and $|\mathbb{E}[\mathbf{\hat{a}}] - a| \leq \eps \eta + \delta$. It has a query complexity of $O(\eps^{-1}(\eta^{-1} + \log(\delta^{-1})))$.
    \end{theorem*}

\textbf{Parallel work.} While we were unaware of any applications of this subroutine in other quantum algorithms as we were writing this manuscript, \cite{2207.08643} were able to improve partition function estimation using an unbiased amplitude estimation algorithm of their own. 

    Finally, in Section~\ref{sec:hybrid} we revisit a problem studied by \cite{2012.03348}: trading classical and quantum resources in amplitude estimation. Early quantum computers may be able to perform the oracles $e^{i \phi Z_\Pi}$ and $e^{i \phi Z_\psi}$, but may not be able to perform them $O(\eps^{-1})$ many times within a single circuit as required by most amplitude estimation algorithms. \cite{2012.03348} develop two algorithms that, for any $\beta \in [0,1)$, have query complexity $O(\eps^{-(1+\beta)})$ and query depth $ O(\eps^{-(1-\beta)})$.  When $\beta =0$, this corresponds to the traditional scaling of amplitude estimation, but as $\beta \to 1$ we approach the complexity of classical probability estimation.

    The algorithms presented in \cite{2012.03348} are rather complicated: one algorithm relies on maximum likelihood estimation, which functions very well in practice but rigorously proving its accuracy and convergence rate is challenging. The other algorithm pieces together several estimates of multiples of $a$ via the Chinese remainder theorem. We find that there exists a simpler procedure based on careful construction of polynomials. The polynomials approximate a line with slope $k$ near the current best estimate of $a$, and have degree $O(k)$. By selecting $k \in O(\eps^{-(1-\beta)})$ and we can achieve depth $\eps^{-(1-\beta)}$ as desired. 

    The main technical challenge in constructing our polynomial-based method is rigorously proving bounds on the quality of the polynomial approximation. We find that Jackson's theorem is not accurate enough, so we must instead rely on a construction based on the Jacobi-Anger expansion (see \cite{1707.05391}). However, we find that we can analyze our method with greater rigor than that of  \cite{2012.03348}, since we do not need to implicitly assume that $a,\beta \in \Theta(1)$ as they do. 

    \begin{theorem*} [Section~\ref{sec:hybrid}, Theorem~\ref{thm:hybridestimation}] \textbf{Hybrid quantum-classical amplitude estimation.} For any $\eps,\delta > 0$ and $0 \leq \beta < 1$, there exists an amplitude estimation algorithm that estimates $a$ to accuracy $\eps$ with probability $\geq 1-\delta$. It makes at most $\tilde O(  (a^{-1} + \eps^{-(1+\beta)}) \log(\delta^{-1}) )$  queries, and the circuits have depth at most $O( a^{-1/(1-\beta)} + \eps^{-(1-\beta)} )$
    \end{theorem*}

Our analysis reveals an additional $O(a^{-1})$ dependence, which we suspect is also present for the algorithms presented in \cite{2012.03348}. Of course, if we treat $a$ and $\beta$ as constants then we recover the same complexity.

\ifdefined\maindocument
\else

    \end{document}
\fi

\ifdefined\maindocument
    \section{Sampling from Polynomials \label{sec:polysamp}}
\else
    \documentclass[11pt]{article}
    \usepackage[margin=0.4in]{geometry}
    \usepackage[margin=0.4in]{geometry}
\usepackage[]{graphicx}
\usepackage{amssymb}
\usepackage{enumitem}
\usepackage{hyperref}
\usepackage{amsmath}
\usepackage{amsthm}
\usepackage{multicol}
\usepackage{qcircuit}
\usepackage{braket}
\usepackage{caption}
\usepackage{subcaption}
\usepackage{xcolor}
\usepackage{svg}
\usepackage[htt]{hyphenat}

\newcommand{\eps}{\varepsilon}
\newcommand{\erf}{\text{erf}}
\newcommand{\sgn}{\text{sgn}}
\newcommand{\rect}{\text{rect}}
\newcommand{\bigoh}[1]{\mathcal{O}(#1)}

\newtheorem{theorem}{Theorem}
\newtheorem*{theorem*}{Theorem}
\newtheorem{corollary}[theorem]{Corollary}
\newtheorem{definition}[theorem]{Definition}
\newtheorem{lemma}[theorem]{Lemma}
\newtheorem{remark}[theorem]{Remark}
\newtheorem{example}[theorem]{Example}
\newtheorem{proposition}[theorem]{Proposition}
\newtheorem{fact}[theorem]{Fact}
\newtheorem{empiricalclaim}[theorem]{Empirical Claim}
\newtheorem*{empiricalclaim*}{Empirical Claim}
    \begin{document}
    \begin{center}
        {\Large Sampling from Polynomials}
    \end{center}
\fi

We begin by establishing our general framework for amplitude estimation algorithms. All future sections will be based on the main result of this section, which, informally, is as follows:  Say $P(a)$ is a `Pellian' or a `semi-Pellian' polynomial in $a$. Then, there exists a quantum algorithm that samples from a random variable over $\{0,1\}$ with bias $|P(a)|^2$. This algorithm makes $O(\text{deg}(P))$ oracle queries to $\ket{\psi}$ and $\Pi$.

The definition of `Pellian' and `semi-Pellian', the exact nature of the oracles, as well as the input and output states of this algorithm will be specified formally in this section. In particular, the outline for the rest of the section is as follows. First, we establish some general notation and discuss the implementation of our oracles. Second, we define the restricted family of `Pellian' polynomials, and give the simpler algorithm for sampling from them. Third, we define the more general family of `semi-Pellian' polynomials and give a more complicated algorithm for sampling from them. Finally, we put all these results together with some standard notation and show how some existing results fit into this framework.

\subsection{Notation and Oracles}
    
We begin by establishing some notation that we borrow from the discussion of Grover's algorithm. Our amplitude estimation framework is very closely related to Grover's algorithm. We are given a state $\ket{\psi}$ and a projector $\Pi$ and want to estimate $a = |\Pi \ket{\psi}|$. Let $\bar a := \sqrt{1-a^2}$, and assume for a moment that $a \not\in \{0,1\}$. Then let:
\begin{align}
    \ket{\Pi} &:= a^{-1} \cdot \Pi\ket{\psi} \\
    \ket{\Pi^\perp} &:= \bar a^{-1} \cdot (I - \Pi)\ket{\psi} 
\end{align}

Observe that $\bar a = |(I-\Pi)\ket{\psi}|$. We see how $\ket{\Pi}$ and $\ket{\Pi^\perp}$ satisfy:
\begin{align}
    \ket{\psi} = a \ket{\Pi} + \bar a \ket{\Pi^\perp}
\end{align}
We furthermore define:
\begin{align}
    \ket{\psi^\perp} := \bar a \ket{\Pi} - a \ket{\Pi^\perp}
\end{align}
Finally, observe that $\braket{\Pi|\Pi^\perp} = \braket{\psi|\psi^\perp} = 0$. We have constructed two orthonormal bases for the two-dimensional Grover subspace:  $\ket{\psi},\ket{\psi^\perp}$ and $\ket{\Pi},\ket{\Pi^\perp}$. 

When $a \in \{0,1\}$ then all the dynamics in the Grover subspace are trivial anyway. Nonetheless we can treat these cases along with $a \not\in \{0,1\}$ in a unified manner. If $a = 0$, let $\ket{\Pi^\perp} := \ket{\psi}$ and select an arbitrary orthogonal state $\ket{\Pi} = \ket{\psi^\perp}$. Similarly, if $a = 1$, let $\ket{\Pi} := \ket{\psi}$ and let $\ket{\Pi^\perp} = \ket{\psi^\perp}$ be an arbitrary orthogonal state.

In the traditional presentation of Grover's algorithm we are given access to reflection oracles $2\ket{\psi}\bra{\psi} - I$ and $2\Pi - I$. Projected down into the two-dimensional subspace, we can write these as:
\begin{align}
    Z_\psi := \ket{\psi}\bra{\psi} - \ket{\psi^\perp}\bra{\psi^\perp}\\
    Z_\Pi := \ket{\Pi}\bra{\Pi} - \ket{\Pi^\perp}\bra{\Pi^\perp}
\end{align}

In this work we consider what is possible when we are instead given access to a larger family of oracles $e^{i \phi Z_\psi}$ and $e^{i \phi Z_\Pi}$ for arbitrary phases $\phi$. These can always be implemented using two queries to controlled-$Z_\psi$ or controlled-$Z_\Pi$, by using a Hadamard test. For example, to implement $e^{i \phi Z_\psi}$:
\begin{align}
    \Qcircuit @C=1em @R=1em {
    \lstick{\ket{+}} & \ctrl{1}     & \gate{H} & \gate{e^{i\phi Z}} & \gate{H} & \ctrl{1} & \rstick{\ket{+}} \qw  \\
                     & \gate{Z_\psi} & \qw  & \qw & \qw & \gate{Z_\psi}    & \qw
}
\end{align}
This lets us implement $e^{i \phi Z_\psi}$ using two queries to controlled-$Z_\psi$. However, we argue that in essentially all practical situations the oracles $e^{i \phi Z_\psi}, e^{i \phi Z_\Pi}$ do not have twice the complexity of $Z_\psi, Z_\Pi$.  Here are some examples:

\begin{itemize}
    \item Say we have a function $f: \{0,1\}^n \to \{0,1\}$ and some classical program that computes it. In the usual situation for Grover's search, we apply reflections $Z_\Pi$ where $\Pi = \sum_x f(x) \ket{x}\bra{x}$ by rendering the classical program into a quantum circuit via Toffoli gates. That circuit produces an output bit $\ket{f(x)}$ plus garbage. To build $Z_\Pi$, we apply $Z$ to the output bit, and then uncompute. But we could also just apply $e^{i\phi Z}$ to the output bit instead of $Z$, and obtain $e^{i\phi Z_\Pi}$ with the same number of gates.
        
    \item Say we have a preparation unitary $U_\psi$ for an $n$-qubit state $\ket{\psi}$, satisfying $U\ket{0^n} = \ket{\psi}$.  Then, can implement $Z_{0^n} := 2\ket{0^n}\bra{0^n} - I$ via a some Toffoli gates and a $Z$ gate, and $Z_\psi$ via $U_\psi Z_{0^n}U_\psi^\dagger$. But notice that implementing $e^{i \phi Z_{0^n}}$ requires the same number of Toffoli gates as implement $Z_{0^n}$. Therefore the circuit for $e^{i \phi Z_\psi} = U e^{i \phi Z_{0^n}}U^\dagger$ has the same complexity.
    
    \item Say $\ket{\psi}$ is the ground state of a Hamiltonian $H$, and assume the ground space is non-degenerate. Modern algorithms for estimating ground state energies or other observables on the ground state (such as \cite{2002.12508}) synthesize approximate reflections via singular value transformation: a polynomial $p(x)$ is approximately $1$ in the ground space, and $-1$ elsewhere, so that $p(H) \approx Z_\psi$. But we can also design polynomials that are $e^{i\phi}$ in the ground space and $e^{-i\phi}$ elsewhere, and then $p(H) \approx e^{i\phi Z_\psi}$. These polynomials have the same degree, so they have the same circuit complexities.

    \item  Quantum algorithms based on Szegedy walk unitaries (such as \cite{1907.09965, 2009.11270}) prepare a unitary $U$ that has an eigenvalue $\ket{\psi}$ with eigenphase $0$, and the other eigenphases are at least $\Delta$. To implement $Z_\psi$, they apply phase estimation to $U$ to extract an estimate of the eigenphase with precision $\Delta$ -- enough to distinguish $0$ from values larger than $\Delta$. Then they apply a phase flip depending on if the estimate is $0$, and uncompute. Similar to the Toffoli circuit, we can also simply apply $e^{i\phi}$ if the estimate is $0$ and $e^{-i\phi}$ otherwise, and we have implemented $e^{i \phi Z_\psi}$ with the same circuit complexity.
    
\end{itemize}

It seems that in the black-box setting two queries to $Z_\Pi$ are necessary to implement $e^{i\phi Z_\Pi}$. But once we `open the black-box' and actually implement it with a quantum circuit, we find that for every example in the literature that we know of the circuits implementing $Z_\Pi$ and $e^{i \phi Z_\Pi}$ are essentially the same size. As we shall see, access to these oracles significantly expands our capabilities.

In a fault tolerant setting, we may find that the complexity of executing a circuit may depend less on its overall size and more on the number of non-Clifford gates. In this case, $e^{i \phi Z}$ may be significantly more expensive than $Z$, since $Z$ is a Clifford gate and $e^{i \phi Z}$ is not. The severity of this effect may vary with the level of precision with which phases $\phi$ are implemented \cite{1603.04230}. We leave an investigation of these concerns to future work.

\subsection{Sampling from Pellian Polynomials}

In this section we present an algorithm based on alternating applications of $e^{i\phi Z_\psi}$ and $e^{i \phi Z_\Pi}$. The behavior of such alternating rotations is extensively studied, initially by \cite{1409.3305, 1603.03996}. We find that this lets us toss a coin that comes up heads with probability $|P(a)|^2$ where $P$ belongs to a certain restricted class of polynomials.

Indeed, \cite{2105.02859} presents a derivation of quantum signal processing in terms of amplitude amplification. Our framework can be seen as a synthesis of this observation with the idea of using quantum signal processing to perform Kitaev's phase estimation algorithm on the Grover operator.

\cite{1806.01838} gives a precise characterization of what these polynomials are that is completely independent of the actual phase rotations $\phi$.  This is exceedingly convenient for quantum algorithms, since we no longer need to worry about circuit synthesis and can simply construct polynomial approximations to the functions we want to apply. We just need to make sure the polynomial we construct meets some simple constraints. Actually computing the phase rotations that implement the polynomial is studied by several recent works \cite{1806.10236, 2003.02831, 2110.04993}.

However, one thing that is lacking from the literature is simple nomenclature for the different families of polynomials. \cite{2003.02831} introduces the `Low-algebra' and `Haah-algebra' for different families of implementable matrices, and \cite{2105.02859} introduced the idea of `Quantum Signal Processing (QSP) Conventions' such as $(W_x,S_z,\bra{0}\cdot\ket{0})$-QSP. These ideas are useful for understanding the inner workings of polynomial synthesis. But if all we want to do is refer to synthesizable polynomials without worrying about how to find the phases, then these terms are still too technical. 

Instead, we are motivated by the fact that the Chebyshev polynomials are a family of polynomials that are implementable via quantum signal processing. \cite{Dem07} furthermore observes that the Chebyshev polynomials of the first kind $T_n$ and of the second kind $U_n$ are solutions of the `Pell equation':
\begin{align}
    |T_n(x)|^2 + (1-x^2)|U_{n-1}(x)|^2 = 1.
\end{align}

Indeed, the Pell equation is satisfied by exactly the polynomials that are directly implementable by  quantum signal processing. We thus propose the following nomenclature:

\begin{definition} A pair of polynomials $P(x), Q(x) \in \mathbb{C}[x]$ \textbf{form a Pell pair} if they have fixed and opposite parity\footnote{A polynomial is even (odd) if it is an even (odd) function. If it is either even or odd then we say it has fixed parity, otherwise it has mixed parity. Two polynomials have opposite parity if one of them is even and the other is odd.} and, for all $x \in [-1,1]$, they together satisfy:
    \begin{align}
        |P(x)|^2 + (1-x^2)|Q(x)|^2 = 1.
    \end{align}
    A polynomial $P(x) \in \mathbb{C}[x]$ is \textbf{Pellian} if there exists a $Q(x) \in \mathbb{C}[x]$ such that $P(x),Q(x)$ form a Pell pair. Similarly, a polynomial $Q(x)$ is \textbf{Pell-complementary} if there exists a $P(x)$ such that $P(x),Q(x)$ form a Pell pair.
\end{definition}

The condition is not symmetrical with respect to $P(x)$ and $Q(x)$, and therefore the different classes of polynomials have different names. While the above definition is independent of the existence of the phases, it still does not make it easy to check if a polynomial $P(x)$ is Pellian. This is facilitated by Theorem~3 of \cite{1806.01838}, which states that a polynomial $P(x) \in \mathbb{C}[x]$ is Pellian if and only if it has fixed parity, and:
    \begin{itemize}
        \item $\forall x \in [-1,1]$ we have $|P(x)| \leq 1$,
        \item $\forall x \in (-\infty,-1]\cup [1,\infty)$ we have $|P(x)| \geq 1$,
        \item if $P$ is even then $\forall x \in\mathbb{R}$ we have $P(ix)P^{*}(ix) \geq 1$.
    \end{itemize}
Here, $P^*$ refers to complex conjugation of the coefficients of $P$ only.

We have defined the family of polynomials in terms of simple conditions independent of the implementation. Now we move on to the actual algorithm that samples from these polynomials.

The core idea of quantum signal processing is to apply alternating rotations in different bases. In our case, this will correspond to alternatingly applying $e^{i \phi Z_\psi}$ and $e^{i \phi Z_\Pi}$. Keeping track of the dynamics when both bases are rotating is very complicated, so in our analysis we find it useful to standardize the bases. We let $\{\ket{0},\ket{0^\perp}\}$ be some other basis for the two-dimensional Grover subspace, and introduce a new unitary $V$ that is used entirely for analysis purposes and never appears in the implementation. 
    \begin{align}
        V &:= \ket{\psi}\bra{0} + \ket{\psi^\perp}\bra{0^\perp}\\
        Z_0 &:= \ket{0}\bra{0} - \ket{0^\perp}\bra{0^\perp}\\
        e^{i \phi Z_\psi} &= V e^{i \phi Z_0} V^\dagger
    \end{align}
Now, a pair of rotations in different bases becomes:
\begin{align}
 e^{i \phi_{1} Z_\psi}  e^{i \phi_{2} Z_\Pi} = V e^{i \phi_{1} Z_0}  V^\dagger e^{i \phi_{2} Z_\Pi}
\end{align}
Next, observe what happens when we write out $V$ in the $\{\ket{0},\ket{0^\perp}\}$ and $\{\ket{\Pi},\ket{\Pi^\perp}\}$ bases via some abuse of notation:
\begin{align}
        V &= a \ket{\Pi}\bra{0} + \bar a \ket{\Pi}\bra{0^\perp} + \bar a \ket{\Pi^\perp}\bra{0}  - a\ket{\Pi^\perp}\bra{0^\perp}   \\
        &=: \begin{bmatrix} a \ket{\Pi}\bra{0} & \bar a \ket{\Pi}\bra{0^\perp} \\ \bar a \ket{\Pi^\perp}\bra{0} & -a\ket{\Pi^\perp}\bra{0^\perp}   \end{bmatrix}
\end{align}
Although $V$ is not actually a reflection, when viewed in these particular bases the matrix elements look just like those of a reflection matrix. Our goal is to track the complicated motion through the multiple bases of the Grover subspace via just matrix multiplication over $\mathbb{C}^2$. In that case, the $\ket{\Pi}\bra{0}, \ket{\Pi}\bra{0^\perp},...$ labels disappear and we are actually just left with a reflection, call it $R$. Now we can think of $V e^{i \phi_{1} Z_0}  V^\dagger e^{i \phi_{2} Z_\Pi}$ just as $Re^{i\phi_1 Z}R e^{i\phi_2 Z}$.

Given access to alternating phase rotations and reflections, we invoke quantum signal processing to implement the desired polynomial. Throughout the paper, products run from left to right: $\prod_{j=1}^3 c_j = c_1c_2c_3$.

\begin{lemma} \label{lemma:circuits} \textbf{Implementing Pellian polynomials via phase rotations and reflections.} Say $P,Q \in \mathbb{C}[a]$ form a Pell pair, and say $d := \text{deg}(P)$. Then there exist phases $\phi_1, ..., \phi_{d-1}$, and another phase $\varphi$ such that:
    \begin{align}
        \prod_{j=1}^{d-1} ( R e^{i\phi_j Z} ) R = i^{-d}  \begin{bmatrix}e^{i\varphi} P & i Q \bar a \\ i Q^* \bar a & e^{-i\varphi} P^*  \end{bmatrix}  \text{ where } R := \begin{bmatrix} a & \bar a \\ \bar a & -a \end{bmatrix}.
\end{align}
\end{lemma}

\begin{proof} This follows from Theorem~3 and Corollary~8 of \cite{1806.01838} via\\ the substitution $e^{i \arccos(a) X} = i e^{-i\frac{\pi}{4}Z} R e^{-i \frac{\pi}{4}Z}$.
\end{proof}

\begin{proposition} \label{prop:pellian} Say $P$ is a Pellian polynomial. Then, there exist algorithms that:
    \begin{itemize}
        \item take any one of $\ket{\psi},\ket{\psi^\perp}, \ket{\Pi},\ket{\Pi^\perp}$ as input,
        \item apply some number of $e^{i\phi Z_\psi}$, $e^{i\phi Z_\Pi}$ gates,
        \item and finally measure in the $\ket{\psi},\ket{\psi^\perp}$ basis or the $\ket{\Pi},\ket{\Pi^\perp}$ basis. 
    \end{itemize}
    The $\ket{\psi}$ and $\ket{\Pi}$ outcomes occur with probability $|P(a)|^2$, and the $\ket{\psi^\perp}$ and $\ket{\Pi^\perp}$ outcomes occur with probability $1-|P(a)|^2$. The final measurement basis and the number of queries to $e^{i\phi Z_\psi}$, $e^{i\phi Z_\Pi}$ are determined by the degree of $P$ and the input state. The details are given by the transition diagrams in Figure~\ref{fig:transitions}. 
\end{proposition}

\begin{figure}[h]
     \centering
     \hfill
     \begin{subfigure}[b]{0.4\textwidth}
         \centering
         \includegraphics[width=0.7\textwidth]{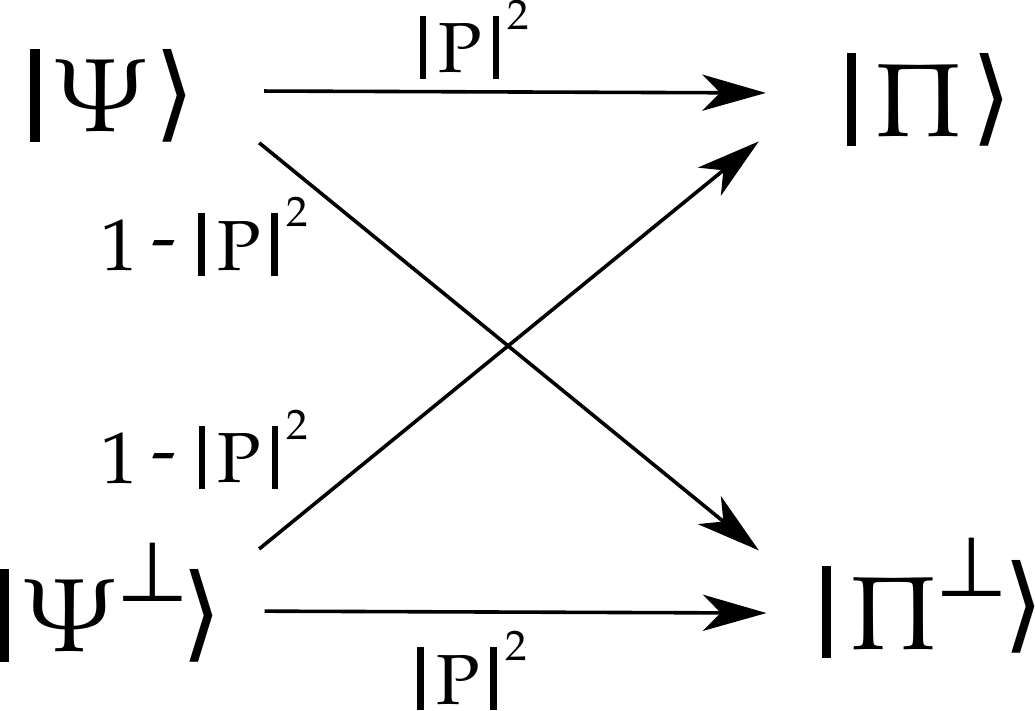}
         \caption{$\text{deg}(P) = 2k+1,\hspace{3mm}  Q_\psi = k,\hspace{3mm} Q_\Pi = k+1$}
         \label{fig:psi2pi}
     \end{subfigure}
     \hspace{5mm}
     \begin{subfigure}[b]{0.4\textwidth}
         \centering
         \includegraphics[width=0.7\textwidth]{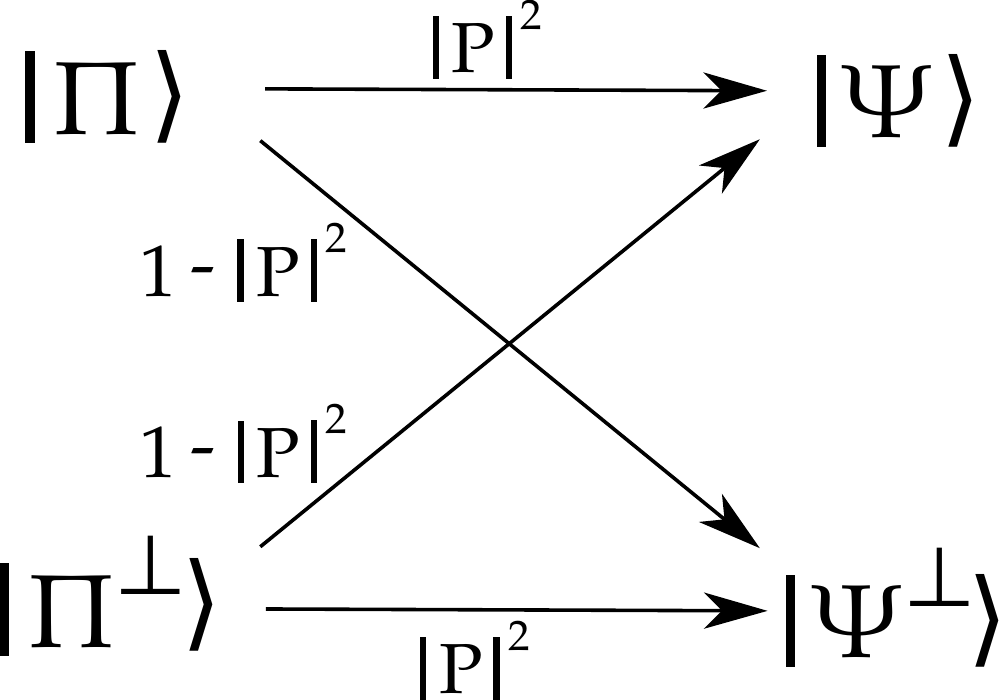}
         \caption{$\text{deg}(P) = 2k+1,\hspace{3mm}  Q_\psi = k+1,\hspace{3mm} Q_\Pi = k$}
         \label{fig:pi2psi}
     \end{subfigure}
     \hfill
     \vspace{10mm}
     \hfill
     \begin{subfigure}[b]{0.4\textwidth}
         \centering
         \includegraphics[width=0.7\textwidth]{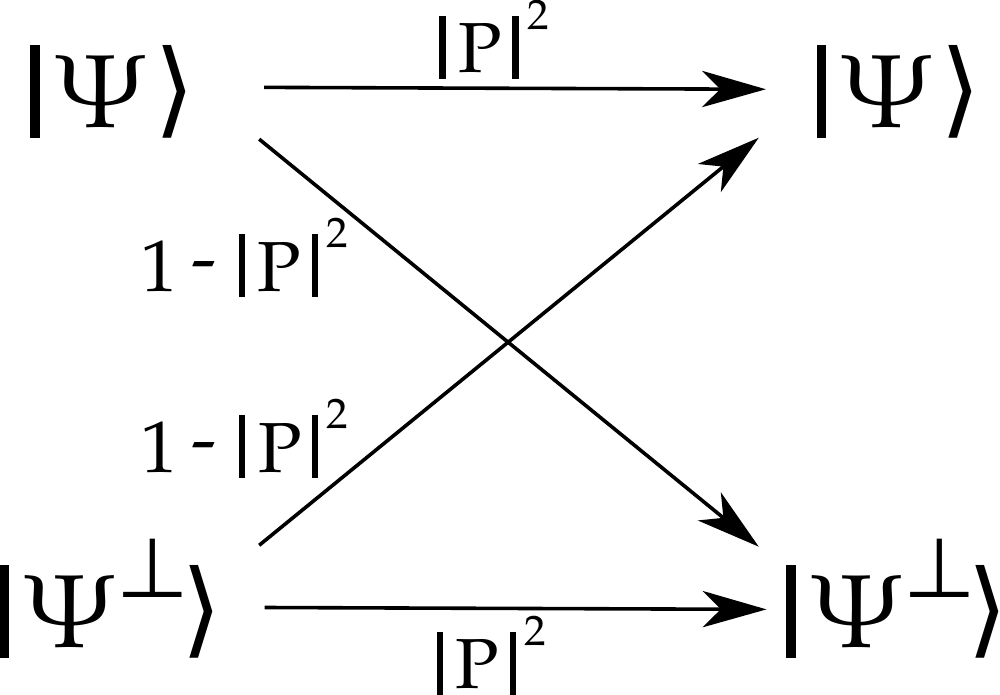}
         \caption{$\text{deg}(P) = 2k,\hspace{3mm}  Q_\psi = k,\hspace{3mm} Q_\Pi = k$}         \label{fig:psi2psi}
     \end{subfigure}
     \hspace{5mm}
     \begin{subfigure}[b]{0.4\textwidth}
         \centering
         \includegraphics[width=0.7\textwidth]{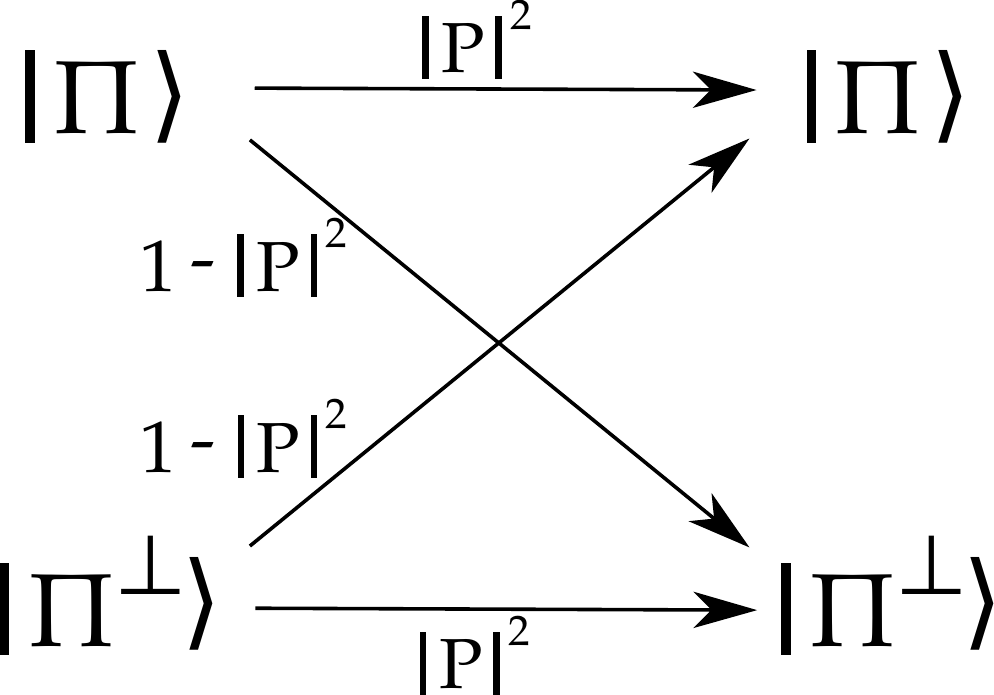}
         \caption{$\text{deg}(P) = 2k,\hspace{3mm}  Q_\psi = k,\hspace{3mm} Q_\Pi = k$}
         \label{fig:pi2pi}
     \end{subfigure}
     \hfill
     
        \caption{Transition diagrams implemented by the algorithms in Proposition~\ref{prop:pellian}. Given a Pellian polynomial $P$ and one of the states $\ket{\psi},\ket{\psi^\perp}, \ket{\Pi},\ket{\Pi^\perp}$, we can implement one of these transitions using some number of queries $Q_\psi, Q_\Pi$ to $e^{i\phi Z_\psi}, e^{i\phi Z_\Pi}$ respectively.}
        \label{fig:transitions}
\end{figure}

\begin{proof}
    The transitions in Figure~\ref{fig:transitions} are performed by applying one of the following unitaries, and then measuring in either the $\{\ket{\psi},\ket{\psi^\perp}\}$ basis or the $\{\ket{\Pi},\ket{\Pi^\perp}\}$ basis by using a Hadamard test on $Z_\psi$ or $Z_\Pi$. Say $P,Q$ form a Pell pair. Then, there exist phases $\phi_1,...,\phi_{\text{deg}(P)-1}$ such that:
\begin{align}
    U_{(a)} &:= \prod_{j=1}^{k} \left(  e^{i \phi_{2j-1} Z_\psi}  e^{i \phi_{2j} Z_\Pi}  \right)    = i^{-d}  \begin{bmatrix}e^{i\varphi} P \ket{\Pi}\bra{\psi} & i Q \bar a \ket{\Pi}\bra{\psi^\perp}  \\ i Q^* \bar a \ket{\Pi^\perp}\bra{\psi} & e^{-i\varphi} P^* \ket{\Pi^\perp}\bra{\psi^\perp}  \end{bmatrix}\\
    U_{(b)} &:= \prod_{j=1}^{k} \left(  e^{i \phi_{2j-1} Z_\Pi}  e^{i \phi_{2j} Z_\psi}  \right)    = i^{-d}  \begin{bmatrix}e^{i\varphi} P \ket{\psi}\bra{\Pi} & i Q \bar a \ket{\psi}\bra{\Pi^\perp}  \\ i Q^* \bar a \ket{\psi^\perp}\bra{\Pi} & e^{-i\varphi} P^* \ket{\psi^\perp}\bra{\Pi^\perp}  \end{bmatrix}\\
    U_{(c)} &:= \prod_{j=1}^{k-1} \left(  e^{i \phi_{2j-1} Z_\Pi}  e^{i \phi_{2j} Z_\psi}  \right) e^{i \phi_{2k-1} Z_\Pi}    = i^{-d}  \begin{bmatrix}e^{i\varphi} P \ket{\psi}\bra{\psi} & i Q \bar a \ket{\psi}\bra{\psi^\perp}  \\ i Q^* \bar a \ket{\psi^\perp}\bra{\psi} & e^{-i\varphi} P^* \ket{\psi^\perp}\bra{\psi^\perp}  \end{bmatrix}\\
    U_{(d)} &:= \prod_{j=1}^{k-1} \left(  e^{i \phi_{2j-1} Z_\psi}  e^{i \phi_{2j} Z_\Pi}  \right)    e^{i \phi_{2k-1} Z_\psi} = i^{-d}  \begin{bmatrix}e^{i\varphi} P \ket{\Pi}\bra{\Pi} & i Q \bar a \ket{\Pi}\bra{\Pi^\perp}  \\ i Q^* \bar a \ket{\Pi^\perp}\bra{\Pi} & e^{-i\varphi} P^* \ket{\Pi^\perp}\bra{\Pi^\perp}  \end{bmatrix}
\end{align}
    We will later show that these equalities hold -- suppose for now that they do. Then, the query complexities $Q_\psi$ and $Q_\Pi$ can just be read off from the products on the left hand sides plus the extra query for measurement, and the transition probabilities can be obtained by taking the magnitude squared of the matrix elements on the right hand sides. Recall that $|\bar a Q(a)|^2 = 1 - |P(a)|^2$ since $P,Q$ form a Pell pair. So all that remains to show is that these equalities are indeed true, that $U_{(a)}, U_{(b)}$ correspond to odd $P$, and that $U_{(c)},U_{(d)}$ correspond to even $P$.

    The calculation for all of these unitaries are extremely similar, so we will just prove the equality for $U_{(a)}$. Recall the analysis sketch above, where we considered another basis $\ket{0},\ket{0^\perp}$ for the Grover subspace, and let $V := \ket{\psi}\bra{0} + \ket{\psi^\perp}\bra{0^\perp}$. We then wrote $Z_0 := \ket{0}\bra{0} - \ket{0^\perp}\bra{0^\perp}$ and $e^{i \phi Z_\psi} = V e^{i \phi Z_0} V^\dagger$. Now we can rewrite the $U_{(a)}$ as: 
    \begin{align}
        U_{(a)}  := \prod_{j=1}^{k} \left(  e^{i \phi_{2j-1} Z_\psi}  e^{i \phi_{2j} Z_\Pi}  \right) =  \prod_{j=1}^{k} \left( V e^{i \phi_{2j-1} Z_0}  V^\dagger e^{i \phi_{2j} Z_\Pi}  \right) V \cdot V^\dagger \label{eqn:circuitrewrite}
    \end{align}
    Now we expand each of the components $V$, $e^{i \phi Z_0}$ and $e^{i\phi Z_\Pi}$ in the $\{\ket{0},\ket{0^\perp}\}$ and $\{\ket{\Pi},\ket{\Pi^\perp}\}$ bases:
    \begin{align}
        V &= \begin{bmatrix} a \ket{\Pi}\bra{0} & \bar a \ket{\Pi}\bra{0^\perp} \\ \bar a \ket{\Pi^\perp}\bra{0} & -a\ket{\Pi^\perp}\bra{0^\perp}   \end{bmatrix}\\[2mm]
            e^{i \phi Z_0} &= \begin{bmatrix} e^{i \phi} \ket{0}\bra{0} & \\ & e^{- i\phi} \ket{0^\perp}\bra{0^\perp}   \end{bmatrix}\\[2mm]
            e^{i \phi Z_\Pi} &= \begin{bmatrix} e^{i \phi} \ket{\Pi}\bra{\Pi} & \\ & e^{- i\phi} \ket{\Pi^\perp}\bra{\Pi^\perp}   \end{bmatrix}
    \end{align}
    We observe that the basis expansions of the components line up perfectly when plugged into (\ref{eqn:circuitrewrite}). That means that the unitary implemented by the algorithm can be obtained by matrix multiplication over $\mathbb{C}^2$:
    \begin{align}
        \prod_{j=1}^{k} \left( V e^{i \phi_{2j-1} Z_0}  V^\dagger e^{i \phi_{2j} Z_\Pi}  \right) V  &= \begin{bmatrix} \alpha \ket{\Pi}\bra{0} & \beta \ket{\Pi}\bra{0^\perp} \\ \gamma \ket{\Pi^\perp}\bra{0} &  \delta \ket{\Pi^\perp}\bra{0^\perp}  \end{bmatrix}\\
            \prod_{j=1}^{k} \left( R e^{i \phi_{2j-1} Z}  R^\dagger e^{i \phi_{2j} Z} \right) R    &=  \begin{bmatrix} \alpha & \beta \\ \gamma & \delta \end{bmatrix}
    \end{align}
    Since $R = R^\dagger$ this multiplication over $\mathbb{C}^2$ matches the expression in Lemma~\ref{lemma:circuits}. Recall $d:=\text{deg}(P) $.
    \begin{align}
        \prod_{j=1}^{k} \left( R e^{i \phi_{2j-1} Z}  R^\dagger e^{i \phi_{2j} Z} \right) R  &= \prod_{j=1}^{d} \left( R e^{i \phi_{j} Z}  \right) R= i^{-d}  \begin{bmatrix}e^{i\varphi} P & i Q \bar a \\ i Q^* \bar a & e^{-i\varphi} P^*  \end{bmatrix} \\
            \prod_{j=1}^{k} \left( V e^{i \phi_{2j-1} Z_0}  V^\dagger e^{i \phi_{2j} Z_\Pi}  \right) V  &=i^{-d}  \begin{bmatrix}e^{i\varphi} P \ket{\Pi}\bra{0} & i Q \bar a \ket{\Pi}\bra{0^\perp}  \\ i Q^* \bar a \ket{\Pi^\perp}\bra{0} & e^{-i\varphi} P^* \ket{\Pi^\perp}\bra{0^\perp}  \end{bmatrix}
    \end{align}
    So we have established:
    \begin{align}
        U_{(a)} :=   \prod_{j=1}^{k} \left(  e^{i \phi_{2j-1} Z_\psi}  e^{i \phi_{2j} Z_\Pi}  \right)    = i^{-d}  \begin{bmatrix}e^{i\varphi} P \ket{\Pi}\bra{\psi} & i Q \bar a \ket{\Pi}\bra{\psi^\perp}  \\ i Q^* \bar a \ket{\Pi^\perp}\bra{\psi} & e^{-i\varphi} P^* \ket{\Pi^\perp}\bra{\psi^\perp}  \end{bmatrix}
    \end{align}

    The calculations for $U_{(b)},U_{(c)},$ and $ U_{(d)}$ proceed similarly: plug in $e^{i \phi Z_\psi} = V e^{i \phi Z_0} V^\dagger$, and write out $V, e^{i\phi Z_0}$, and $e^{i \phi Z_\Pi}$ explicitly in the $\{\ket{0},\ket{0^\dagger}\}$ and $\{\ket{\Pi},\ket{\Pi^\perp}\}$ bases, while inserting $VV^\dagger$ as needed. All the bases will line up, so we can treat this as just regular multiplication of 2x2 matrices. Then invoke Lemma~\ref{lemma:circuits}.

    Observe also that the number of phases that appear in Lemma~\ref{lemma:circuits} is $\text{deg}(P)-1$. Since each phase corresponds to one call to either $e^{i\phi Z_\psi}$ or $e^{i\phi Z_\Pi}$, we see that $U_{(a)},U_{(b)}$ correspond to odd polynomials and $U_{(c)},U_{(d)}$ to even polynomials.

\end{proof}

\subsection{Sampling from Semi-Pellian Polynomials}

Some of the later sections in this paper require sampling access to a very broad family of polynomials. The polynomials are obtained from techniques that approximate arbitrary functions, and are then post-processed through shifting and scaling. To this end, Pellian polynomials are still rather restricted. The condition that $P(ix)P^*(ix) > 1$ for even $P$ is very difficult to ensure, essentially forcing us to use odd $P$ only. Furthermore, Pellian polynomials must also satisfy $|P(\pm 1)| = 1$, which is also not always desirable. Thus, we would like access to a more flexible family of polynomials.

Fortunately, there exists an implementable family of real polynomials that is not so restricted.

\begin{definition} A polynomial $P \in \mathbb{R}[x]$ is \textbf{semi-Pellian} if it is the real part of a Pellian polynomial.
\end{definition}

Indeed, \cite{1806.01838} not only characterize this family, but also show how to obtain the corresponding Pellian polynomials: Corollary~10 of \cite{1806.01838} states that a polynomial $P(x) \in \mathbb{R}[x]$ is semi-Pellian if and only if it has fixed parity, and $\forall x \in [-1,1]$ we have $|P(x)| \leq 1$.

To sample from semi-Pellian polynomials, we follow a similar technique to Corollary~18 of \cite{1806.01838}. Say $P$ is semi-Pellian, and say $\tilde P$ is a Pellian polynomial with $\text{Re}(\tilde P) = P$. We conditionally apply phase rotations $e^{i\phi Z}$ or $e^{-i\phi Z}$, which implements either $\tilde P$ or $\tilde P^*$. Then we use a linear combination of unitaries to average these together to obtain $P$. 

Linear combinations of unitaries require an additional flag qubit, and depending on the value of the flag we either obtain $\text{Re}(\tilde P)$ or $i\text{Im}(\tilde P)$. Thus, the bit $\mathbf{b}$ that we sample from depends on the joint outcome of the flag qubit and the output register. Independently of the flag qubit, the state found in the output register depends only on $|\tilde P|$ just like Proposition~\ref{prop:pellian}. The proof of Proposition~\ref{prop:semipellian} shows this calculation in detail.

To efficiently implement phase rotations with a conditional sign flip, we require oracles of the form:
\begin{align}
  e^{i \phi Z_\psi} \oplus e^{-i \phi Z_\psi} = \ket{0}\bra{0} \otimes e^{i \phi Z_\psi} +  \ket{1}\bra{1} \otimes e^{-i \phi Z_\psi}
\end{align}
and similarly for $\Pi$. Such an oracle can be implemented with one controlled and one regular query to $e^{i \phi Z_\psi}$. But, again, we claim that in most practical situations $e^{i \phi Z_\psi} \oplus e^{-i \phi Z_\psi}$ can be implemented with about the same query complexity as $Z_\psi$. As articulated above, most quantum algorithms take the following form: some computation flips a qubit depending on if a phase should be applied, applies $Z$ to that qubit, and then performs an uncomputation. The cost of the circuit is dominated by the computation and uncomputation. So, if the $Z$ is replaced with $e^{i \phi Z}$ or even $e^{i\phi Z}\oplus e^{-i\phi Z}$, the additional cost is negligible compared to the rest of the circuit.

\begin{proposition} \label{prop:semipellian} Say $P$ is a semi-Pellian polynomial, and say $\tilde P$ is a Pellian polynomial such that $P = \text{Re}(\tilde P)$. Then, there exist algorithms similar to the ones in Proposition~\ref{prop:pellian} that output a state $\ket{\text{out}} \in \{\ket{\psi},\ket{\psi^\perp}\}$ or  $\ket{\text{out}}\in \{\ket{\Pi},\ket{\Pi^\perp}\}$ as well as an additional state $\ket{\text{flag}} \in \{\ket{0},\ket{1}\}$. 
    
    Each of these algorithms have a target state $\ket{\text{target}}  \in \{\ket{\psi},\ket{\psi^\perp},\ket{\Pi},\ket{\Pi^\perp}\} $  so that:
    \begin{align}
        \text{Pr}\left[ \ket{\text{flag}}=\ket{0} \text{ and }\ket{\text{out}}= \ket{\text{target}}\right] &= \left|P(a)\right|^2 \\
        \text{Pr}\left[\ket{\text{out}}= \ket{\text{target}}\right] &= \left|\tilde P(a)\right|^2
    \end{align}

    The output basis, $\ket{\text{target}}$, as well as the query complexities to $e^{i\phi Z_\psi} \oplus e^{-i\phi Z_\psi}$ and $e^{i\phi Z_\Pi}\oplus e^{-i\phi Z_\Pi}$ are shown in the following table:

    \begin{center}
\def\arraystretch{1.2}
        \begin{tabular}{|c|c|c|c|c|c|}
            \hline $\text{deg}(P)$ & Input & Output Basis & $\ket{\text{target}}$ & $e^{i \phi Z_\psi} \oplus e^{-i \phi Z_\psi}$ &  $e^{i \phi Z_\Pi} \oplus e^{-i \phi Z_\Pi}$\\  \hline\hline
            2k+1 & $\ket{\psi}$ & $\ket{\Pi},\ket{\Pi^\perp}$ & $\ket{\Pi}$   & k+1 & k+2 \\ \hline
            2k+1 & $\ket{\psi^\perp}$ & $\ket{\Pi},\ket{\Pi^\perp}$ & $\ket{\Pi^\perp}$  & k+1 & k+2\\ \hline
            2k+1 & $\ket{\Pi}$ & $\ket{\psi},\ket{\psi^\perp}$ & $\ket{\psi}$  & k+2 & k+1 \\ \hline
            2k+1 & $\ket{\Pi^\perp}$ & $\ket{\psi},\ket{\psi^\perp}$ & $\ket{\psi^\perp}$  &  k+2 & k+1 \\ \hline
            2k & $\ket{\psi}$ & $\ket{\psi},\ket{\psi^\perp}$ & $\ket{\psi}$ &  k+3 & k\\ \hline
            2k & $\ket{\psi^\perp}$ & $\ket{\psi},\ket{\psi^\perp}$ & $\ket{\psi^\perp}$  &  k+3 & k\\ \hline
            2k & $\ket{\Pi}$ & $\ket{\Pi},\ket{\Pi^\perp}$ & $\ket{\Pi}$  &  k & k+3\\ \hline
            2k & $\ket{\Pi^\perp}$ & $\ket{\Pi},\ket{\Pi^\perp}$ & $\ket{\Pi^\perp}$ &  k & k+3\\ \hline
        \end{tabular}
    \end{center}
\end{proposition}

\begin{proof}
    Just as in Proposition~\ref{prop:pellian} there are actually only four algorithms since the algorithm only depends on the input basis and the parity of the degree. We only work though the case when $\text{deg}(P) = 2k+1$ and the input is $\ket{\psi}$ or $\ket{\psi^\perp}$, since the other cases follow the same pattern. Let $\phi_1,...,\phi_{2k}$ and $\varphi$ be the phases from Lemma~\ref{lemma:circuits} corresponding to $\tilde P$ and its Pell-complementary polynomial $\tilde Q$. Select $\phi_0 = \phi_{2k+1} = -\varphi/2$ to achieve:
    \begin{align}
        e^{i\phi_0Z} \prod_{j=1}^{2k+1} ( R e^{i \phi_j Z} ) = (-i)^{2k+1} \begin{bmatrix}  \tilde P & i \tilde Q \bar a \\ i \tilde Q \bar a & \tilde P  \end{bmatrix}
    \end{align}
    
    Using the same techniques from Proposition~\ref{prop:pellian} we show:
    \begin{align}
        i^{2k+1}  e^{i \phi_{0} Z_\Pi}  \prod_{j=1}^{k} (  e^{i \phi_{2j-1} Z_{\psi}}  e^{i \phi_{2j} Z_{\Pi}} )  e^{i \phi_{2k+1} Z_\psi}  &= \begin{bmatrix} \tilde P \ket{\Pi}\bra{\psi}  & i \tilde Q \bar a\ket{\Pi}\bra{\psi^\perp} \\ i \tilde Q^* \bar a\ket{\Pi^\perp}\bra{\psi} & \tilde P^*\ket{\Pi^\perp}\bra{\psi^\perp}  \end{bmatrix} =: A \label{eqn:defA}\\
            (-i)^{2k+1}  e^{-i \phi_{0} Z_\Pi}  \prod_{j=1}^{k} (  e^{-i \phi_{2j-1} Z_{\psi}}  e^{-i \phi_{2j} Z_{\Pi}} )  e^{-i \phi_{2k+1} Z_\psi}  &= \begin{bmatrix} \tilde P^* \ket{\Pi}\bra{\psi}  & -i \tilde Q^* \bar a\ket{\Pi}\bra{\psi^\perp} \\ -i \tilde Q \bar a\ket{\Pi^\perp}\bra{\psi} & \tilde P\ket{\Pi^\perp}\bra{\psi^\perp}  \end{bmatrix} =:B \label{eqn:defB}
    \end{align}

   Now we apply the trick from Corollary~18 of \cite{1806.01838}: we use a linear combinations of unitaries circuit to perform the average of two circuits $A,B$, so the top left matrix element becomes $P$. The circuit involves preparing an ancilla qubit in the $\ket{+}$ state, applying either $A$ or $B$ depending on the ancilla, and then applying a Hadamard gate to the ancilla:
    \begin{align}
        (H \otimes I)  (A \oplus B) (\ket{+} \otimes I) = \ket{0} \otimes \frac{A+B}{2} + \ket{1} \otimes \frac{A-B}{2}
    \end{align}
    This extra register on the left is $\ket{\text{flag}}$. We find: 
    \begin{align}
        \frac{A+B}{2} &=  \begin{bmatrix} \text{Re}(\tilde P) \ket{\Pi}\bra{\psi}  & - \text{Im}(\tilde Q) \bar a\ket{\Pi}\bra{\psi^\perp} \\ \text{Im}( \tilde Q) \bar a\ket{\Pi^\perp}\bra{\psi} & \text{Re}(\tilde P)\ket{\Pi^\perp}\bra{\psi^\perp}  \end{bmatrix}\\
            \frac{A-B}{2} &=  \begin{bmatrix} i\text{Im}(\tilde P) \ket{\Pi}\bra{\psi}  &  i\text{Re}(\tilde Q) \bar a\ket{\Pi}\bra{\psi^\perp} \\ i\text{Re}( \tilde Q) \bar a\ket{\Pi^\perp}\bra{\psi} &-i \text{Im}(\tilde P)\ket{\Pi^\perp}\bra{\psi^\perp} \end{bmatrix}
    \end{align}
    If we begin in the state $\ket{\psi}$, then the $\ket{0}\ket{\Pi}$ component of the output state has amplitude $\text{Re}(\tilde P) = P$, so we obtain this case with probability $|P|^2$. The $\ket{1}\ket{\Pi}$ component has amplitude $i\text{Im}(\tilde P)$, so the total probability of seeing $\ket{\text{out}} = \ket{\Pi}$ regardless of $\ket{\text{flag}}$ is $ |\text{Re}(\tilde P)|^2 + |i\text{Im}(\tilde P)|^2 = |\tilde P|^2 $. A similar pattern holds for the input state $\ket{\psi^\perp}$.

    All that remains to show is how to actually implement the unitary $A\oplus B$ and to count the total query complexity. Using the identity $CD \oplus EF = (C \oplus E)(D \oplus F)$ we see from (\ref{eqn:defA}, \ref{eqn:defB}) that $A \oplus B$ factors into a product of a $(i^{2k+1} \oplus (-i)^{2k+1})$ gate, and $k+1$ many $(e^{i \phi Z_\psi} \oplus e^{-i \phi Z_\psi})$ and $(e^{i \phi Z_\Pi} \oplus e^{-i \phi Z_\Pi})$ gates each. Finally, we require one more query to $(e^{i \phi Z_\Pi} \oplus e^{-i \phi Z_\Pi})$ to measure in the $\{\ket{\Pi},\ket{\Pi}^\perp\}$ basis.
\end{proof}

\subsection{Previous Results in this Framework}

Having established Propositions~\ref{prop:pellian} and~\ref{prop:semipellian}, we can state the full version of the theorem. 

    \begin{theorem}  \label{thm:polysamp} \textbf{Sampling from a polynomial.} Say $P(a)$ is a Pellian or a semi-Pellian polynomial. Then, there exists a quantum algorithm that samples from a random variable $\mathbf{b}_{P} \in \{0,1\}$ such that:
        \begin{align}
            \text{Pr}[\mathbf{b}_P = 1] = |P(a)|^2.
        \end{align}
        This algorithm takes a copy of $\ket{\psi},\ket{\psi^\perp},\ket{\Pi}$ or $\ket{\Pi^\perp}$ as input, and also outputs one of these four states at random. This algorithm makes $O(\text{deg}(P))$ oracle queries to $e^{i\phi Z_\psi}, e^{i \phi Z_\Pi}$ if $P$ is Pellian, and $(e^{i \phi Z_\psi} \oplus e^{-i \phi Z_\psi}),$\hspace{1mm} $(e^{i \phi Z_\Pi} \oplus e^{-i \phi Z_\Pi})$ if $P$ is semi-Pellian.
\end{theorem}

We call invoking this theorem on a particular polynomial $P$ to \textbf{sample from the polynomial} $P$. The algorithm is specifically designed to be used repeatedly. It accepts any of the four states $\ket{\psi},\ket{\psi^\perp},\ket{\Pi},\ket{\Pi^\perp}$ as input, and also only ever outputs one of these. So, regardless of the actual measurement outcomes and regardless of the sequence of polynomials  sampled from, we can always chain invocations of Theorem~\ref{thm:polysamp} into a very long circuit. 

This is only useful in situations where an input state $\ket{\psi}$ is expensive to prepare. In that case we might want to reuse that state as much as possible, and indeed Theorem~\ref{thm:polysamp} can in principle be applied arbitrarily many times given only one copy of $\ket{\psi}$. But we are by no means forced to do this: at any point, we can discard the output of the algorithm and re-prepare $\ket{\psi}$. Furthermore, the schedule of polynomials is completely independent of this choice.

Recall that Chebyshev polynomials of the first kind $T_n(a) = \cos( n \arccos(a))$ are Pellian, and can be sampled from by setting all the $\phi_j$ to $\pi/2$, in which case $e^{i \phi_j Z_\psi} = Z_\psi$ and $e^{i \phi_j Z_\Pi}= Z_\Pi$. In this case, $U_{(a)}$ from Proposition~\ref{prop:pellian} simply becomes Grover's algorithm: alternating applications of $Z_\Pi$ and $Z_\psi$, followed by a $\{\ket{\Pi},\ket{\Pi^\perp}\}$-basis measurement. So we see that several previous results in amplitude estimation can be viewed as invoking a special case of Theorem~\ref{thm:polysamp}. We will list these momentarily. 

To aid in comparing to previous results and also describe our new results in later sections, we set up some standard nomenclature about amplitude estimation algorithms. In particular, we define the random variables $\mathbf{\hat a}, \mathbf{Q}_\Pi, \mathbf{Q}_\psi, \mathbf{d}$ and $\mathbf{D}$. The query complexities $ \mathbf{Q}_\Pi, \mathbf{Q}_\psi,$ and degrees $\mathbf{d}, \mathbf{D}$ may be random since the algorithms are adaptive sometimes and can vary their runtime based on the results of measurement outcomes.

\begin{definition} \label{def:amp_est_alg}
    An \textbf{amplitude estimation algorithm} is an algorithm that, starting with $\ket{\psi}$, repeatedly samples from polynomials until it finally outputs an estimate of $a$ as well as one of $\ket{\psi},\ket{\psi^\perp},\ket{\Pi},\ket{\Pi^\perp}$. The estimate is denoted by $\mathbf{\hat{a}}$ and is a random variable over $\mathbb{R}$. The total number of queries to $e^{i \phi Z_\Pi}$, $e^{i \phi Z_\psi}$ or $(e^{i \phi Z_\psi} \oplus e^{-i \phi Z_\psi}),$\hspace{1mm} $(e^{i \phi Z_\Pi} \oplus e^{-i \phi Z_\Pi})$ are denoted by the random variables $\mathbf{Q}_\Pi, \mathbf{Q}_\psi \in \mathbb{Z}^+$ respectively. Finally, the highest degree of a polynomial ever sampled from is $\mathbf{d} \in \mathbb{Z}^+$ and the sum of the degrees of all the sampled polynomials is $\mathbf{D} \in \mathbb{Z}^+$.
\end{definition}

$\mathbf{D}$ can be regarded as the overall query complexity of the algorithm. If only one state $\ket{\psi}$ is used, then $\mathbf{D}$ is also the circuit depth. But if the state $\ket{\psi}$ is reset every time, then maximum circuit depth is only $\mathbf{d}$, even though the total query complexity is still $\mathbf{D}$.

\begin{proposition} Say an amplitude estimation algorithm only ever samples from Pellian polynomials. Then $ \mathbf{Q}_\Pi + \mathbf{Q}_\psi  = \mathbf{D}$. 
\end{proposition}

Now we list some previous results on amplitude estimation using the notation above.

\begin{description}
    \item[Relative-error estimation.] (\cite{1908.10846}, Theorem 3) For any $\eps,\delta \in [0,1/2]$ there exists an amplitude estimation algorithm such that:
    \begin{align}
        \text{Pr}[a(1-\eps) \leq \mathbf{\hat a} \leq a(1+\eps)] \geq 1-\delta
    \end{align}\\[-1.3cm]
    \begin{align}
        \exists f(a,\eps,\delta) \in O\left( a^{-1} \eps^{-1}\log(\delta^{-1}) \right) \text{ such that } \text{Pr}[ \mathbf{D} \leq f(a,\eps,\delta) ] \geq 1-\delta.
    \end{align} 
    \item[Fast additive-error estimation.] (\cite{1912.05559}, Theorem 1) For any $\eps,\delta \in [0,1/2]$ there exists an amplitude estimation algorithm such that:
    \begin{align}
        \text{Pr}[|\mathbf{\hat a}  - a| \leq \eps] &\geq 1- \delta
    \end{align}\\[-1.3cm]
    \begin{align}
        \text{Pr}\left[ \mathbf{Q_\Pi} \leq \frac{14}{\eps} \log\left( \frac{2}{\delta} \log_2 \left( \frac{\pi}{4\eps} \right) \right)  \right] &\geq 1-\delta
    \end{align}
    \item[Non-adaptive relative-error approximate counting.] (\cite{2010.04370}, Algorithm 1) Suppose $a \geq 1/N$ for some known $N \in \mathbb{Z}^+$. For any $\eps,\delta \in [0,1/2]$ there exists an amplitude estimation algorithm such that all polynomials depend exclusively on $\eps,\delta,N$ and are thus all known beforehand. It satisfies:
    \begin{align}
        \text{Pr}[a(1-\eps) \leq \mathbf{\hat a} \leq a(1+\eps)] \geq 1-\delta
    \end{align}\\[-1.3cm]
    \begin{align}
        \mathbf{D} \text{ is not random and } \mathbf{D} \in O\left(\sqrt{N/\eps}\hspace{1mm} \log(\delta^{-1})  \right)
    \end{align}
    \item[Hybrid quantum-classical estimation.] (\cite{2012.03348}, QoPrime Algorithm.) For any $\eps,\delta \in [0,1/2]$, and for any $0 < \beta  \leq 1$ there exists an amplitude estimation algorithm such that: 
    \begin{align}
        \text{Pr}[|\mathbf{\hat a}  - a| \geq \eps] &\leq \delta
    \end{align}\\[-1.3cm]
    \begin{align}
        \mathbf{D} \text{ is not random and } \mathbf{D} &\in O\left( \eps^{-(1+\beta)}  \hspace{1mm} \log(\delta^{-1})  \right)\\
        \mathbf{d} \text{ is not random and } \mathbf{d} &\in O\left( \eps^{-(1-\beta)} \right)
    \end{align}
\end{description}

We focused here on results that proved the correctness of their algorithms, which leaves out \cite{1904.10246}. All of the above algorithms, \cite{1904.10246} included, repeatedly invoke the special case of Theorem~\ref{thm:polysamp} where $P(a) = T_n(a)$. The research papers presenting these algorithms describe discarding and re-preparing the input state $\ket{\psi}$ every single time. But since Theorem~\ref{thm:polysamp} does not require this, these algorithms could all also work with just a single copy of $\ket{\psi}$. 

This means our framework has already extended the capabilities of previous works. Although, for the result of \cite{2012.03348} in particular, the whole point of separating $\mathbf{d}$ of $\mathbf{D}$ was to reduce the maximum circuit depth. If a single copy of $\ket{\psi}$ were re-used, the maximum depth would be $\mathbf{D}$, not $\mathbf{d}$. So this extension is not always sensible.

We also did not list \cite{0005055} and \cite{1907.09965} because these do not fit into the framework of Theorem~\ref{thm:polysamp}. Both of these algorithms involve phase estimation of the Grover operator, rather than simply repeated application. However, this complicated and expensive technique has largely been superseded by the faster, simpler, and more versatile algorithms listed above. There is only one capability that has not yet been `modernized': the non-destructive estimation algorithm from \cite{1907.09965}. This task is the subject of the next section.

\ifdefined\maindocument
\else
    
    \end{document}
\fi

\ifdefined\maindocument
    \section{State Repair \label{sec:repair}}
\else
    \documentclass[11pt]{article}
    \usepackage[margin=0.4in]{geometry}
    \usepackage[margin=0.4in]{geometry}
\usepackage[]{graphicx}
\usepackage{amssymb}
\usepackage{enumitem}
\usepackage{hyperref}
\usepackage{amsmath}
\usepackage{amsthm}
\usepackage{multicol}
\usepackage{qcircuit}
\usepackage{braket}
\usepackage{caption}
\usepackage{subcaption}
\usepackage{xcolor}
\usepackage{svg}
\usepackage[htt]{hyphenat}

\newcommand{\eps}{\varepsilon}
\newcommand{\erf}{\text{erf}}
\newcommand{\sgn}{\text{sgn}}
\newcommand{\rect}{\text{rect}}
\newcommand{\bigoh}[1]{\mathcal{O}(#1)}

\newtheorem{theorem}{Theorem}
\newtheorem*{theorem*}{Theorem}
\newtheorem{corollary}[theorem]{Corollary}
\newtheorem{definition}[theorem]{Definition}
\newtheorem{lemma}[theorem]{Lemma}
\newtheorem{remark}[theorem]{Remark}
\newtheorem{example}[theorem]{Example}
\newtheorem{proposition}[theorem]{Proposition}
\newtheorem{fact}[theorem]{Fact}
\newtheorem{empiricalclaim}[theorem]{Empirical Claim}
\newtheorem*{empiricalclaim*}{Empirical Claim}
    \begin{document}
    \begin{center}
        {\Large State Repair}
    \end{center}
\fi

    In the traditional setting for Grover's algorithm where we are searching for a marked item, the state $\ket{\psi}$ is just a uniform superposition over all items. In this case, $\ket{\psi}$ is trivial to prepare and to reflect about, at least compared to querying which items are marked. 

    However, there also exist settings where $\ket{\psi}$ is extremely expensive to prepare despite being relatively easy to reflect about. In this situation, one might hope for an amplitude estimation algorithm that could make do with a single copy of $\ket{\psi}$, and ideally also give that copy of $\ket{\psi}$ back once an estimate of $|\Pi\ket{\psi}|$ has been extracted. $\ket{\psi}$ could then be used for another part of the algorithm. We call this task \textbf{non-destructive amplitude estimation}.

    We know of two examples of such situations in the literature. One of them stems from the study of physical systems: say we have a Hamiltonian $H$ with a ground state $\ket{\psi_0}$. Our goal is to estimate the expectations of several observables $\bra{\psi_0}O_1\ket{\psi_0}, \bra{\psi_0}O_2\ket{\psi_0}$, etc. Suppose for the moment that the ground space is non-degenerate. In that case, following \cite{2002.12508}, we can design a polynomial $p(x)$ that is $\approx 1$ in the ground space and $\approx -1$ elsewhere, so that $p(H) \approx Z_{\psi_0}$. Assuming the spectral gap of $H$ is not too small, we can now efficiently reflect about $\ket{\psi_0}$. In order to prepare $\ket{\psi_0}$ however, we have to guess some trial state $\ket{\phi}$ that hopefully has large overlap $|\braket{\phi|\psi_0}|$ with $\ket{\psi_0}$, and use amplitude amplification to transform $\ket{\phi}\to\ket{\psi_0}$. Since ground state finding is QMA-complete, this step may take exponential time in general. So if it is really worth investing the time to prepare $\ket{\psi_0}$, we better only need to do so once. The previous section already describes an algorithm that requires only one copy of $\ket{\psi_0}$ to estimate a quantity, but a non-destructive algorithm could be used to extract several expectations from a single state\footnote{To be more explicit here, say we can prepare $\ket{\psi_0}$ using $P$ reflections, and we can repair $\ket{\psi_0}$ using $R$ reflections. If the trial state $\ket{\phi}$ has small overlap with $\ket{\psi_0}$ then $P \gg R$. A naive approach for estimating $M$ many observables would then cost $O(MP)$, while an approach based on state repair only costs $O(P+MR)$.}.
    
    Another example is the kind of quantum algorithm studied by \cite{1907.09965} and \cite{2009.11270}. These works consider a classical Hamiltonian $H(x)$, and are interested in preparing a qsample of its thermal distribution $\ket{\psi_\beta}$, where $\braket{x|\psi_\beta} \propto e^{-\beta H(x)/2}$. We are given reflection operators $ 2\ket{\psi_\beta}\bra{\psi_\beta} - I$ not just for the target temperature $\beta^*$ but also all intermediate temperatures. The idea is to find a sequence of temperatures $\beta_{0},\beta_{1},...\beta^*$ called a `Chebyshev cooling schedule' that guarantees $\braket{\psi_{\beta_i}|\psi_{\beta_{i+1}}}$ is never too small. With $\beta_0 = 0$, $\ket{\psi_{\beta_0}}$ is just the uniform superposition which is easy to prepare. Then we can use amplitude amplification to transform $\ket{\psi_{\beta_0}}\to\ket{\psi_{\beta_1}}\to...\to\ket{\psi_{\beta^*}}$. The remaining challenge is to identify the temperatures $\beta_1,\beta_2,...$ such that  $\braket{\psi_{\beta_i}|\psi_{\beta_{i+1}}}$  is never too small. If we have a copy of $\ket{\psi_{\beta_i}}$, we can identify $\beta_{i+1}$ via binary search, invoking amplitude estimation at each candidate temperature to estimate the overlap $\braket{\psi_{\beta_i}|\psi_{\beta_{i+1}}}$. This crucially requires amplitude estimation to be non-destructive, since the state $\ket{\psi_{\beta_i}}$ was expensive to prepare and needs to be re-used for several estimations and finally transformed into $\ket{\psi_{\beta_{i+1}}}$.

    To this end,  \cite{1907.09965} developed a non-destructive amplitude estimation algorithm based on the methods of \cite{0005055}. After running \cite{0005055} to obtain an estimate, we obtain an eigenstate of the Grover operator:
        \begin{align}
            \ket{\psi_{\pm}} := \frac{1}{\sqrt{2}}\left( \ket{\psi} \pm i \ket{\psi^\perp} \right)
        \end{align}
    Measuring this state in the $\{\ket{\psi},\ket{\psi^\perp}\}$ basis, we have `repaired' a copy $\ket{\psi}$ with probability 1/2. If we obtain $\ket{\psi^\perp}$, we can just repeat the procedure until we succeed.

    \begin{proposition} \textbf{State repair given bounds on $a$.} (\cite{1907.09965}, Theorem 18.) Say we are given any quantum state in the Grover subspace $\text{span}(\ket{\psi},\Pi\ket{\psi})$, and suppose $a := |\Pi\ket{\psi}|$ satisfies $\kappa < a < \sqrt{1-\kappa^2}$ for some known $\kappa > 0$. Then, for any $\delta >0$, there exists a quantum algorithm that, with success probability $\geq 1-\delta$, prepares either $\ket{\psi}$ or $\ket{\psi^\perp}$, each with probability 1/2. It makes $O( \kappa^{-1}  \log(\delta^{-1}) )$ queries to $Z_\psi$ and $Z_\Pi$.
\end{proposition}

    However, \cite{1907.09965}, featured a hidden assumption. We made this assumption explicit in the above: that some bounds $\kappa < a < \sqrt{1-\kappa^2}$ are known. $\kappa$ informs the precision of phase estimation on the Grover operator: if precision is not high enough, then we fail to distinguish the eigenvectors $\ket{\psi_\pm}$, so measuring the output fails to collapse the superposition over them. That means we do not have a copy of $\ket{\psi_{\pm}}$ after running \cite{0005055}, so the probability of obtaining $\ket{\psi}$ after measuring might be very small\footnote{In particular, looking at \cite{1907.09965}'s informal discussion on page 17, the states $F_M^{-1}\ket{S_M(\theta/\pi)}$ and $F_M^{-1}\ket{S_M(1-\theta/\pi)}$ are only orthogonal when $\theta$ is sufficiently far from 0 or $\pi$. }. In their analysis, \cite{1907.09965} implicitly focus on the case where $\eps < a < 1-\eps$, in which case the procedure always works. But if no bounds on $a$ are known, then it is not clear how to use their result to achieve non-destructive estimation.

In this section we present an amplitude estimation algorithm that overcomes this limitation. In Theorem~\ref{thm:repair} we show that any amplitude estimation algorithm with total degree $\mathbf{D}$ can be modified to also output a copy of the input state $\ket{\psi}$ at the end with probability $\geq 1-\delta$ for any $\delta > 0$. This new algorithm has total degree $O(  \delta^{-1/2}\hspace{1mm} \mathbf{D})$.

    Just like \cite{1907.09965}, the main idea is to add a `repair step' to the end of another amplitude estimation algorithm that transforms whatever state we have back into $\ket{\psi}$. However, the only piece of information that the repair step needs is the \emph{length} of the estimation algorithm that was just executed. We emphasize the flexibility of this result: our repair procedure can be appended to any amplitude estimation algorithm in the framework of Theorem~\ref{thm:polysamp}, which includes several algorithms from previous works as well as the other new algorithms we present in later sections of this work.

    We now outline the proof. Since $a$ is unknown, the main idea is to split the analysis into two cases. Let $\kappa$ be some threshold depending on $\delta$ and $\mathbf{D}$, and let $\bar \kappa := \sqrt{1-\kappa^2}$ as usual.
    \begin{itemize}
        \item Say $a \leq \kappa$ or $\bar \kappa \leq a$. Then the estimation algorithm is unlikely to have damaged $\ket{\psi}$ at all.
        \item Say $\kappa < a < \bar \kappa$. Then $\ket{\psi}$ can be repaired using $O(\kappa^{-1})$ rounds of amplitude amplification.
    \end{itemize}
    In other words, we have used the runtime $\mathbf{D}$ of the previous procedure to essentially obtain bounds on $a$. This is possible because we need $\mathbf{D}$ to be $\sim a^{-1}$ in order to perturb $\ket{\psi}$. We establish this fact first, and then we discuss amplitude amplification.

    Amplitude estimation algorithms in the framework of Theorem~\ref{thm:polysamp} may appear difficult to analyze in such generality because they could sample from so many polynomials $P$. But we can simplify things via two observations. First, consider sampling from a semi-Pellian polynomial $P$ that is the real part of a Pellian polynomial $\tilde P$. Looking at Proposition~\ref{prop:semipellian}, the actual transition probabilities between states are determined by $\tilde P$, not $P$. So we can assume without loss of generality that we only ever sample from Pellian polynomials. Second, it turns out that, among Pellian polynomials, Chebyshev polynomials essentially maximize the probability of damaging $\ket{\psi}$ near when $a \leq \kappa$ or $\bar \kappa\leq a$. That means that we can assume without loss of generality that the previous algorithm only ever sampled from Chebyshev polynomials.

    Our first goal is to use properties of Chebyshev polynomials to determine good choices for bounds on $a$, that is, to determine $\kappa$. To do so, we employ a geometrical argument involving small rotations on the Bloch sphere. Chebyshev polynomials are commonly interpreted in terms of rotation: For $|a| \leq 1$, it is a standard fact that $T_d(a) = \cos(d \arccos(a))$. However, for small $a$, $\arccos(a)$ is a large angle, not a small angle. We would prefer to consider rotations about $\theta := \arcsin(a)$, which is small whenever $a$ is small. Fortunately we also have the following:

    \begin{fact} \label{fact:sincheb} Say $|x| \leq 1$. If $d$ is odd then $|T_d(x)|^2 = |\sin(d \arcsin(x))|^2$. If $d$ is even, then $|T_d(x)|^2 = 1 - |\sin(d \arcsin(x))|^2$
\end{fact}

In the previous section, we expressed Pellian polynomials in terms of alternating reflections and rotations. However, it is also possible to express them in terms of alternating rotations alone. Rephrasing part of Theorem~3 of \cite{1806.01838}, we find that if $P,Q \in\mathbb{C}[a]$ is a Pell pair, then there exist phases $\phi_1,...,\phi_{d}$ such that:
\begin{align}
\bra{0} e^{i \phi_0 Z} \prod_{j=1}^{d} ( W  e^{i \phi_j Z}) \ket{0} = P \text{ where } W := \begin{bmatrix} a & i\bar a \\ i\bar a & a \end{bmatrix}. \label{eqn:rotpellian}
\end{align}

    Notice that $W = e^{iX\arccos(a)}$, but again we would like to write things in terms of $\theta := \arcsin(a)$. This is achieved by adding an extra factor of $iX$:
\begin{align}
W = \begin{bmatrix} \sin(\theta) & i\cos(\theta) \\ i\cos(\theta) & \sin(\theta) \end{bmatrix} = \begin{bmatrix} \cos\left(\frac{\pi}{2}-\theta\right) & i\cos\left(\frac{\pi}{2}-\theta\right) \\ i\cos\left(\frac{\pi}{2}-\theta\right) & \sin\left(\frac{\pi}{2}-\theta\right) \end{bmatrix} = e^{i \left(\frac{\pi}{2}-\theta\right) X} = i X e^{i\theta X}
\end{align}

Now we can express both general Pellian polynomials and Chebyshev polynomials in terms of small rotations $\theta$. This allows us to show that Chebyshev polynomials always bound Pellian polynomials near $a\approx 0$ or $a \approx 1$.

\begin{lemma} \label{lemma:chebyextreme} \textbf{Chebyshev polynomials extremize Pellian polynomials.} Say $P$ is a Pellian polynomial of degree $d$, and say $a \leq \sin\left(\frac{\pi}{2d} \right)$. If $d$ is odd, then $|P(a)| \leq |T_d(a)|$. If $d$ is even, then $|P(a)| \geq |T_d(a)|$.

Similarly, say $\cos\left(\frac{\pi}{2d}\right) \leq a$. Then $|P(a)| \geq |T_d(a)|$, regardless of if $d$ is even or odd.
\end{lemma}
\begin{proof} We begin by plugging $W = iX e^{i\theta X}$ into (\ref{eqn:rotpellian}). Observe that $e^{i\phi Z}X = Xe^{-i\phi Z}$, so we can find new $\phi_j'$ such that:
    \begin{align}
P = \bra{0} e^{i \phi_0 Z} \prod_{j=1}^{d} (  i X e^{i \theta X} e^{i \phi_j Z} ) \ket{0} = \bra{0} (iX)^d e^{i \phi_0' Z} \prod_{j=1}^{d} ( e^{i\theta X} e^{i \phi_j' Z}  ) \ket{0}
\end{align}
   Next, we repeatedly insert $e^{i \phi Z} e^{- i \phi Z}$ for various values of $\phi$, and absorb these rotations into the $\phi_j'$ to make $\phi_j''$ such that:
\begin{align}
    P = \bra{0} (iX)^d e^{i  \phi_0'' Z} \prod_{j=1}^{d} ( e^{i \phi_j'' Z} e^{i\theta X} e^{-i  \phi_j'' Z}  ) \ket{0} \label{eqn:blochrots}
\end{align}

Now we make a geometric argument involving the Bloch sphere to bound $|P|^2$.  We begin with the $a \leq \sin\left(\frac{\pi}{2d}\right)$ case. Our goal is to show that if $d$ is odd, then $|P(a)| \leq |T_d(a)|$. If $d$ is even, then $|P(a)| \geq |T_d(a)|$.

Consider the dynamics of (\ref{eqn:blochrots}) on the Bloch sphere, as depicted in Figure~\ref{fig:bloch}. We begin in the $\ket{0}$ state, and repeatedly make rotations $e^{i \theta X_{\phi}}$ around axes $X_{\phi} := e^{i \phi Z} X e^{-i \phi Z} $ that are confined to the XY-plane. Finally, $|P|^2$ is the probability of measuring $\ket{0}$ in the even case, and $\ket{1}$ in the odd case.
 
    Say we are in the odd case, where we want to give an upper bound on $|P|^2$. This is the probability of measuring $\ket{1}$, which corresponds to how close to the south pole of the Bloch sphere we are.  Each rotation $e^{i \theta X_{\phi}}$ moves us an arclength of $\theta$, satisfying $\theta \leq \frac{\pi}{2d}$. While these could cancel each other out depending on the $\phi_j''$, they achieve the greatest distance when they all rotate in the same direction, in which case they achieve a total arclength of $d\theta$. Since $d\theta \leq \frac{\pi}{2}$, we can never overshoot the south pole $\ket{1}$ if we do this. So the final state's amplitude on $\ket{1}$ is at most $\sin(d \theta)$, so $|P|^2 \leq \sin^2(d\theta) = |T_d(a)|^2$. 
    
In the even case we want to give a lower bound on $|P|^2$. This is the probability of measuring $\ket{0}$, which corresponds to how close to the north pole we are. But a lower bound on $|P|^2$ can be obtained from an upper bound on $1 - |P|^2$, which measures how close to the south pole we are. But this is exactly the quantity we bounded in the previous paragraph: we have $1 - |P|^2 \leq \sin^2(d \theta)$. With some trigonometric identities this yields $|P|^2 \geq |T_d(a)|^2$.
 
 The $\cos\left( \frac{\pi}{2d}\right)\leq a$ case is very similar. We first notice that $\bar a \leq \sin\left( \frac{\pi}{2d} \right)$, so we define $\bar\theta := \arcsin(\bar a) = \arccos(a)$, and observe that $W = e^{i \bar \theta X}$. We repeat the calculation without needing to propagate forward an $iX$, and arrive at the existence of some $\bar \phi_j''$ such that:
  \begin{align}
P = \bra{0} e^{i \bar \phi_0'' Z} \prod_{j=1}^{d} ( e^{i \bar \phi_j'' Z} e^{i\bar \theta X} e^{-i \bar \phi_j'' Z}  ) \ket{0}
\end{align}
 So, regardless of the parity of $d$, we arrive in a similar situation: we start at the north pole of the Bloch sphere with $\ket{0}$, apply $d$ small rotations of angle $\bar\theta \leq \frac{\pi}{2d}$ around various axes $X_{\bar\phi_k''}$. Afterwards, $|P|^2$ is the probability that we measure $\ket{0}$, or equivalently $1- |P|^2$ is the probability that we measure $\ket{1}$. As we argued before we have $1-|P|^2 \leq \sin^2(d\bar \theta)$, so we have $|P|^2 \geq |T_n(a)| $.
 
\end{proof}

\begin{figure}[h]
     \centering
         \includegraphics[width=0.4\textwidth]{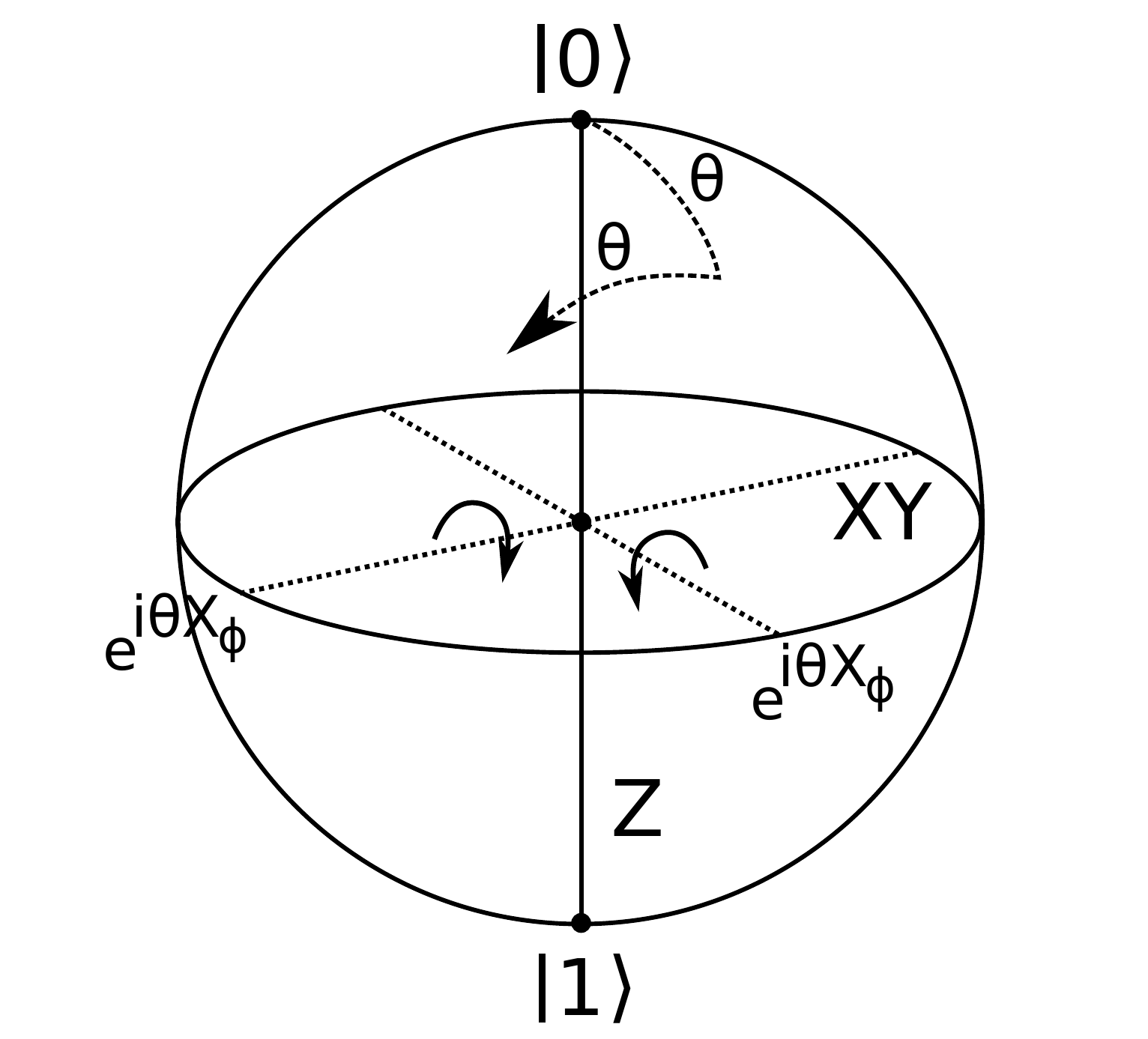}
        \caption{Visualization of the proof of Lemma~\ref{lemma:chebyextreme}. In particular, consider (\ref{eqn:blochrots}). The quantum state starts at $\ket{0}$, and is then acted on by several rotations $e^{i \theta X_{\phi}}$ around axes $X_{\phi} := e^{i \phi Z} X e^{-i \phi Z} $. The arclength is always $\theta$, and the axes are confined to the XY-plane. We see that the furthest angle we could achieve after $d$ many such rotations is $d\theta$. }
        \label{fig:bloch}
\end{figure}

We have demonstrated that for sufficiently extreme $a$ (that is, $a \leq \kappa$ or $\bar \kappa \leq a$), we can consider Chebyshev polynomials without loss of generality. We want to argue that for extreme $a$, the probability that we damage $\ket{\psi}$ is very small. However, recalling Figure~\ref{fig:transitions}, sampling from an odd polynomial essentially forces us to obtain a copy of $\ket{\Pi}$ or $\ket{\Pi^\perp}$, therefore unavoidably damaging $\ket{\psi}$. We argue that this is fine, since for extreme $a$, one of these is always very close to $\ket{\psi}$ and the other is very far from $\ket{\psi}$. In particular, if $a \leq \kappa$, then as we sample from odd polynomials we will just bounce back and forth between $\ket{\psi}$ and $\ket{\Pi^\perp}$ with high probability. Figure~\ref{fig:nearbystates} summarizes the two cases.

\begin{figure}[h]
     \centering
     \begin{subfigure}[b]{0.3\textwidth}
         \centering
         \includegraphics[width=0.9\textwidth]{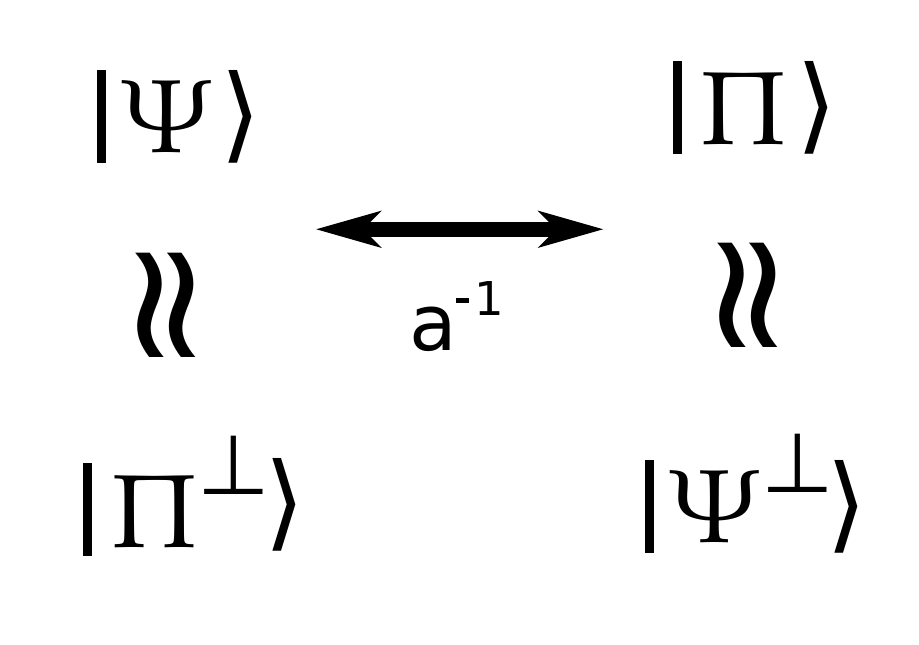}
         \caption{$a$ is small: $a \leq \kappa$}
     \end{subfigure}
     \hspace{5mm}
     \begin{subfigure}[b]{0.3\textwidth}
         \centering
         \includegraphics[width=0.9\textwidth]{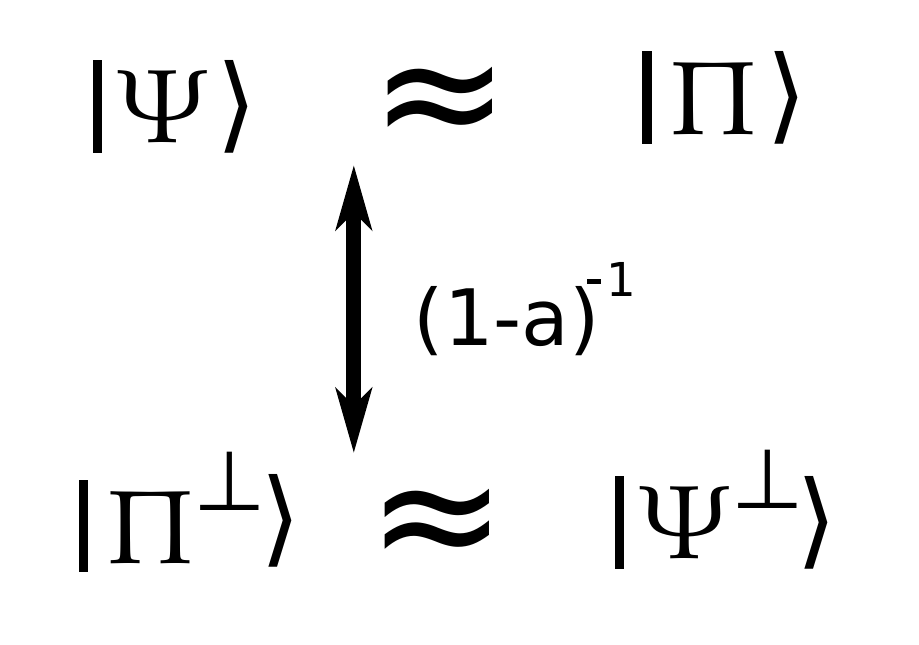}
         \caption{$a$ is large: $\bar \kappa \leq a$}
     \end{subfigure}
     \hspace{5mm}
     \begin{subfigure}[b]{0.3\textwidth}
         \centering
         \includegraphics[width=0.8\textwidth]{figures/pi2psi.pdf}
        \vspace{3mm}
         \caption{Reprint of Figure~\ref{fig:pi2psi} \label{fig:pi2psi_again}}
        \vspace{-3mm}
     \end{subfigure}
     \hfill
        \caption{For extreme $a$, the states $\ket{\psi},\ket{\psi^\perp}, \ket{\Pi}, \ket{\Pi^\perp}$ group into two pairs. The states in each pair are very close together, and the two pairs are very far apart. That means that if we sample from an odd polynomial, then unless that polynomial has very high degree we are going to just bounce back and forth between $\{\ket{\psi},\ket{\Pi^\perp}\}$ for $a\leq \kappa$, or between $\{\ket{\psi},\ket{\Pi}\}$ when $a \geq \bar \kappa$.}
        \label{fig:nearbystates}
\end{figure}

    \begin{proposition} \label{prop:extreme}  \textbf{For extreme $a$, we are unlikely to damage $\ket{\psi}$}. Say we just ran an amplitude estimation algorithm such that the sum of the degrees of all the sampled polynomials is $\mathbf{D}$. Then, the following holds for any $\delta > 0$:

\begin{itemize}
    \item If $a \leq \sin(\sqrt{\delta}/\mathbf{D})$, then the algorithm returns a copy of $\ket{\psi}$ or $\ket{\Pi^\perp}$ with probability $\geq 1-\delta$.
    \item If $a \geq \cos(\sqrt{\delta}/\mathbf{D}) $, then the algorithm returns a copy of $\ket{\psi}$ or $\ket{\Pi}$ with probability $\geq 1-\delta$.
\end{itemize}

\end{proposition}
\begin{proof} For the purposes of transitions between states $\{\ket{\psi},\ket{\psi^\perp},\ket{\Pi},\ket{\Pi^\perp}\}$ we can without loss of generality assume that the algorithm sampled from Pellian polynomials only. This is because sampling from a semi-Pellian polynomial $P$ entails finding a Pellian polynomial $\tilde P$ such that $\text{Re}(\tilde P) = P$, and the transitions between states are determined by $\tilde P$ in exactly the same way (see Proposition~\ref{prop:semipellian}).

While in general $\mathbf{D}$ is a random variable, the randomness does not really matter for this proof so we drop the boldface and just write $D$.  We begin with the $a \leq \sin(\sqrt{\delta}/D)$ case. Looking at Figure~\ref{fig:transitions}, we will show that with probability $\geq 1-\delta$, we only see the transitions $\ket{\psi}\to\ket{\Pi^\perp}, \hspace{1mm} \ket{\Pi^\perp}\to\ket{\psi},\hspace{1mm} \ket{\psi}\to\ket{\psi}$ or $\ket{\Pi^\perp}\to\ket{\Pi^\perp}$. That way, we can only finish with $\ket{\psi}$ or $\ket{\Pi^\perp}$.

    Say we sampled from $m$ polynomials total, and the $j$'th polynomial had degree $d_j$. That way $\sum_{j=1}^m d_j =D$. If the $j$'th polynomial $P_j$ is odd, we must show that the $\ket{\psi}\to\ket{\Pi}$ and  $\ket{\Pi^\perp}\to\ket{\psi^\perp}$ transitions are unlikely. These each occur with probability $|P_j|^2$. Since $\arcsin(a) \leq \sqrt{\delta}/D$, we can invoke Lemma~\ref{lemma:chebyextreme} to obtain $|P_j(a)| \leq |T_{d_j}(a)|$:
\begin{align}
 \text{(odd $d_j$) } \to \hspace{1mm} |P_j(a)|^2 \leq |T_{d_j}(a)|^2 \leq |\sin(d_j \arcsin(a)  )   |^2 \leq | \sin(\sqrt{\delta} \cdot d_j/ D) |^2 \leq \delta \cdot d_j^2 / D^2
\end{align}
If the $j$'th polynomial $P_j$ has even degree $d_j$, then we must show that the $\ket{\psi}\to\ket{\psi^\perp}$ and  $\ket{\Pi^\perp}\to\ket{\Pi}$ transitions are unlikely. These occur with probability $1-|P_j|^2$. Lemma~\ref{lemma:chebyextreme} gives us $|P_j(a)| \geq |T_{d_j}(a)|$.
\begin{align}
 \text{(even $d_j$) } \to \hspace{1mm} 1-|P_j(a)|^2 \leq 1 - |T_{d_j}(a)|^2 \leq  |\sin(d_j \arcsin(a)  )   |^2  \leq \delta \cdot d_j^2 / D^2
\end{align}
Finally, we observe that $\sum_j d_j^2 \leq \sum_{j,k} d_j d_k = D^2 $ to complete the proof with a union bound:
\begin{align}
\text{Pr}[\text{ever create }\ket{\psi^\perp},\ket{\Pi^\perp}] \leq \sum_{j=1}^m \delta \cdot d_j^2 / D^2 \leq \delta.
\end{align}

    Now we consider $a \geq \cos(\sqrt{\delta}/D) $, which implies $\arccos(a) \leq \sqrt{\delta}/D$. In this case we want to demonstrate that we only ever see the transitions $\ket{\psi}\to\ket{\Pi}, \hspace{1mm} \ket{\Pi}\to\ket{\psi},\hspace{1mm} \ket{\psi}\to\ket{\psi}$ or $\ket{\Pi}\to\ket{\Pi}$, so we finish only with $\ket{\psi}$ or with $\ket{\Pi}$. Since $a \geq \cos( \sqrt{\delta}/D )$, Lemma~\ref{lemma:chebyextreme} gives us $|P_j(a)| \geq |T_{d_j}(a)|$ regardless of $d_j$.

If $d_j$ is odd, then we want to show that the $\ket{\psi}\to\ket{\Pi^\perp}$ and $\ket{\Pi}\to\ket{\psi^\perp}$ transitions are unlikely. These each occur with probability $1-|P_j|^2$, so:
    \begin{align}
        \text{(odd $d_j$)} \to\hspace{1mm}    1-|P_j(a)|^2 &\leq 1 - |T_{d_j}(a)|^2 = 1 - |\cos( d_j \arccos(a))|^2 \\
        &= |\sin( d_j \arccos(a) )|^2 \leq |\sin( \sqrt{\delta} \cdot d_j/D   )|^2 \leq \delta \cdot d_j^2/D
    \end{align}

    Similarly, if $d_j$ is even then we want the $\ket{\psi}\to \ket{\psi^\perp}$ and $\ket{\Pi}\to \ket{\Pi^\perp}$ transitions to be unlikely. These occur with probability $1 - |P_j|^2$, so:
    \begin{align}
 \text{(even $d_j$) } \to \hspace{1mm} 1-|P_j(a)|^2 \leq 1 - |T_{d_j}(a)|^2 \leq  |\sin(d_j \arccos(a)  )   |^2  \leq \delta \cdot d_j^2 / D^2.
    \end{align}
    The same union bound argument shows that we never create $\ket{\psi^\perp}, \ket{\Pi}$ with probability $\geq 1- \delta$.
\end{proof}

This completes the first part of the argument: if $a$ is sufficiently close to $0$ or $1$ then it is probably not necessary to repair the state at all, since we either already have $\ket{\psi}$ or something very close to it. Next, we simply turn this argument on its head: say we do not have $\ket{\psi}$ or something close to it. Then $a$ is probably far from $0$ or $1$, meaning it is easy to repair using amplitude amplification.

We just ran an amplitude estimation algorithm, so we probably have a decent estimate of $a$. This means we could just use that estimate to inform a very standard amplitude estimation protocol based on Grover rotations only. However, part of the strength of this section's state repair algorithm is that it can be appended to \emph{any} amplitude estimation algorithm, including those that give pretty weak guarantees on the quality of the final estimation. So, we will stick to our polynomial sampling framework and construct an algorithm for fixed-point amplitude estimation within it.  We can construct Pellian polynomials that, when sampled from, will essentially force the desired transitions in Figure~\ref{fig:pi2psi}, reprinted for convenience in Figure~\ref{fig:pi2psi_again}. It turns out that this method is extremely similar to \cite{1409.3305}.

To construct the desired polynomials, we require a characterization of Pell-complementary polynomials from Theorem~3 of \cite{1806.01838} we know that a polynomial $Q \in \mathbb{C}[x]$ is Pell-complementary if and only if:
\begin{itemize}
    \item $\forall x \in [-1,1]$, we have $\bar x|Q(x)| \leq 1$
    \item If $Q$ is even, then $\forall x\in \mathbb{R}$ we have $ (1+ x^2)  Q(ix)Q^*(ix) \geq 1$.
\end{itemize}

Next, we will leverage this fact to construct the two polynomials we need. We will need two: we sample from $J(a)$ when $a \leq \bar \kappa$, and we will sample from $K(a)$ when $\kappa \leq a$. The polynomials are plotted in Figure~\ref{fig:amppolys}.  The construction of $J(a)$ already appears in the literature due to its applications of fixed-point Grover search, but the construction of $K(a)$ is novel as far as we know.

\begin{figure}[h]
     \centering
        \includegraphics[width=0.8\textwidth]{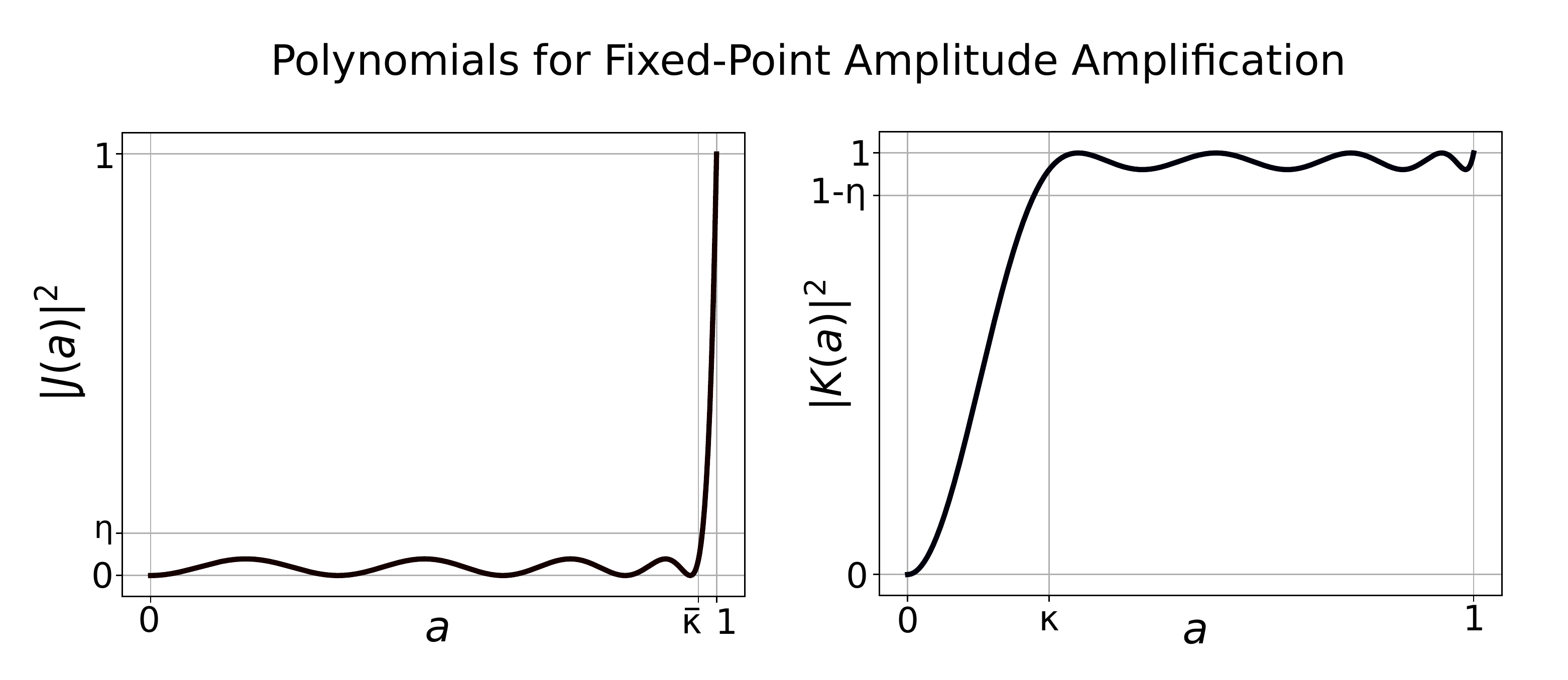}
        \caption{Polynomials constructed in Lemma~\ref{lemma:repairpolys} and used in Proposition~\ref{prop:repair}. Here we selected $\eta = 0.1$ and $\kappa = 0.25$, which makes these polynomials have degree $l = 9$.}
        \label{fig:amppolys}
\end{figure}

\begin{lemma} \label{lemma:repairpolys} \textbf{Polynomials for fixed-point amplitude amplification.} Take any $\kappa,\eta \in(0,1)$, and let $\bar\kappa := \sqrt{1-\kappa^2}$. There exists an odd Pellian polynomial $J(x)$ satisfying $|J(x)|^2 < \eta$ for $|x| < \bar\kappa$. There also exists an odd Pellian polynomial $K(x)$ satisfying $|K(x)|^2 > 1-\eta$ for $|x| > \kappa$. The degree of these polynomials is the smallest odd number $\geq \kappa^{-1} \ln\left( 2/\sqrt{\eta}  \right)$. \end{lemma}

\begin{proof} For the construction of $J(x)$ we follow Lemma~4.1 in \cite{1301.1162}, but note that \cite{1409.3305} do something similar. For some odd $l \in \mathbb{Z}^+$ and some $\gamma \in [0,1]$, let:
    \begin{align}
        J(x) := \frac{T_l(x/\gamma)}{T_l(1/\gamma)}
    \end{align}
    First, we show that for any such $l,\gamma$, there is a polynomial $Q$ such that $J,Q$ form a Pell pair because $J$ satisfies the first set of constraints of Theorem~3 of \cite{1806.01838}.  Chebyshev polynomials $T_l(x)$ are already Pellian. Since $1/\gamma \geq 1$, we have $1/|T_l(1/\gamma)| \leq 1$. Then, for $x/\gamma \leq 1$ we have $|T_l(x/\gamma)| \leq 1$, so $|J(x)| \leq 1$. As for $1 \leq x/\gamma$, we utilize the fact that $|T_l(x/\gamma)|$ is increasing. That means for $x \leq 1$ we have $|T_l(x/\gamma)| \leq |T_l(1/\gamma)|$, so we also have   $|J(x)| \leq 1$. But for $x \geq 1$ we have  $|T_l(x/\gamma)| \geq |T_l(1/\gamma)|$, so $|J(x)| \geq 1$.

    \newcommand{\acosh}{\text{arccosh}}

    Second, we show how to actually select $l,\gamma$ such that $|J(x)|^2 \leq \eta$ for $|x| \leq \bar\kappa$. Certainly for $|x| \leq \gamma$ we have $|T_l(x/\gamma)|\leq 1$, so therefore $|J(x)| \leq 1/|T_l(1/\gamma)|$. We select $\gamma := \bar\kappa$ and pick the minimum odd $l$ such that $1/|T_l(\gamma^{-1})| \leq \sqrt{\eta}$. We have $T_l(\gamma^{-1}) = \cosh(l \acosh(\gamma^{-1})) \geq \frac{1}{2}\exp(l \acosh(\gamma^{-1}))$. Also, we have:
    \begin{align}
        \acosh(\gamma^{-1}) \geq 2\tanh( \acosh(\gamma^{-1})/2 ) \geq 2\sqrt{ \frac{\gamma^{-1}-1}{\gamma^{-1}+1}} 
    \end{align}
    We desire $1/|T_l(\gamma^{-1})| \leq \sqrt{\eta}$, or rather $|T_l(\gamma^{-1})| \geq 1/\sqrt{\eta}$, which is ensured when:
    \begin{align}
        \frac{1}{2} \exp\left( 2l \sqrt{\frac{\gamma^{-1}-1}{\gamma^{-1}+1}}  \right)  \geq \frac{1}{\sqrt{\eta}} 
    \end{align}
    With $ \gamma := \bar\kappa $, we satisfy the above if $l$ satisfies:
    \begin{align}
        l   \geq \frac{1}{2} \ln\left(\frac{2}{\sqrt{\eta}}\right) \sqrt{ \frac{ \bar\kappa^{-1} + 1 }{ \bar\kappa^{-1}  - 1}  } =  \frac{1}{2} \ln\left(\frac{2}{\sqrt\eta}\right) \sqrt{ \frac{ (1 + \bar\kappa)^2 }{ 1  - \bar\kappa^2}  } 
    \end{align}
    which in turn is implied by $ l   \geq  \frac{1 }{\kappa}\ln\left(\frac{2}{\sqrt{\eta}}\right) $, seeing as $\bar\kappa < 1$.

    We move on to the construction of the Pellian polynomial $K(x)$ satisfying $|K(x)|^2 \geq 1-\eta$ for $|x| \geq \kappa$, using $J(x)$ as a starting point. Indeed, the desired properties follow from:
    \begin{align}
        |K(x)|^2 = 1 - \left|J\left(\bar x\right)\right|^2
    \end{align}
    We obtain $K$ from the Pell-complementarity characterization of Theorem~3 from  \cite{1806.01838}: we show that there is an even polynomial $E(x)$ such that $J(\bar x) = \bar x E(x)$, and that $E(x)$ satisfies the conditions in the lemma.

    $J(\bar x)$ is an odd polynomial in $\bar x$, so $J(\bar x)/\bar x$ is an even polynomial in $\bar x$. Even polynomials in $\bar x$ are really polynomials in $\bar x^2 = 1 - x^2$, so $E(x) := J(\bar x)/\bar x$ is an even polynomial in $x$ as well. By construction $J(\bar x) = \bar x E(x)$. 

    The first property $\bar x |E(x)| \leq 1$ for $|x| \leq 1$ is guaranteed by the fact that $|J(\bar x)| \leq 1$  for $|x| \leq 1$. Since $E(x)$ is even, we also need to show the second property that $(1+x^2) E(ix) E^*(ix) \geq 1$ for all real $x$. Since the coefficients of $J(x)$ are real, the coefficients of $E(x)$ are also real, so $E^*(ix) = E(ix) $. So it is sufficient to show $(1+x^2) E(ix)^2 \geq 1$. 
    \begin{align}
        (1+ x^2)  E(ix)^2 =   J( \sqrt{1 - (ix)^2}   )^2  \frac{1 + x^2}{1 - (ix)^2}   = J( \sqrt{1 + x^2}   ) \geq 1
    \end{align}
    Since $\sqrt{1+x^2} \geq 1$ for all real $x$, and $J(x) \geq 1$ for all positive $x$, we have shown the second condition for $E$. Thus, by Theorem~3 of  \cite{1806.01838} $E$ is Pell-complementary, so there exists an odd polynomial $K(x)$ such that $K$ and $E$ form a Pell pair.  We have $|K(x)|^2 = 1 - \bar x^2 |E(x)|^2 = 1 - |J(\bar x)|^2$, so $|K(x)|^2 \geq 1-\eta$ for $|x| \geq \kappa$ as desired.
\end{proof}

These polynomials immediately yield an algorithm for amplitude amplification. At this point we finally identify $\kappa := (4/5) \cdot \sqrt{\delta}/D$ in order to ensure $\kappa \leq \sin(\sqrt{\delta}/D)$ and $\cos(\sqrt{\delta}/D) \leq \bar \kappa $.

\begin{proposition} \label{prop:repair}  \textbf{Obtaining $\ket{\psi}$ from $\ket{\Pi}$ or $\ket{\Pi^\perp}$.} Suppose $a < \cos(\sqrt{\delta}/D)$, and we have a copy of $\ket{\Pi^\perp}$. Then, for any $\eta > 0$ we can prepare a copy of $\ket{\psi}$ with success probability $\geq 1-\eta$ by sampling from a polynomial of degree $\lceil (5/4) (D/\sqrt{\delta}) \ln(2/\sqrt{\eta})    \rceil$. 

    Similarly, suppose $\sin(\sqrt{\delta}/D) < a$, and say we have a copy of $\ket{\Pi}$. Then, we can prepare $\ket{\psi}$ with success probability $\geq 1-\eta$ by sampling from another polynomial with the same degree. 
\end{proposition}
\begin{proof} Figure~\ref{fig:pi2psi_again} shows the transition probabilities when sampling from an odd polynomial starting with $\ket{\Pi}$ or $\ket{\Pi^\perp}$.
    
    Given $\ket{\Pi^\perp}$ and $a < \cos(\sqrt{\delta}/D)$, our goal is to implement the $\ket{\Pi^\perp} \to \ket{\psi}$ transition with probability $\geq 1-\eta$. If we implement the transition with an odd Pellian polynomial $J$, then this transition happens with probability $1 - |J(a)|^2$. So our goal is to find a $J(a)$ satisfying $|J(a)|^2 \leq \eta$ for $a < \cos(\sqrt{\delta}/D)$. The polynomial $J(x)$ from Lemma~\ref{lemma:repairpolys} exactly has this property when $ \cos(\sqrt{\delta}/D) \leq \bar\kappa$. We pick $\kappa := (4/5)\cdot \sqrt{\delta}/D$. Since $0.8x \leq \sin(x)$ for $x\in[0,1]$, we have $\kappa \leq \sin( \sqrt{\delta}/D )$ which implies the desired bound.

    If we are given $\ket{\Pi}$ and $a \geq \sin(\sqrt{\delta}/D)$, then can implement the transition $\ket{\Pi}\to{\psi}$ by sampling from an odd Pellian polynomial with $K$ with success probability $|K(a)|^2$. The polynomial $K(x)$ from  Lemma~\ref{lemma:repairpolys} satisfies $|K(a)|^2 \geq 1-\eta$ for the same choice of $\kappa$.
\end{proof}

Finally, we combine Propositions~\ref{prop:extreme} and~\ref{prop:repair} to prove the main result of this section. The main challenge of this algorithm is that we do not actually know the value of $a$, and moreover never obtain conclusive evidence that $\sin(\sqrt{\delta}/\mathbf{D}) < a$ or $ a < \cos(\sqrt{\delta}/\mathbf{D})$. All we have is the output state of the previous amplitude estimation algorithm and its total degree $\mathbf{D}$. Given this information we essentially \emph{guess} that one of the bounds on $a$ hold. 

\begin{theorem} \label{thm:repair} \textbf{Non-destructive amplitude estimation.} Take any amplitude estimation algorithm that produces an estimate $\mathbf{\hat a}$ and has total degree $\mathbf{D}$. For any $\mu > 0$ there exists another amplitude estimation algorithm that produces the same $\mathbf{\hat a}$, outputs a copy of $\ket{\psi}$ at then end with probability $\geq 1-\delta$, and has total degree at most $\mathbf{D} + \lceil (5/4) (\mathbf{D}/\sqrt{\delta}) \ln(2/\sqrt{\eta})    \rceil$ where $\delta = (4/5)\cdot\mu$ and $\eta = (1/5)\cdot\mu$.
\end{theorem}
\begin{proof} The new algorithm is just the old algorithm plus an extra repair step. We first run the old algorithm, which gives our final estimate $\mathbf{\hat a}$ in some total degree $\mathbf{D}$, as well as one of $\ket{\psi},\ket{\psi^\perp},\ket{\Pi},\ket{\Pi^\perp}$. Then:
    \begin{itemize}

        \item   If we have $\ket{\psi}$, then there is nothing to do, and we successfully return $\ket{\psi}$.

\item   If we have $\ket{\psi^\perp}$, then we measure in the $\ket{\Pi},\ket{\Pi^\perp}$ basis (i.e., sample from $P(a) = a$) and proceed according to the next two cases.

\item  If we have $\ket{\Pi^\perp}$, then we guess that $a \leq \sin(\sqrt{\delta}/\mathbf{D})$ and sample from the polynomial $J(x)$ from Proposition~\ref{prop:repair}. If we obtain $\ket{\psi}$, then we have succeeded and we are done. If we obtain $\ket{\psi^\perp}$, then we have failed and we give up.
    
\item Finally, if we have $\ket{\Pi}$, then we guess that $a \geq \cos(\sqrt{\delta}/\mathbf{D})$ and sample from the polynomial $K(x)$ from Proposition~\ref{prop:repair}. Similarly, we return $\ket{\psi}$ if we obtain it, and otherwise we give up.

    \end{itemize}

    We want to show that the probability that we fail and give up is at most $\mu$. We split the analysis into three cases depending on $a$, although we do not need know know which case we are in in order to run the procedure above.

    Say $a \leq \sin(\sqrt{\delta}/\mathbf{D})$. Then, by Proposition~\ref{prop:extreme}, we have a copy of $\ket{\psi}$ or $\ket{\Pi^\perp}$ with probability at least $1-\delta$. If we have $\ket{\psi}$ we are done, and if we have $\ket{\Pi^\perp}$ then we sample from $J(a)$ which returns a copy of $\ket{\psi}$ with probability $\geq 1-\eta$. So overall we fail with probability at most $\delta + \eta = \mu$.

    Say $a \geq \cos(\sqrt{\delta}/\mathbf{D})$. Then, we have either $\ket{\psi}$ or $\ket{\Pi}$ with probability $\geq 1-\delta$. If we have $\ket{\Pi}$ then we sample from $K(a)$ and obtain $\ket{\psi}$ with probability $\geq 1 - \eta$. So again we fail with probability at most $\delta + \eta$.

    Finally, say  $\sin(\sqrt{\delta}/\mathbf{D}) \leq a \leq \cos(\sqrt{\delta}/\mathbf{D})$. Then we could get any of $\ket{\psi},\ket{\psi^\perp},\ket{\Pi},\ket{\Pi^\perp}$. As usual if we have $\ket{\psi}$ we are done, and if we have $\ket{\psi^\perp}$ then we turn it into $\ket{\Pi}$ or $\ket{\Pi^\perp}$. the condition on $a$ satisfies \emph{both} of the required assumptions in Proposition~\ref{prop:repair}, so the $\ket{\Pi}\to\ket{\psi}$ and $\ket{\Pi^\perp}\to\ket{\psi}$ both succeed with probability $\geq 1-\eta$. So we fail with probability at most $\eta < \mu$.
   
    We numerically find the choice $\delta = (4/5)\cdot\mu$ and $\eta = (1/5)\cdot\mu$ is reasonably close to minimizing $ (1/\sqrt{\delta}) \ln( 2/\sqrt{\eta} ) $.
\end{proof}

In practice, it may be advantageous to actually use the information on $a$ obtained from the previous amplitude estimation algorithm in order to reduce the overall complexity. In theory, however, we have demonstrated that this is not really necessary and that state repair is always possible regardless of the actual value of $a$.

We conclude the section with some additional remarks about this result.

\begin{remark} \textbf{Always returning $\ket{\psi}$, even if it takes forever.} The modified algorithm from Theorem~\ref{thm:repair} is a Monte Carlo algorithm: it only adds at most a fixed number of additional queries, but does not always succeed at preparing $\ket{\psi}$. But it can easily be transformed into a Las Vegas algorithm: First, observe that when $a = 0$ or $a = 1$, then sampling from a polynomial can never produce anything other than $\ket{\psi}$. Otherwise, if $0 < a < 1$ then we consider the following protocol. If we have $\ket{\psi^\perp}$ then measure in the $\ket{\Pi},\ket{\Pi^\perp}$ basis. If we have $\ket{\Pi}$ or $\ket{\Pi^\perp}$, measure in the $\ket{\psi},\ket{\psi^\perp}$ basis. Since all measurement outcomes have nonzero probability, this protocol must reach $\ket{\psi}$ eventually. Now we have Las Vegas algorithm that will eventually produce $\ket{\psi}$ with certainty, but can only inherit the performance guarantees of the original algorithm with probability $\geq 1-\mu$.
\end{remark}

\begin{remark} \textbf{State repair via sampling given bounds on $a$.} Say we are in the same situation as \cite{1907.09965}, where we are given a bound $\kappa < a < \bar \kappa$. Then we can also achieve state repair with only $O(\kappa^{-1} \log(\delta^{-1}))$ queries by sampling from another polynomial. To do so, we follow Theorem~\ref{thm:repair}, but instead of using the bounds $\sin(\sqrt{\delta}/D) < a < \cos(\sqrt{\delta}/D)$ to determine $\kappa$, we just use the $\kappa$ we were given.
\end{remark}

\ifdefined\maindocument
\else
    
    \end{document}
 \fi

\ifdefined\maindocument
    \section{Improved Performance\label{sec:performance}}
\else
    \documentclass[11pt]{article}
    \usepackage[margin=0.4in]{geometry}
    \usepackage{algorithm}
    \usepackage{algpseudocode}
    \usepackage[margin=0.4in]{geometry}
\usepackage[]{graphicx}
\usepackage{amssymb}
\usepackage{enumitem}
\usepackage{hyperref}
\usepackage{amsmath}
\usepackage{amsthm}
\usepackage{multicol}
\usepackage{qcircuit}
\usepackage{braket}
\usepackage{caption}
\usepackage{subcaption}
\usepackage{xcolor}
\usepackage{svg}
\usepackage[htt]{hyphenat}

\newcommand{\eps}{\varepsilon}
\newcommand{\erf}{\text{erf}}
\newcommand{\sgn}{\text{sgn}}
\newcommand{\rect}{\text{rect}}
\newcommand{\bigoh}[1]{\mathcal{O}(#1)}

\newtheorem{theorem}{Theorem}
\newtheorem*{theorem*}{Theorem}
\newtheorem{corollary}[theorem]{Corollary}
\newtheorem{definition}[theorem]{Definition}
\newtheorem{lemma}[theorem]{Lemma}
\newtheorem{remark}[theorem]{Remark}
\newtheorem{example}[theorem]{Example}
\newtheorem{proposition}[theorem]{Proposition}
\newtheorem{fact}[theorem]{Fact}
\newtheorem{empiricalclaim}[theorem]{Empirical Claim}
\newtheorem*{empiricalclaim*}{Empirical Claim}
    \begin{document}
    \begin{center}
        {\Large Improved Performance}
    \end{center}
\fi

The asymptotic performance of amplitude estimation has been known for some time. However, for those interested in forecasting the resource requirements of quantum algorithms built on amplitude estimation, constant-factor improvements in query complexity could make a large difference in the time-scale on which certain quantum applications may be practically realizable. To this end, many competing implementations of amplitude estimation have emerged in the last few years with differing empirical query complexity. 

The first recent attempt at improving the constant-factor performance of amplitude estimation was \cite{1904.10246}, which describes an algorithm called Maximum Likelihood Amplitude Estimation (MLAE). There is sound empirical evidence that this algorithm significantly improves over the performance of \cite{0005055} and numerics indicate that it achieves $O(1/\eps)$ scaling. However, there is no rigorous proof that this algorithm always delivers an accurate answer and the desired performance.

Later, \cite{1912.05559} gave an algorithm called Iterative Quantum Amplitude Estimation (IQAE) that seems to have the best of both worlds. Its empirical performance also outperforms \cite{0005055} and seems comparable to that of MLAE, but it also has a rigorous proof of correctness. The numerics supporting the speedup have been independently reproduced \cite{2005.05300}. The proofs in \cite{1912.05559} even offer bounds on the constant-factor performance, but these bounds are not strong enough to rigorously prove outperforming \cite{0005055}. As the desired precision grows, the likelihood-landscape of which MLAE needs to find the maximum becomes more and more complicated, requiring more and more classical resources to analyze. IQAE also avoids this problem with a simpler classical component to the algorithm. 

Are there any amplitude estimation algorithms that improve over the constant-factor performance of IQAE while simultaneously featuring a rigorous proof? \cite{2003.02417} claims to present a modified version of IQAE with better query complexity, but the experiments in the manuscript do not support this claim\footnote{Compare Figure~3 in \cite{2003.02417} to Figure~3 in \cite{1912.05559}. Recent work \cite{2206.08449} also makes this observation.}. Further, \cite{2206.08449} gave an algorithm where the classical processing is significantly faster then IQAE, but the query complexity is about the same.

In this section we present a modified version of IQAE called `ChebAE' that we empirically demonstrate to require only about $50\%$ to $60\%$ of the queries of IQAE. Since it is a fairly simple modification, it inherits the proof of correctness presented in \cite{1912.05559}. ChebAE was discovered in an attempt to improve the performance of IQAE using the polynomial sampling framework in Theorem~\ref{thm:polysamp}. Recall that IQAE is a special case of such algorithms where only odd-degree Chebyshev polynomials are considered. Does access to a larger family of polynomials allow us to improve the query complexity in terms of constant factors? We were not able to devise an algorithm that answers this question in the affirmative because it seems that the Chebyshev polynomials deliver the best performance. However, in our efforts to perform hyperparameter optimization, we discovered a faster algorithm that is based on Chebyshev polynomials only. The algorithm also has the property that it requires a smaller query complexity when $a \approx 1$.

We now give a brief description of IQAE in terms of the polynomial sampling framework in order to point out what optimization was made. IQAE is inspired by \cite{1908.10846} in that it gradually improves a confidence interval $[a_\text{min},a_\text{max}]$ that contains $a$ with high probability. The method for improving this confidence interval is sketched in Figure~\ref{fig:poly_invert}. First, we find the largest odd integer $d$ such that $|T_d(a)|^2$ is invertible on $[a_\text{min},a_\text{max}]$. Then, we sample from $T_d$ in order to obtain a confidence interval $[p_\text{min},p_\text{max}]$ containing $|T_d(a)|^2$ with high probability. For this purpose we can rely on the Clopper-Pearson method \cite{CP34} which is based on tight bounds on the binomial distribution. Then, since $|T_d(a)|^2$ is invertible, we can compute a new interval $[a_\text{min}^*,a_\text{max}^*]$ where (if $|T_d(a)|^2$ is increasing on the interval):
\begin{align}
    |T_d(a_\text{min}^*)|^2 = p_\text{min} \text{ and } |T_d(a_\text{max}^*)|^2 = p_\text{max}
\end{align}
If $|T_d(a)|^2$ is decreasing on the interval, then we just swap $p_\text{min}$ and $p_\text{max}$ in the equations above. We repeat this process until the confidence interval is as small as desired, or until a better Chebyshev polynomial can be found.

In order to guarantee a bound $\delta$ on the failure probability, we anticipate the total number of confidence intervals $[p_\text{min},p_\text{max}]$ required and divide the failure probability evenly among them. If we demand that the degree $d$ increases by at least a factor of $r$ at every iteration, then the confidence interval must also shrink by a factor of $r$. Thus, at most $\log_{r}(1/\eps)$ iterations are required. With this method, the final asymptotic complexity is $O(\eps^{-1} \log(\delta^{-1}\log(\eps^{-1})))$. \cite{1908.10846} gave a method to remove this extra $\log(\log(\eps^{-1}))$ factor, but this method blows up the constant factor on the query complexity. 

How many samples from $T_d$ should we take at every step? We need to take enough samples in order to shrink the confidence interval by a factor of $r$. But since the shape of $|T_d(a)|^2$ on $[a_\text{min},a_\text{max}]$ is non-linear and the width of the Clopper-Pearson confidence interval depends on the value of $|T_d(a)|^2$ itself, this exact number is hard to anticipate. The best method would be to take samples from $T_d$ one at a time until the confidence interval is small enough so that a larger $d$ can be found. But \cite{1912.05559} found that this is prohibitively slow in terms of classical performance. Observing that the query complexity is dominated by the final iterations of the algorithm, they introduced a parameter $N_\text{shots} = 100$ - the number of samples for the `early' iterations. Once we get to the `later' iterations we take fewer samples at a time.

This is where our optimization comes in: how do we decide what iterations are `early' and what iterations are `late'? IQAE leverages a heuristic based on the maximum error of any Clopper-Pearson confidence interval as well as the assumption that $|T_d(a)|^2$ is linear on the interval. We found that the cutoff from this heuristic is too late, causing IQAE to take more samples than needed. Instead, we introduce a hyperparameter $\nu$ that determines when we switch from sampling $N_\text{shots}$ to $1$. When we set $\nu = 8$, then ChebAE needs only about $50\%$ to $60\%$ of the queries of IQAE.

Having established context, we briefly digress to give intuition as to why Chebyshev polynomials appear to deliver better performance over the broader class of Pellian polynomials. Optimizing Pellian polynomials introduces several computational overheads relative to using only Chebyshev polynomials. Classical evaluation of a Pellian polynomial takes $\mathcal{O}(d)$ time where Chebyshev polynomials can be evaluated in constant time. This can have a large impact on the efficiency of choosing a new polynomial once the confidence interval has been updated. Additionally, checking whether a Chebyshev polynomial is invertible can be done exactly, but for Pellian polynomials one must check for extrema within the target interval via binary search. Failure to detect an extrema within the target interval results in an increased failure probability. Finally, the optimization of QSP angles is a highly non-convex problem which must be tuned carefully. Given that nearly optimal Chebyshev polynomials can be found inexpensively, the potential benefit of using Pellian polynomials is overshadowed by the associated overhead and reduced algorithmic stability. 

ChebAE does add a minor additional flexibility over IQAE: rather than considering only odd-degree Chebyshev polynomials, ChebAE considers Chebyshev polynomials of any degree. But we do not find that this improves performance much. Nonetheless, we find that phrasing the algorithm in terms of Chebyshev polynomials significantly declutters the algorithm compared to IQAE. We also find that our implementation of ChebAE is faster to study classically than the implementation of IQAE that was kindly provided by the authors of \cite{1912.05559}. Our numerical experiments of IQAE were based on the same implementation as the one used for \cite{1912.05559}. Although the modification of IQAE underlying ChebAE is rather slight, we believe that the ChebAE's simpler presentation as well as significant performance improvement merit a full presentation of the algorithm here.

\begin{figure}
    \begin{center}
        \includegraphics[width=0.75\textwidth]{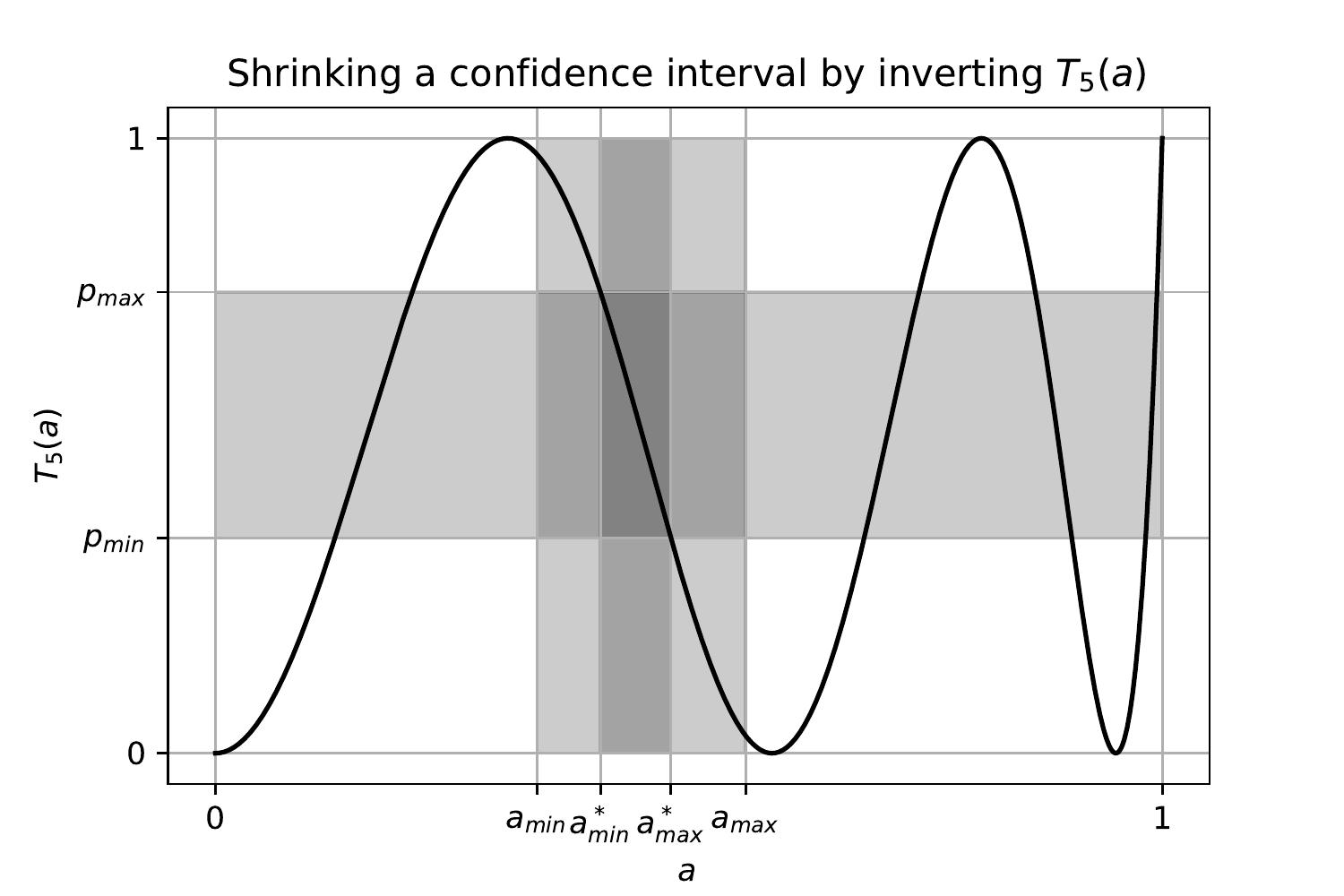}
    \end{center}
    \caption{\label{fig:poly_invert} Say we had established that $a \in [a_\text{min},a_\text{max}] = [0.34, 0.56]$. On this interval, $|T_5(a)|^2$ is invertible. By sampling from the polynomial we obtain a confidence interval $|T_5(a)|^2 \in [p_\text{min},p_\text{max}] = [0.35,0.75]$. Since $T_5(a)$ is invertible on $[a_\text{min},a_\text{max}]$, we can invert $[p_\text{min},p_\text{max}]$ to obtain an improved confidence interval $a \in [a^*_\text{min},a^*_\text{max}] = [0.405,0.481]$. Both IQAE and ChebAE repeat this step for various $T_d(a)$ until the desired accuracy is achieved.} 
\end{figure}

\begin{empiricalclaim}  \label{empericalclaim:chebae} There is an amplitude estimation algorithm `ChebAE' that for any $\eps,\delta > 0$ samples from a random variable $\mathbf{\hat a}$ satisfying $\text{Pr}[ | \mathbf{\hat a} -  a| \geq \eps  ] \leq \delta$. For $a = 0.5$, $\delta = 0.05$, and $\eps \in [10^{-3}, 10^{-6}]$, the observed average query complexity $\langle Q_\Pi\rangle$ satisfies\footnote{When we write $A \approx_\eta B$, we mean $|1 - A/B|\leq\eta$. So for $\eta=4\%$, the smallest values of $\langle Q_\Pi \rangle$ are about $0.96 \times$ the prediction, and the largest are about $1.04 \times$ the prediction.}:
    \begin{align}
        \langle Q_\Pi \rangle \approx_{3.15\%} \frac{1.71}{\eps} \ln \left( 2.08 \ln\left(\frac{1}{\eps} \right) \right)
    \end{align}\\[-1.1cm]
\end{empiricalclaim}
\begin{proof} The ChebAE algorithm relies on a subroutine \textsc{find\_next\_cheb} which we present after the description of the main algorithm. It also has three hyperparameters $r,N_\text{shots},\nu$: $r$ is the factor by which we grow the degree of the Chebyshev polynomial at each iteration, $N_\text{shots}$ is the number of samples from $T_d$ we take at each iteration in the `early' phase of the algorithm, and $\nu$ determines when we switch to the `late' phase of the algorithm. We find the best values of these parameters are:
\begin{align}
    r = 2, \hspace{1cm} N_\text{shots} = 100, \hspace{1cm} \nu = 8
\end{align}

The ChebAE algorithm is as follows.

\begin{enumerate}

    \item Let $T := \lceil \log_r( (2\eps)^{-1} ) \rceil$ be a bound on the number of confidence intervals needed.
    \item Let $\eps^{p}_\text{max}$ be the largest possible error on the estimate of the bias of a coin guaranteed by the Clopper Pearson method with $N_\text{shots}$ flips and confidence $1-\delta/T$.
    
    \item Initialize the confidence interval $[a_\text{min},a_\text{max}] \gets [0,1]$ and the coin toss tally $n_\text{heads},n_\text{flips} \gets 0,0$. Initialize the degree $d \gets 1$.
    
    \item While $a_\text{max} - a_\text{min} \geq 2\eps$ do:\
    \begin{enumerate}
        \item Call \textsc{find\_next\_cheb}$(a_\text{min},a_\text{max})$ to try to find a new degree $ d_\text{new}$.\\ If $d_\text{new} \geq rd$, reset the tally $n_\text{heads},n_\text{flips} \gets 0,0$ and set $d \gets d_\text{new}$.
        \item If the following condition holds then we are `late', otherwise we are     `early'.
        \begin{align}
        \eps^{p}_\text{max} \cdot \frac{a_\text{max} - a_\text{min}}{|T_d(a_\text{max}) - T_d(a_\text{min})|} \leq \eps\nu
        \end{align}
        \item If we are `early', sample $N_\text{shots}$ many times from $T_d$. If we are late, sample from $T_d$ once. Increment the tally $n_\text{heads},n_\text{flips}$ accordingly.
        \item Compute a $\delta/T$ confidence interval $[p_\text{min},p_\text{max}]$ on $|T_d(a)|^2$ using the Clopper-Pearson method \cite{CP34}. 
        \item Invert $[p_\text{min},p_\text{max}]$ to $[a^*_\text{min},a^*_\text{max}]$ using the method in Figure~\ref{fig:poly_invert}.
        \item Update the interval: $[a_\text{min},a_\text{max}] \gets [a^*_\text{min},a^*_\text{max}]\cap [a_\text{min},a_\text{max}]$.
    \end{enumerate}
    \item Let $\hat a$ be the midpoint of $[a_\text{min},a_\text{max}]$.
\end{enumerate}

The purpose of the \textsc{find\_next\_cheb}$(a_\text{min},a_\text{max})$ subroutine is to find the highest $d$ such that $|T_d(a)|^2$ is invertible on the interval $[a_\text{min},a_\text{max}]$. The method is based on the identity $T_d(a) = \cos(d \arccos(a))$.

\begin{enumerate}
    \item Let $[\theta_\text{min},\theta_\text{max}] := [\arccos(a_\text{max}),\arccos(a_\text{min})]$\ 
    \item Initialize $d \gets \left\lfloor\frac{\pi}{2} (\theta_\text{max}-\theta_\text{min})^{-1} \right\rfloor$.
    \item Decrement $d$ until $\cos^2(d \theta)$ has no extrema for $\theta \in [\theta_\text{min},\theta_\text{max}]$. This is equivalent to $\left\lfloor \frac{2 }{\pi}d\theta_\text{min} \right\rfloor = \left\lfloor \frac{2 }{\pi}d\theta_\text{max} \right\rfloor$.
    \item Return $d$.
\end{enumerate}

The proof of correctness is identical to that of IQAE from \cite{1912.05559}: since we grow the degree by a factor of $r$ with every confidence interval, we will need at most $T$ many confidence intervals. Since each interval contains $T_d(a)$ with probability $\geq 1- \delta/T$, the union bound implies that the algorithm fails with probability at most $\delta$ overall. By the same analysis, the asymptotic query complexity is $O( \eps^{-1} \log(\delta^{-1}\log(\eps^{-1}))  )$.

The data supporting the empirical claim is presented in Figure~\ref{fig:Q_vs_eps}. We took 1000 many runs of the above procedure for each of 9 logarithmically spaced $\eps$ in the interval $\eps \in [10^{-2}, 10^{-6}]$ with $\delta = 0.05$ and $a = 0.5$. We found that for each $\eps$ the fraction of runs where $|a - \hat a| > \eps$ was indeed always $<\delta$.

We took the mean of the sampled query complexities $Q_\Pi$ for each value of $\eps$, call it $\langle Q_\Pi\rangle$, and fitted the function $f_{A,B}(\eps) := A\eps^{-1} \ln( B \ln( \eps^{-1}) ) $. Using brute force, we found the parameters $(A,B) = (1.71,2.08)$ ensured that the model $f_{A,B}(\eps)$ was always within a $3.15\%$ relative error of $\langle Q_\Pi\rangle$. The same parameters ensure that $f_{A,B}(\eps)$ are within $71.77\%$ of all values of $Q_\Pi$. We also fitted a simpler model $f_C(\eps) := C/\eps$ to $\langle Q_\Pi\rangle$, and found that the best parameter choice $C = 4.66$ approximates both $\langle Q_\Pi\rangle$ to within $19.65\%$ and $Q_\Pi$ to within $74.28\%$.

The software that performed this analysis is available at:
\begin{center}
    \url{https://github.com/qiskit-community/ChebAE/blob/main/chebae.ipynb}
\end{center}

\end{proof}

\begin{figure}[h]
    \begin{center}
        \includegraphics[width=0.75\textwidth]{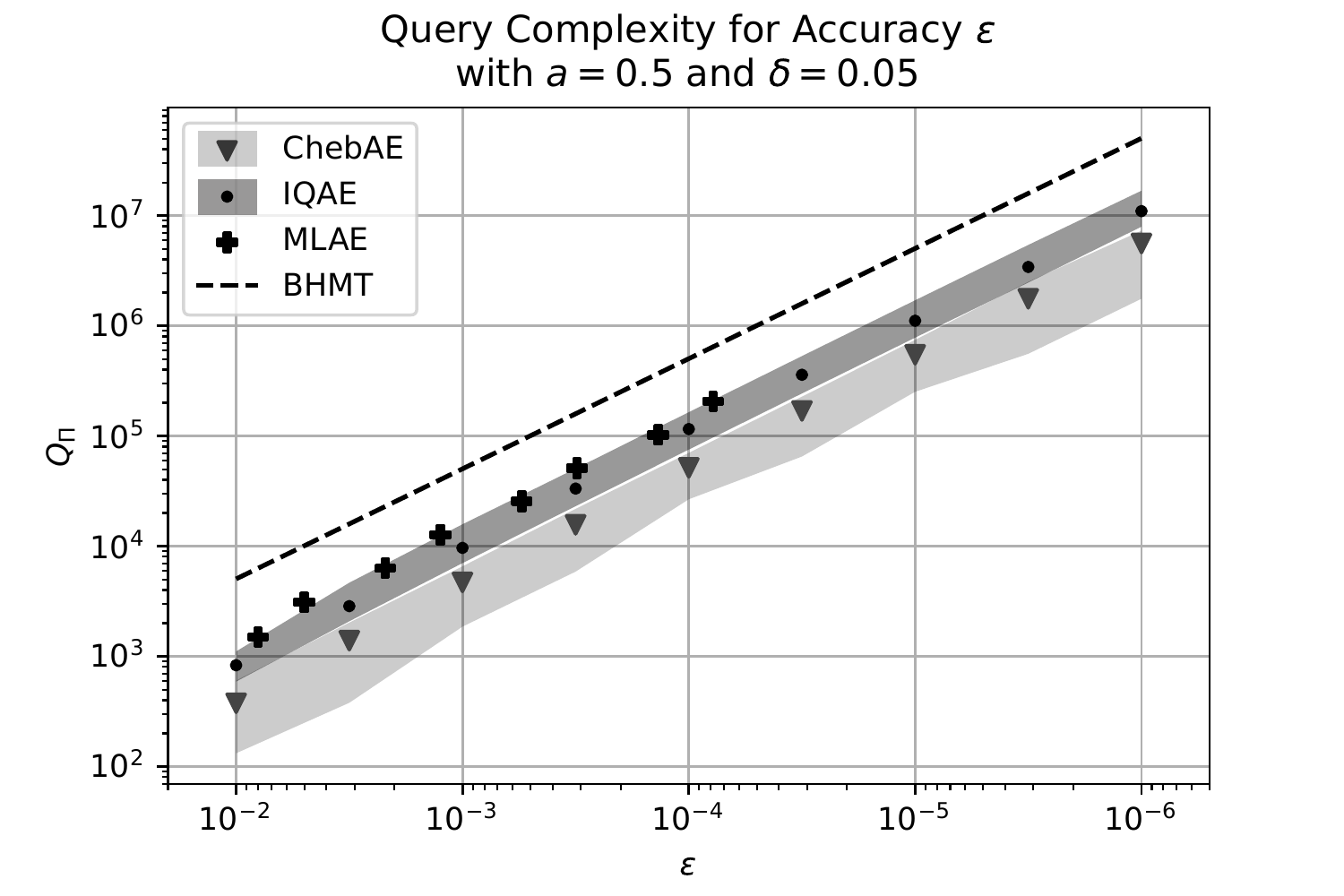}
        \vspace{-7mm}
    \end{center}
    \caption{\label{fig:Q_vs_eps} Comparison of amplitude estimation algorithms with a fixed underlying amplitude and failure probability, and varying target accuracy. Each point is the average of 1000 samples, and the shaded regions around IQAE and ChebAE indicate the minimum and maximum query complexities observed. For MLAE, each point is the average of 100 runs.} 

\end{figure}

A key question in this data analysis is if we should consider in our fits the extra $\log(\log(\eps^{-1}))$ factor, which is present in the asymptotic performance. This is complicated by the fact that $\mathbf{Q}_\Pi$ is a random variable, and visibly exhibits significant variation about the mean. The $f_{A,B}$ fit function captures the empirical mean $\langle Q_\Pi\rangle$ of our data rather accurately, suggesting that the $\log(\log(\eps^{-1}))$ does affect the empirical query complexity. But the most extreme values of $Q_\Pi$ are only within $\sim70\%$ of the $f_{A,B}$ model. However, the simpler $f_{C}$ model does not capture $Q_\Pi$ or $\langle Q_\Pi \rangle$ any better. This suggests that the large error of $\sim70\%$ does not stem from poor fit quality but from actual variation within the random variable $\mathbf{Q}_\Pi$.

Having presented the algorithm, we seek to demonstrate three things. First, we want to argue a significant improvement in query complexity relative to prior art. Second, we want to demonstrate that it was indeed the `late' vs `early' heuristic that resulted in the improvement over IQAE. Third, we show that ChebAE's requires fewer queries when $a \approx 1$.

To compare to the prior art, we characterize the constant-factor query complexity performance of some of the previous algorithms. We consider \cite{0005055} for reference, and also look at MLAE and IQAE since these two seem to have similar performance according to prior literature \cite{1912.05559, 2005.05300}.  Fortunately, the amplitude estimation method from \cite{0005055} has a closed-form expression for its query complexity.

\begin{proposition} \label{prop:bhmt} \cite{0005055} For any $\eps,\delta > 0$, there is an amplitude estimation algorithm (not in the form of Definition~\ref{def:amp_est_alg}) that outputs an estimate $\mathbf{\hat{a}}$ satisfying $\text{Pr}[ |\mathbf{\hat a} - a| \leq \eps  ] > 1- \delta$ with deterministic query complexity $Q_\Pi$ as follows:
    \begin{align}
        Q_\Pi &= \left\lceil \frac{\pi}{\arcsin\eps}  \right\rceil \cdot  \left\lceil \frac{1}{2 } \left(\frac{8}{\pi^2} - \frac{1}{2}\right)^{-2}\ln \frac{1}{\delta}  \right\rceil
    \end{align}
    At $\delta = 0.05$ this is about $Q_\Pi \approx_{1\%} 50/\eps$ for $\eps < 10^{-2}$.
\end{proposition}

For MLAE and IQAE, we perform similar numerical experiments as in Empirical~Claim~\ref{empericalclaim:chebae}.

\begin{empiricalclaim} For $a = 0.5$, $\delta = 0.05$ and $\eps \in [10^{-3}, 10^{-6}]$, the IQAE algorithm from \cite{1912.05559} satisfies:
    \begin{align}
        \langle Q_\Pi \rangle \approx_{7.27\%} \frac{2.62}{\eps} \ln \left( 6.61 \ln\left(\frac{1}{\eps} \right) \right)
    \end{align}
   For the same $\eps,\delta$, the MLAE algorithm from \cite{1904.10246} satisfies:
    \begin{align}
        \langle Q_\Pi \rangle, Q_\Pi \approx_{13.91\%} \frac{3.48}{\eps} \ln \left( 9.89 \ln\left(\frac{1}{\eps} \right) \right) 
    \end{align}
\end{empiricalclaim}
\begin{proof} For IQAE we took 1000 runs with $N_\text{shots}=100$, $r=2$ for the same 9 logarithmically spaced $\eps$ in the interval $\eps \in [10^{-3}, 10^{-6}]$ with $\delta = 0.05$ and $a = 0.5$. The $f_{A,B}$ model with $(A,B) = (2.62,6.61)$ captured the mean with $\langle Q_\Pi\rangle$ to within $7.27\%$, and captured all values with $Q_\Pi$ to within $57.62\%$. The $f_{C}(\eps)$ model fitted to $\langle Q_\Pi\rangle$ with $C = 9.93$ achieved $\langle Q_\Pi\rangle \approx_{16.41\%} f_{C}$ and $Q_\Pi \approx_{72.31\%} f_{C}$.

MLAE does not take $\eps$ as an input parameter. Instead, following \cite{1912.05559}, we take $100$ samples from $T_{2k+1}$ with $k = 0,2^0,2^1,...,2^k$ and compute a confidence interval based on likelihood ratio method. The reported $\eps$ is the half-width of this interval. In some sense, this gives MLAE an unfair advantage: IQAE and ChebAE also produce a confidence interval, and that confidence interval's half-width may be smaller than the requested $\eps$. IQAE and ChebAE do not report that smaller potentially smaller half-width, while MLAE does.

Either way, we perform 100 runs with $k = \{1,\ldots,11\}$ with $\delta=0.05$ and $a=0.5$. Interestingly, we observed that $Q_\Pi$ hardly deviated from its mean $\langle Q_\Pi\rangle$. The $f_{A,B}(\eps)$ model with $(A,B) = (3.48,9.89)$ achieved $\langle Q_\Pi\rangle, Q_\Pi \approx_{13.91\%} f_{A,B}$. The $f_{C}(\eps)$ model with $C = 9.99$ achieved $\langle Q_\Pi\rangle, Q_\Pi \approx_{70.49\%} f_{C}$.

The results of these experiments are also shown in Figure~\ref{fig:Q_vs_eps}, and the fraction of incorrect estimates was $<\delta$ for all $\eps$ (IQAE) and for all $k$ (MLAE).
\end{proof}

Again we are in a situation where the performance can be measured differently depending on what model is chosen. Just like ChebAE, for IQAE the $f_{A,B}$ model captures $\langle Q_\Pi \rangle$ more accurately than $f_{C}$. This again suggests that the $\log(\log(\eps^{-1}))$ dependence is empirically visible. The accuracy of the $f_{A,B}$ and $f_{C}$ models for MLAE are about the same. 

\begin{figure}[h]
    \begin{center}
        \includegraphics[width=0.75\textwidth]{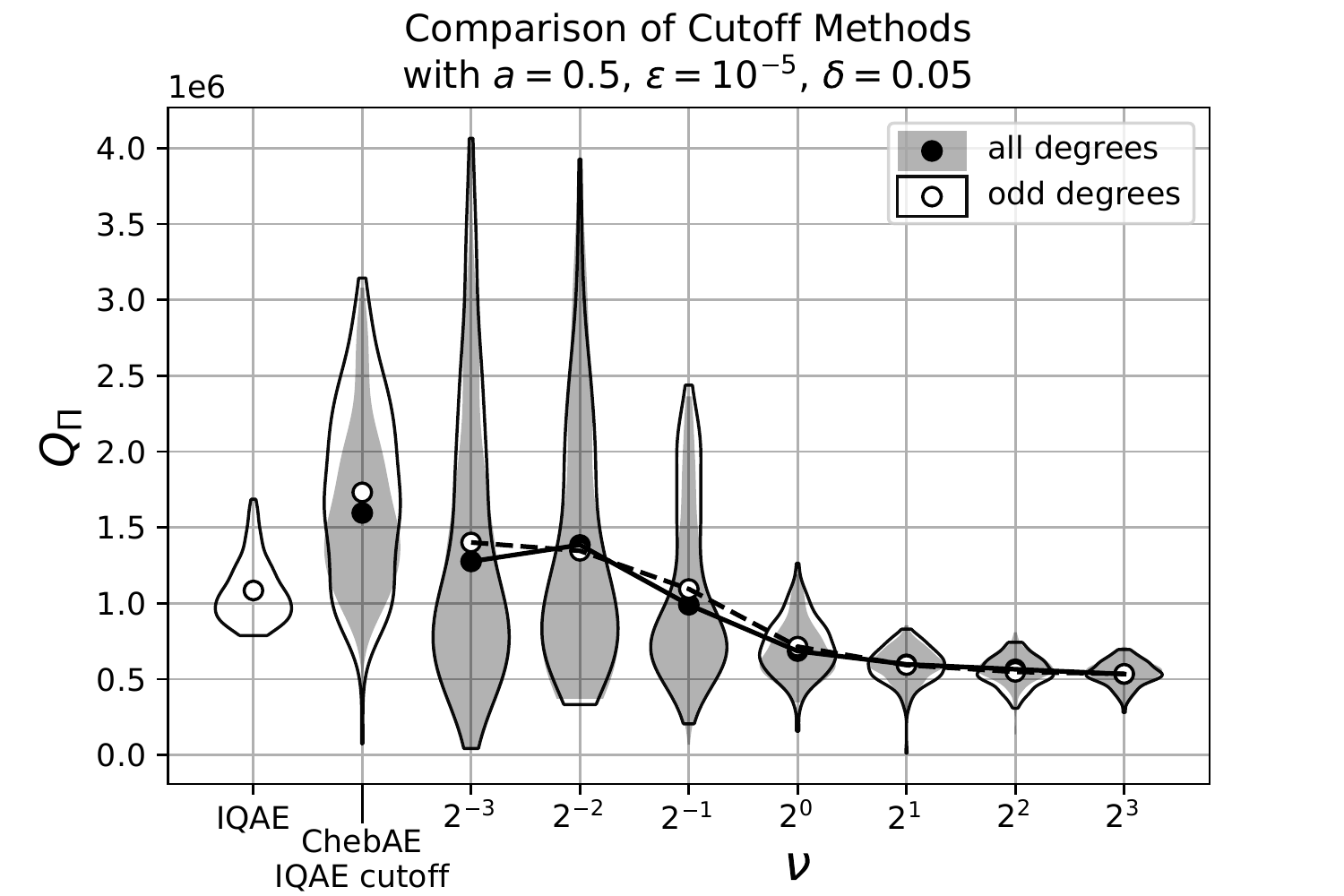}
        \vspace{-5mm}
    \end{center}
    \caption{\label{fig:what_changed} Comparison of different variants of the ChebAE algorithm to the IQAE algorithm. All violins are based on 200 samples and have $a=0.5, \eps=10^{-5},$ and $\delta=0.05$, as well as $r=2$ and $N_\text{shots}=100$. We see how considering Chebyshev polynomials with even degrees as well as odd degrees does not affect the performance much. The primary deciding factor of performance is the cutoff method - step 4(b) in Empirical~Claim~\ref{empericalclaim:chebae}. The leftmost violin shows \cite{1912.05559}'s implementation of IQAE for reference. The violin marked `Chebae IQAE cutoff' is obtained by taking ChebAE and replacing step 4(b) with the cutoff criterion taken from \cite{1912.05559}'s IQAE implementation (see also lines 12-15 in Algorithm~1 of \cite{1912.05559}). The violins on the right leave step 4(b) as in Empirical~Claim~\ref{empericalclaim:chebae}, but vary $\nu$. } 
\end{figure}

Now we justify the claim that ChebAE only needs $45\%$ to $65\%$ of the queries of IQAE. Any such claim will be severely reductive: both of the query complexities are random, depend on the amplitude $a$ being estimated, and the fit methodology also affects the claim. Certainly from Figure~\ref{fig:Q_vs_eps} we can see that ChebAE exhibits a clear speedup over IQAE, with the worst ChebAE $Q_\Pi$ being about the same as the best IQAE $Q_\Pi$. The speedup is visible despite the enormous variation of $Q_\Pi$ around $\langle Q_\Pi \rangle$. To quantify the reduction in query complexiy we will focus only on the average $\langle Q_\Pi \rangle$ and rely on the fits we performed. Both ChebAE and IQAE are better explained by the $f_{A,B}$ model. The fact that we need two parameters to properly fit ChebAE and IQAE makes a comparison more complicated, but ChebAE improves over IQAE in both parameters $A,B$. We can thus be generous to IQAE and ignore the $B$ parameter, and say the fraction is about $1.71/2.62 \approx 65\%$. If we use the $f_{C}$ model then the fraction is $4.66/9.93 \approx 45\%$. By this metric, IQAE only needs about $9.93/50 \approx 20\%$ of the queries of \cite{0005055}.

It remains to justify that it was the altered `late' vs `early' heuristic that was the primary source of the speedup. ChebAE is extremely similar to IQAE, and the differences can be summarized with three alterations: First, rather than obtain an estimate of the Grover angle $\theta := \arcsin(a)$, ChebAE estimates $a$ directly. Second, ChebAE considers both odd and even Chebyshev polynomials whereas IQAE only considers odd polynomials. Third, the `late' vs `early' heuristic in step 4(b) in Empirical~Claim~\ref{empericalclaim:chebae} has a different form and also has the $\nu$ optimization parameter.

We compare these changes individually in Figure~\ref{fig:what_changed}. Here we observe that (for $a = 0.5$) modifying IQAE to estimate $a$ instead of $\theta$ significantly worsens the query complexity on average, and increases the amount of random variation.  We also observe that when ChebAE is restricted to consider only odd polynomials just like IQAE does, then the performance hardly changes. When switching to the new heuristic in step 4(b), the performance is slightly improved even when $\nu$ is small. Then, the query complexity decreases with increasing $\nu$ until even the worst ChebAE runs are faster than the best IQAE runs. So clearly it was the altered heuristic and the tuning of $\nu$ that caused the improvement.

\begin{figure}[h]
    \begin{center}
        \includegraphics[width=0.75\textwidth]{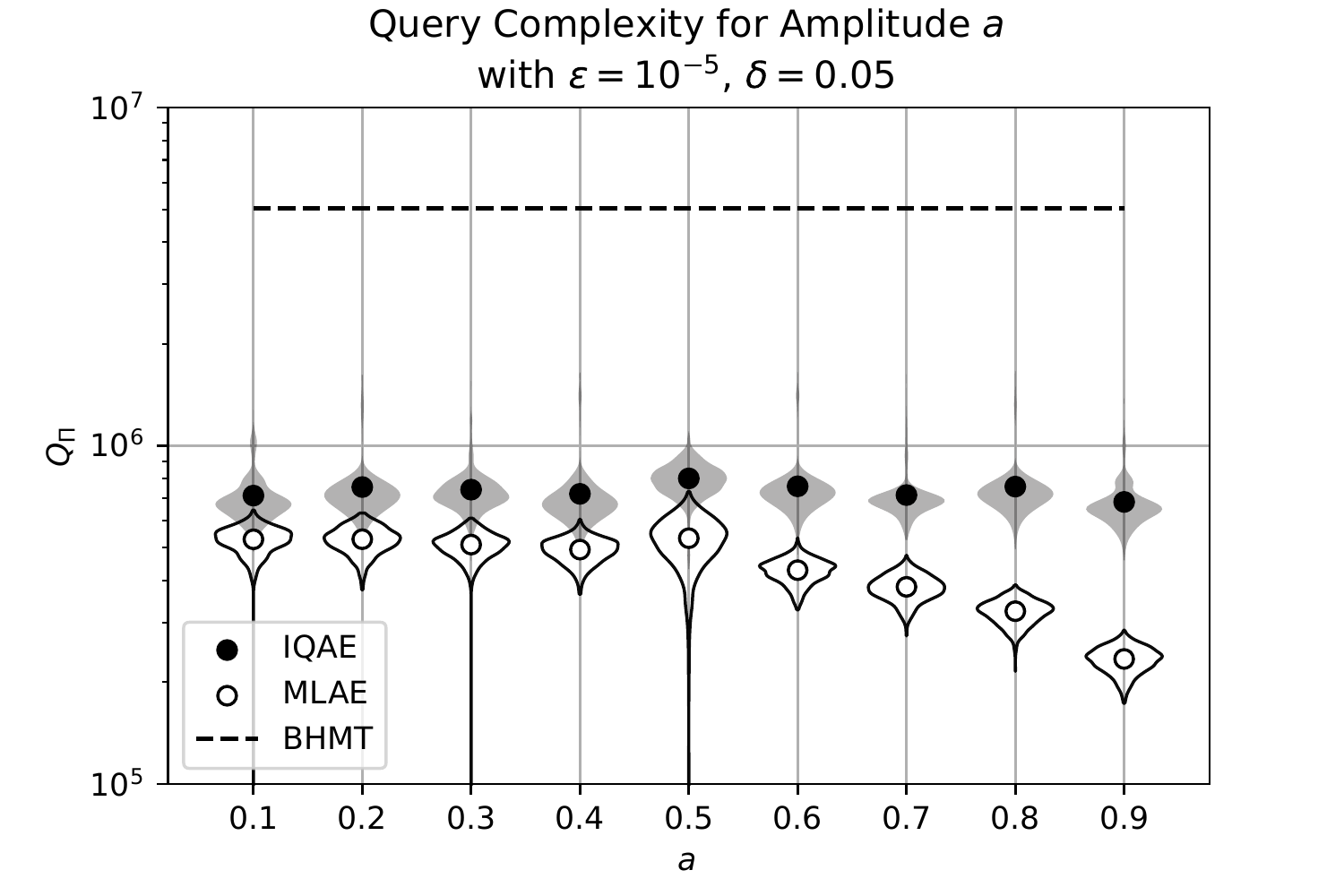}
        \vspace{-7mm}
    \end{center}
    \caption{\label{fig:Q_vs_a} Comparison of amplitude estimation algorithms at a fixed accuracy and success probability while varying underlying amplitude $a$ being estimated. Each violin has 1000 samples, and the dots indicate the means.  } 

\end{figure}

The behavior of $\nu$ can be explained as follows: as $\nu$ increases the heuristic switches from `early' to `late' for earlier iterations, thus decreasing the chances of taking more samples from the expensive Chebyshev polynomials near the end. The improvement with $\nu$ exhibits diminishing returns, and worsens the classical performance of algorithm analysis. We find little benefit in increasing $\nu$ beyond $\nu = 8$, and the classical analysis time is still faster than that of IQAE.

Finally, we remark on the performance as a function of $a$. In Proposition~\ref{prop:bhmt} we claimed that the accuracy \cite{0005055} is independent of $a$. However, Theorem~12 in \cite{0005055} actually has an extra dependence on $a$. Both IQAE and \cite{0005055} estimate the Grover angle $\theta := \arcsin(a)$, and then convert the estimate of $\theta$ to an estimate of $a$. When $a \approx 1$, a less accurate estimate of $\theta$ suffices for the same accuracy on $a$, if only some bounds on $a$ were known beforehand. We consider a setting where nothing is known about $a$ before running the algorithm. Furthermore, we request that the algorithms return an estimate of a certain accuracy $\eps$, and then do not reward the algorithms for achieving a higher accuracy than what was requested. In this more realistic setting for estimation, IQAE and \cite{0005055} must assume the worst-case value of $a$ to determine their query complexity and are then penalized for overshooting their accuracy. Consequently, their query complexity as a function of $a$ is constant.

The experiments presented in Figure~\ref{fig:Q_vs_a} shows that the query complexity of ChebAE delivers smaller query complexities when $a$ is closer to 1. This stems from the fact that ChebAE does not indirectly estimate $a$ via the Grover angle $\theta$, and instead maintains a confidence interval on $a$ directly. By estimating $a$ directly, ChebAE can take advantage of the fact that Chebyshev polynomials have larger slope when $a$ is closer to 1, and thereby exhibit improved performance for these amplitudes even if nothing is known about $a$ beforehand. The algorithm terminates as soon as the estimate of $a$ is good enough, so it does not overshoot and obtain a higher estimate than necessary. Note that for the numerical experiments supporting our claim that ChebAE outperforms IQAE more generally, we deliberately chose $a = 0.5$ because that way we compare IQAE to the slowest possible version of ChebAE.

\ifdefined\maindocument
\else
    
    \end{document}
\fi

\ifdefined\maindocument
    \section{Unbiased Estimation \label{sec:unbiased}}
\else
    \documentclass[11pt]{article}
    \usepackage[margin=0.4in]{geometry}
    \usepackage[margin=0.4in]{geometry}
\usepackage[]{graphicx}
\usepackage{amssymb}
\usepackage{enumitem}
\usepackage{hyperref}
\usepackage{amsmath}
\usepackage{amsthm}
\usepackage{multicol}
\usepackage{qcircuit}
\usepackage{braket}
\usepackage{caption}
\usepackage{subcaption}
\usepackage{xcolor}
\usepackage{svg}
\usepackage[htt]{hyphenat}

\newcommand{\eps}{\varepsilon}
\newcommand{\erf}{\text{erf}}
\newcommand{\sgn}{\text{sgn}}
\newcommand{\rect}{\text{rect}}
\newcommand{\bigoh}[1]{\mathcal{O}(#1)}

\newtheorem{theorem}{Theorem}
\newtheorem*{theorem*}{Theorem}
\newtheorem{corollary}[theorem]{Corollary}
\newtheorem{definition}[theorem]{Definition}
\newtheorem{lemma}[theorem]{Lemma}
\newtheorem{remark}[theorem]{Remark}
\newtheorem{example}[theorem]{Example}
\newtheorem{proposition}[theorem]{Proposition}
\newtheorem{fact}[theorem]{Fact}
\newtheorem{empiricalclaim}[theorem]{Empirical Claim}
\newtheorem*{empiricalclaim*}{Empirical Claim}
    \begin{document}
    \begin{center}
        {\Large Unbiased Estimation}
    \end{center}
\fi

    In this section we present a simple new algorithm for amplitude estimation based on sampling from semi-Pellian polynomials. This algorithm has a new capability: it samples from a random variable $\mathbf{\hat a}$ satisfying $\mathbb{E}[\mathbf{\hat a}] \approx a$. We call this task \textbf{unbiased amplitude estimation.}

    This is the first algorithm achieving unbiased amplitude estimation in the literature to our knowledge. On the other hand, unbiased \emph{phase} estimation has appeared in \cite{2109.10215, 2207.08800, 2207.08643}. The main idea is that if $U$ has eigenphase $e^{i\phi}$, we can prepare $Ue^{i\boldsymbol{\theta}}$ with eigenphase $e^{i(\phi+\boldsymbol{\theta})}$ for a random shift $\boldsymbol{\theta}$. After applying phase estimation to $Ue^{i\boldsymbol{\theta}}$ and subtracting off the random shift yields an unbiased estimator for $\phi$.
    
    However, one cannot use \cite{0005055} to directly translate the technique of \cite{2109.10215} to amplitude estimation. \cite{0005055} directly estimates $\theta := \arcsin(a)$ via phase estimation. Say we had an unbiased estimator $\mathbf{\hat{\boldsymbol{\theta}}}$ of $\theta$. Then, $\sin(\mathbf{\hat{\boldsymbol{\theta}}})$ is not an unbiased estimator of $a$. A different approach is needed.

    Instead, we rely on the following strategy. Say we have established a confidence interval $[a_\text{min}, a_\text{max}]$ within which $a$ is contained with high probability. Suppose furthermore that $a_\text{max} - a_\text{min} < \eps$, so any value in the interval is within $\eps$ of $a$. We construct a polynomial $P(a)$ such that:
    \begin{align}
        |P(a)|^2 \approx  \frac{ a - a_\text{min}  }{a_\text{max} - a_\text{min}}
    \end{align}
    Observe that $|P(a)|^2$ approximates a line with $|P(a_\text{min})|^2 \approx 0$ and $|P(a_\text{max})|^2 \approx 1$. So, we let our final output be:
    \begin{align}
        \mathbf{\hat a} = a_\text{max} \text{ if } \mathbf{b}_{P} = 1\text{, else } a_\text{min}
    \end{align}
    recalling that $\Pr[\mathbf{b} = 1] = |P(a)|^2$ from Theorem~\ref{thm:polysamp}.  Then $\mathbb{E}[\mathbf{\hat a}] = |P(a)|^2 a_\text{max} + (1 - |P(a)|^2) a_\text{min} \approx a$, so the estimator is approximately unbiased. But since we also have $\mathbf{\hat a} \in [a_\text{min}, a_\text{max}]$, we also have a separate guarantee that $\mathbf{\hat a}$ is within $\eps$ of $a$.

    Sampling from the polynomial $P(a)$ immediately solves the problem once we have obtained the confidence interval $[a_\text{min}, a_\text{max}]$. But it turns out that the same polynomial can also be used to refine such a confidence interval, that is, shrinking it by a factor of 0.9 while keeping $a$ in the interval with high probability. This immediately yields a full algorithm for amplitude estimation that is rather similar to \cite{1908.10846, 1912.05559}, but dramatically simpler than either of them.

    Thus, the main technical work of this section is to find a polynomial $P(a)$ with the desired properties. To this end, we use a standard result form the theory of function approximation using polynomials.

    \begin{theorem} \label{thm:jackson} \textbf{Jackson's Theorem.} \cite{rivlin} For any continuous function $g(x)$ on the interval $[-1,1]$ there exists a polynomial $P(x)$ of degree at most $d$ so that for all $x \in [-1,1]$:
\begin{align}
    |J(x) - g(x)| \leq 6 \omega_g(1/d),
\end{align}
where $\omega_g(\delta)$ is the modulus of continuity of $g(x)$ given by:
\begin{align}
  \omega_g := \text{sup}\{  |g(x) - g(y)| \text{ for } x,y \in [-1,1] \text{ with } |x-y|\leq \delta  \}.
\end{align}
\end{theorem}

    Essentially, we can approximate any continuous function $g(x)$ with a polynomial to error $\eps$. The degree will be on the order of the maximum slope of $g$ times $\eps^{-1}$. One key complication remains: we want $|P(a)|^2$ to approximate a line, not $P(a)$. This means we need to account for the Born rule in our choice of $g(x)$, which makes the error analysis a bit more involved. 

    \begin{lemma} \label{lemma:linepoly}  For $0 < a_\text{min} < a_\text{max} <1$, let $\Delta := a_\text{max}-a_\text{min}$. Then for any $\eta>0$, there exists an even semi-Pellian polynomial $P_\eta(x)$ such that:
\begin{align}
\left| |P_\eta(x)|^2 - \frac{x - a_\text{min}}{\Delta}  \right| \leq \eta\end{align}
for all $x \in [a_\text{min},a_\text{max}]$. Furthermore $\text{deg}(P_\eta) \in O( \eta^{-1} \Delta^{-1}    )$. 
\end{lemma}
\begin{proof} Let $c$ be some constant that we will pick later in the proof, and let $g(x)$ be defined on the interval $[-1,1]$ as
\begin{align}
 g(x)= \left\{
    \begin{array}{cl} 
      0 & |x| < a_\text{min} \\
        \sqrt{c\cdot \frac{|x| - a_\text{min}}{\Delta}}&  |x| \in [a_\text{min},a_\text{max}]\\
      \sqrt{c} & |x| > a_\text{max} 
      \end{array}
    \right. 
\end{align}

    Observe that $g(x)$ is even, continuous, and that $\omega_g(1/d) \leq \sqrt{\frac{c}{\Delta d}}$.  Following Theorem~\ref{thm:jackson}, let $J(x)$ be a degree $d$ polynomial satisfying $|J(x) - g(x)| \leq 6 \omega_g(1/d)$. We will also pick $d$ later.

    We would like to ensure that our final polynomial $P_\eta(x)$ is even, even though $J(x)$ might not be. Actually, under the hood, Jackson's theorem relies on a Chebyshev expansion of $g(x)$. Since $g(x)$ is even, all the odd components of the expansion will vanish, so $J(x)$ will be even anyway. But we do not need to rely on this fact to proceed: we simply let $P_\eta(x)$ be the even part of $J(x)$. Then:
    \begin{align}
        P_\eta(x) &:= \frac{1}{2} \big(J(x) + J(-x)\big) \\
        \left| P_\eta(x) - g(x) \right| &\leq   \frac{1}{2} \big(|J(x) - g(x)| + |J(-x) - g(-x)|\big) \\
        &\leq  6 \omega_g(1/d) 
    \end{align}
    $P_\eta(a)$ is even, so it has fixed parity. If we can also show that $|P_\eta(x)| \leq 1$ then by Corollary~10 of \cite{1806.01838} we have demonstrated that $P_\eta(x)$ is semi-Pellian.
    \begin{align}
        \left| |P_\eta(x)|^2 - g(x)^2 \right| &\leq (6\omega_g(1/d) )^2\leq \frac{36c}{\Delta d}\\
        |P_\eta(x)|^2 &\leq \max_x g(x)^2 + \left| |P_\eta(x)|^2 - g(x)^2 \right|\\
        &\leq c +  \frac{36c}{\Delta d}
    \end{align}
Thus, if we select $ c := \left(1 + \frac{36}{\Delta d} \right)^{-1}$, then we have $|P_\eta(x)|^2 \leq 1$, so $|P_\eta(x)| \leq 1$, so $P_\eta$ is semi-Pellian.

    It remains to show that $|P_\eta(x)|^2$ approximates $ \frac{x-a_\text{min}}{\Delta}$ on the interval $[a_\text{min}, a_\text{max}]$, and to select $d$. Let $\bar{c} := 1-c$ and apply the triangle inequality:
    \begin{align}
        \left| |P_\eta(x)|^2 - \frac{x - a_\text{min}}{\Delta}  \right| &\leq \left| |P_\eta(x)|^2 - c\frac{x - a_\text{min}}{\Delta}  \right| + \left| c\frac{x - a_\text{min}}{\Delta}  - \frac{x - a_\text{min}}{\Delta} \right| \\
        &\leq \frac{c}{\Delta d} + (1-c) \\
        &\leq \frac{1 - \bar{c}}{\Delta d} + \bar{c} \\
        &\leq \frac{1 - \bar{c} + \bar{c}\Delta d}{\Delta d} \\ 
        &\leq \frac{1}{\Delta d} - \bar{c}\frac{1+\Delta d}{\Delta d} \leq \frac{1}{\Delta d} 
    \end{align}
    So, we achieve $ \frac{1}{\Delta d} \leq \eta$ when $d := \lceil \eta^{-1}  \Delta^{-1} \rceil$. This completes the proof.
\end{proof}

Now that we have a construction for $P_\eta(a)$, we can present the algorithm for unbiased estimation. The main idea is to take an interval $[a_\text{min}, a_\text{max}]$ starting at $[0,1]$, and to refine it until $\Delta := a_\text{min}-a_\text{max}$ satisfies $\Delta < \eps$. Then we apply the trick we described at the beginning of the section to sample a $\mathbf{\hat a}$ satisfying both $\mathbb{E}[\mathbf{\hat a}] \approx a$ and $\text{Pr}[|\mathbf{\hat a} -a | \geq \eps] \leq \delta$.

Before we step into the main proof, we take note of three things:
\begin{itemize}
    \item Since we always shrink $\Delta$ by the same factor, it is possible to anticipate how many iterations we need: $\sim \log(\eps^{-1})$ many. If we grow our failure probability at each iteration $\delta_t$ via a geometric series, we can avoid an extra $\log(\eps^{-1})$ factor in the query complexity. This analysis is taken directly from \cite{1908.10846}.
    \item For shrinking the interval, we actually can get away with a constant accuracy $\eta = 0.1$  in our polynomial approximation. It is only the final step where we sample $\mathbf{\hat a}$ that requires high accuracy.
    \item Since any random variable bounded in the interval $[a_\text{min}, a_\text{max}]$ must also have its expectation in that interval, we automatically get $|\mathbb{E}[\mathbf{\hat a}] - a | \leq \eps$. So the main capability of the algorithm is to make expected value even closer to $a$. This additional accuracy $\eta$ enters the runtime as $\eta^{-1}$, stemming from the accuracy of Jackson's theorem. Is there a polynomial approximation achieves polylogarithmic scaling in $\eta^{-1}$?
\end{itemize}

\begin{theorem} \label{thm:unbiasedamplitudeestimation} \textbf{Unbiased amplitude estimation.}  For every $\eps,\delta,\eta > 0$ there exists an amplitude estimation algorithm satisfying:
     \begin{align}
        \text{Pr}[|\mathbf{\hat a}  - a| \geq \eps] \leq \delta
    \end{align}\\[-1.3cm]
    \begin{align}
        \mathbf{D} \text{ is not random and } \mathbf{D} \in O( \eps^{-1}( \log(\delta^{-1}) + \eta^{-1}  )  )
    \end{align}\\[-1.3cm]
      \begin{align}
        \left| \mathbb{E}[\mathbf{\hat a}] - a \right| \leq  \eps\eta + \delta
   \end{align}
\end{theorem}
    \begin{proof} The algorithm proceeds as follows:
        \begin{enumerate}
            \item Initialize $a^{(0)}_\text{min} \leftarrow 0$ and $a^{(0)}_\text{max} \leftarrow 1$. Let $\Delta_t := a^{(t)}_\text{max} - a^{(t)}_\text{min}$ throughout.
            \item Let $T := \lceil\log_{0.9}(\eps)\rceil$. For $t = 0,1,2,...,T-1$, do:
                \begin{enumerate}
                    \item Use Lemma~\ref{lemma:linepoly} to construct a polynomial $P^{(t)}_{0.1}(a)$ whose square is $0.1$-close to $(a - a^{(t)}_\text{min})/\Delta_t$ on $[a^{(t)}_\text{min},a^{(t)}_\text{max}]$.
                    \item Let $\delta_t := (\eps\delta/10)  \cdot 0.9^{-t}$, sample from $\mathbf{b}_{P^{(t)}_{0.1}}$ a total of $m_t := 6 \ln(\delta_t^{-1})$ times, and let $\mathbf{S}^{(t)}$ be the total.
                    \item If $\mathbf{S}^{(t)} > m_t/2$, then set:
                        \begin{align}
                            [a^{(t+1)}_\text{min},\hspace{1mm}a^{(t+1)}_\text{max}] \leftarrow [a^{(t)}_\text{min} + 0.1\Delta_t,\hspace{1mm}a^{(t)}_\text{max}]
                        \end{align}
                             Otherwise, set:
                        \begin{align}
                            [a^{(t+1)}_\text{min},\hspace{1mm}a^{(t+1)}_\text{max}] \leftarrow [a^{(t)}_\text{min},\hspace{1mm}a^{(t)}_\text{max} -  0.1\Delta_t]
                        \end{align}
                \end{enumerate}
            \item  Use Lemma~\ref{lemma:linepoly} to construct a polynomial $P^{(T)}_\eta(a)$ whose square is $\eta$-close to $(a - a^{(T)}_\text{min})/\Delta_T$ on $[a^{(T)}_\text{min},a^{(T)}_\text{max}]$.
            \item Sample from $\mathbf{b}_{P^{(T)}_\eta}$ once. If  $\mathbf{b}_{P^{(T)}_\eta} = 1$, then return $\mathbf{\hat a} \leftarrow a^{(T)}_\text{max}$. Otherwise return $\mathbf{\hat a} \leftarrow a^{(T)}_\text{min}$.
        \end{enumerate}

        At the beginning we have $\Delta_0 = 1$, and each of the $T$ iterations of step 2 transforms $\Delta_{t+1} = 0.9\Delta_t$. So, at the end we have $\Delta_T = 0.9^{T} \leq \eps$.
      
        We first establish that $\text{Pr}[|\mathbf{\hat a}  - a| \geq \eps] \leq \delta$. This is guaranteed if $a \in [a^{(T)}_\text{min}, a^{(T)}_\text{max}]$ after step 2 finishes, since afterwards both $|a^{(T)}_\text{max} - a| \leq \Delta_T \leq \eps$ and $|a^{(T)}_\text{min} - a| \leq \Delta_T \leq \eps$. So it suffices to show that the probability we ever have $a \not\in [a^{(T)}_\text{min}, a^{(T)}_\text{max}]$ is at most $\delta$.
   
        Say we are at step $t$ of step 2, and we have $a \in [a^{(t)}_\text{min}, a^{(t)}_\text{max}]$. Then, there are two ways we could update the interval to make $a \not\in [a^{(t+1)}_\text{min}, a^{(t+1)}_\text{max}]$: either $a \leq a^{(t)}_\text{max} - 0.9\Delta_t$ and $\mathbf{S}^{(t)} > m_t/2$, or $a \geq a^{(t)}_\text{min} + 0.9\Delta_t$ and $\mathbf{S}^{(t)} \leq m_t/2$. We show that both of these cases occur with probability at most $\delta_t$.
    
        Say $a \leq a^{(t)}_\text{max} - 0.9 \Delta_t = a^{(t)}_\text{min} + 0.1 \Delta_t$. We have:
        \begin{align}
            |P^{(t)}_{0.1}(a)|^2 &\leq \frac{a - a^{(t)}_\text{min}}{\Delta_t} + 0.1\\
                        &\leq 0.1 + 0.1 = 0.2.
        \end{align}
        So $\mathbb{E}[\mathbf{b}_{P^{(t)}_{0.1}}] \leq 0.2$, implying $\mathbb{E}[\mathbf{S}^{(t)}] \leq 0.2m_t$. Using the Chernoff-Hoeffding theorem:
        \begin{align}
            \text{Pr}[ \mathbf{S}^{(t)} > m_t/2 ] &\leq \text{Pr}[ \mathbf{S}^{(t)} \geq \mathbb{E}[\mathbf{S}^{(t)}] +    0.3 m_t]\\
        &\leq e^{-2(0.3)^2 m_t} \leq \delta_t.
        \end{align}
        Similarly say $a \geq a^{(t)}_\text{min} + 0.9 \Delta_t$. Then $\mathbb{E}[\mathbf{b}_{P^{(t)}_{0.1}}]   \geq 0.8$, so $\text{Pr}[ \mathbf{S}^{(t)} \leq m_t/2 ] = \text{Pr}[  \mathbf{S}^{(t)} \leq \mathbb{E}[\mathbf{S}^{(t)}] - 0.3m_t ] \leq \delta_t$ by the same argument.

        To finish the argument that $\text{Pr}[|\mathbf{\hat a}  - a| \geq \eps] \leq \delta$, we follow the analysis of \cite{1908.10846}. Let $b := 1/0.9 \approx 1.11$ and observe that $T = \lceil \log_b(\eps^{-1}) \rceil$. We bound the probability that any of the steps fail with a union bound:
        \begin{align}
            \sum_{t=0}^{T-1} \delta_t = \frac{\eps\delta}{10}  \sum_{t=0}^{T-1} b^t  \leq \frac{\eps\delta}{10}\frac{b^T - 1}{b-1} \leq  \frac{\delta}{10}\frac{1}{b-1} \leq \delta.
        \end{align}

        Now we show that $\mathbf{D}$ only takes a single value with probability 1. We observe that in Lemma~\ref{lemma:linepoly}, while $P^{(t)}_\eta(a)$ itself depends on the individual values of $a^{(t)}_\text{min}, a^{(t)}_\text{max}$, the degree only depends on their difference $\Delta_t$. While the values $a^{(t)}_\text{min}, a^{(t)}_\text{max}$ may vary randomly in the algorithm, $\Delta_t$ is always exactly $0.9^t$. Since the total number of polynomials we sample from is also independent of $a^{(t)}_\text{min}, a^{(t)}_\text{max}$, to total degree of all polynomials $\mathbf{D}$ is deterministic. This makes asymptotic claims on $\mathbf{D}$ well defined.

        At the $t$'th iteration of step 2, we sample from a polynomial of degree $\in O( \Delta_t^{-1} 0.1^{-1} ) = O(0.9^{-t}) =  O( b^{t} )$ a total of $m_t$ times. Then, in step 4, we sample from a polynomial of degree $\in O( \Delta_t^{-1} \eta^{-1} ) \subseteq O(  \eps^{-1}  \eta^{-1} )$ once. So the total degree is:
        \begin{align}
            \mathbf{D} \in O\left( \sum_{t=0}^{T-1}   \frac{m_t}{0.9^t}  + \frac{1}{\eta \eps}   \right) &\leq  O\left( \sum_{t=0}^{T-1}   \frac{1}{0.9^t} \ln\left(\frac{1}{\delta_t}\right)  + \frac{1}{\eta \eps}   \right)\\
            &\leq  O\left( \sum_{t=0}^{T-1}   \frac{1}{0.9^t} \ln\left(\frac{0.9^t}{\eps\delta}\right)  + \frac{1}{\eta \eps}   \right)\\
            &\leq  O\left( \frac{1}{\eps} \ln\left(\frac{1}{\delta}\right)  + \frac{1}{\eta \eps} \right) + O\left( \sum_{t=0}^{T-1} b^t \left[ \ln\left(\frac{1}{\eps} \right) - t \ln(b) \right]\right) 
        \end{align}
        The first term is the desired bound, so all that remains to show is that the second term is dominated by the first term. Indeed, a somewhat cumbersome calculation in equations (38-44) of \cite{1908.10846} demonstrates that the second term is $\leq O( \eps^{-1} )$ (albeit with $T+1$ instead of $T$, which is a stronger result).

        Finally, we prove the fact that makes this construction unique: that $\left|\mathbb{E}[\mathbf{\hat a}] - a\right| \leq \eps\eta + \delta$. To do so, we interpret the dynamics of the algorithm as a binary decision tree: at each node we can either proceed to the left child where we increase $a_\text{min}$ or the right child where we decrease $a_\text{max}$. At the $t$'th iteration of step 2 we are at the $t$'th layer of the tree, and we can think of step 4 as the $T$'th layer containing all the leaves. See Figure~\ref{fig:recursiveexpectation}.

        We can label each node in the $t$'th layer of the tree with a bit string $\vec x\in\{0,1\}^t$, where the bits encode the choices made to get to that node. We say a node $\vec x$ is `good' if the corresponding interval $[a^{(t)}_\text{min}, a^{(t)}_\text{max}]$ contains $a$, otherwise we say $\vec x$ is `bad'. Each node is associated with a random variable $\mathbf{\hat a}_{\vec x}$, denoting the output of the algorithm assuming we were at the $\vec x$ node at iteration $t = \text{dim}(\vec x)$. Correspondingly, at the root of the tree $\mathbf{\hat a}$ equals $\mathbf{\hat a}_{\vec x}$ where $\vec x$ is the `empty' bit string with dimension 0.

        We establish $\left|\mathbb{E}[\mathbf{\hat a}] - a\right| \leq \eps\eta + \delta$ by proving the following fact. For all $t \in [0,1,..., T]$ and for all good $\vec x \in \{0,1\}^t$:
        \begin{align}
            \left|\mathbb{E}[\mathbf{\hat a}_{\vec x}] - a\right| \leq \eps\eta + \sum_{k=t}^{T-1} \delta_k
        \end{align}
        Then $\left|\mathbb{E}[\mathbf{\hat a}] - a\right| \leq \eps\eta + \delta$ follows by plugging in $t=0$ and noting that $\sum_{k=0}^{T-1}\delta_k \leq \delta$ as shown above.

            We prove the statement recursively, starting with the leaves and working our way to the root. For a leaf, $t = T$ so the sum $\sum_{k=T}^{T-1} \delta_k$ vanishes. We have:
        \begin{align}
            \mathbb{E}[  \mathbf{\hat a}_{\vec x}] &= |P^{(T)}_{\eta}(a)|^2 \cdot a_\text{max} - (1 - |P^{(T)}_\eta(a)|^2) \cdot a_\text{min} \\
            &= |P^{(T)}_{\eta}(a)|^2 \cdot \Delta_T + a_\text{min}
        \end{align}
By Lemma~\ref{lemma:linepoly}, we have $||P^{(T)}_\eta(a)|^2 - (a-a^{(T)}_\text{min})/\Delta_T| \leq \eta$. So:
    \begin{align}
        \eta \Delta_T &\geq \left||P^{(T)}_\eta(a)|^2 - (a-a^{(T)}_\text{min})/\Delta_T\right| \cdot \Delta_T\\
        &=\left| \left(|P_\eta(a)|^2 \cdot \Delta_T +a^{(T)}_\text{min}\right) - a\right| =\left| \mathbb{E}[ \mathbf{\hat a}_{\vec x}] - a\right| 
    \end{align}

    At step 4 we have $\Delta_T \leq \eps$, so we have established $\left| \mathbb{E}[ \mathbf{\hat a}_{\vec x}] - a\right| \leq \eps \eta$ as desired.

    Now we recursively show the statement for $t$ and all good $\vec x \in \{0,1\}^t$, assuming it holds for all good $\vec xy$ at $t+1$, with $y \in \{0,1\}$. There are two cases: either $\vec xy$ is good for both $y \in \{0,1\}$, or it is only good for one of them. It cannot be bad for both $y$, because then $a$ would need to be both $\leq a^{(t)}_\text{min}+0.1\Delta_t$ and $\geq a^{(t)}_\text{min}+0.9\Delta_t$.

    If $\vec xy$ is good for both $y$ (like $\vec x = 01$ in Figure~\ref{fig:recursiveexpectation}), then by assumption $|\mathbb{E}[ \mathbf{\hat a}_{\vec xy}] - a| \leq \eps \eta + \sum_{k=t+1}^T \delta_k$. For some probability $p$ we have:
    \begin{align}
        \mathbb{E}[ \mathbf{\hat a}_{\vec x} ] &=  p \cdot \mathbb{E}[ \mathbf{\hat a}_{\vec x0} ] + (1 - p) \cdot  \mathbb{E}[ \mathbf{\hat a}_{\vec x1} ]\\
        |\mathbb{E}[ \mathbf{\hat a}_{\vec x} ]  - a| &\leq p \cdot |\mathbb{E}[ \mathbf{\hat a}_{\vec x0} ]  - a| +  (1-p) \cdot |\mathbb{E}[ \mathbf{\hat a}_{\vec x1} ]  - a| \\
        &\leq \eps \eta + \sum_{k=t+1}^T \delta_k \leq \eps \eta + \sum_{k=t}^T \delta_k
    \end{align}
    Now assume without loss of generality that $\vec x0$ is bad, and $\vec x1$ is good (like $\vec x = 0$ in Figure~\ref{fig:recursiveexpectation}). Now the assumption says nothing about $| \mathbb{E}[\mathbf{\hat a}_{\vec x0}]  - a |$, but we at least know that $| \mathbb{E}[\mathbf{\hat a}_{\vec x0}]  - a | \leq 1$. Furthermore, we know that the probability $p$ with which we proceed down the $\vec x0$ path is at most $\delta_t$. So:
    \begin{align}
        |\mathbb{E}[ \mathbf{\hat a}_{\vec x} ]  - a| &\leq \delta_t \cdot |\mathbb{E}[ \mathbf{\hat a}_{\vec x0} ]  - a| +   |\mathbb{E}[ \mathbf{\hat a}_{\vec x1} ]  - a| \\
        &\leq \eps \eta + \delta_t +  \sum_{k=t+1}^T \delta_k = \eps \eta + \sum_{k=t}^T \delta_k.
    \end{align}
    \end{proof}

\begin{figure}
    \begin{center}
        \includegraphics[width=0.6\textwidth]{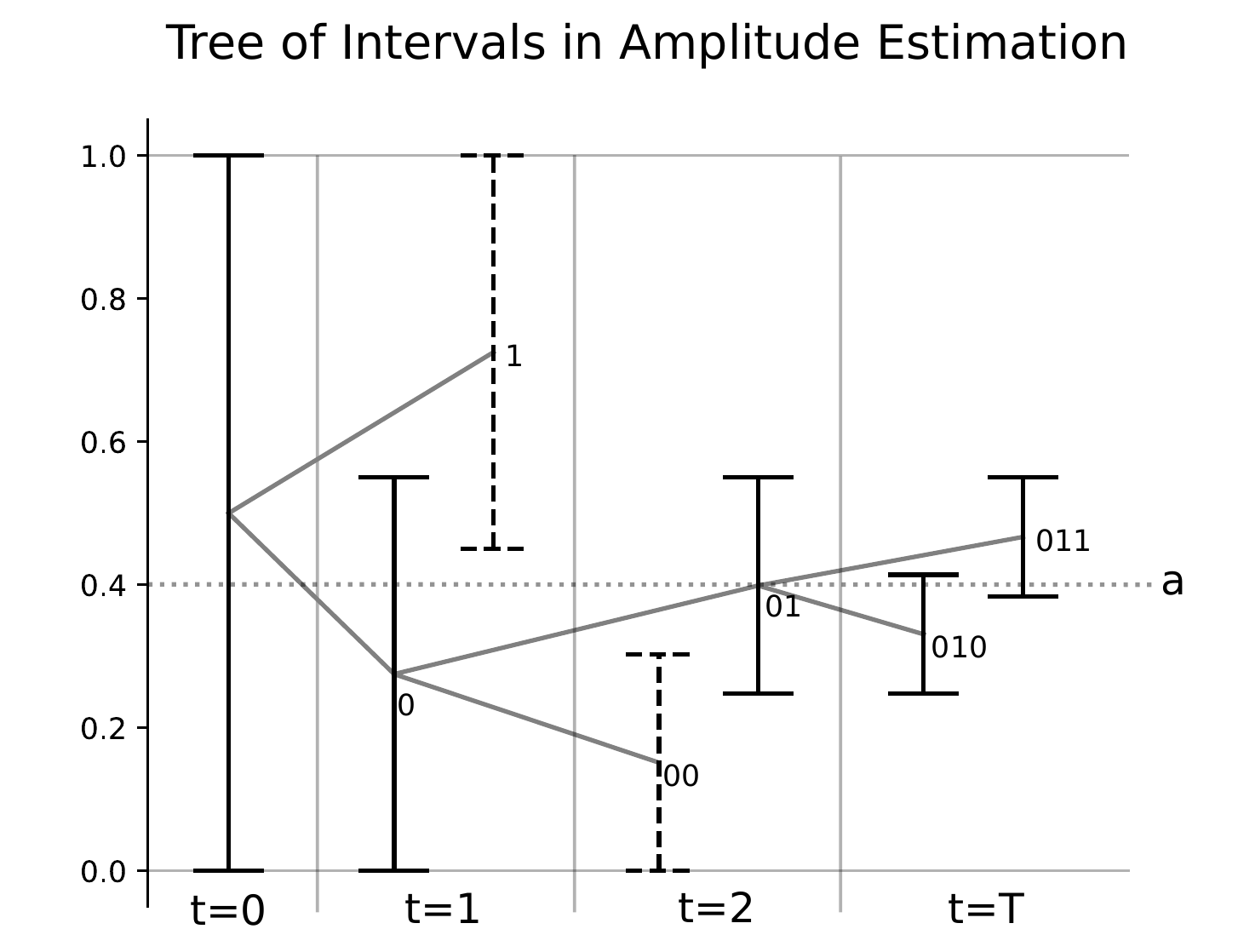}
    \end{center}
    \caption{\label{fig:recursiveexpectation} Sketch of the intervals considered in the proof of Theorem~\ref{thm:unbiasedamplitudeestimation}, labeled by the string $\vec x \in \{0,1\}^t$. At each iteration, we cut off either the right or left side of the interval which corresponds to adding either a $0$ or $1$ to $\vec x$. `Bad' intervals not containing the true amplitude $a = 0.4$ are marked with dashed lines, and their children are not shown. The algorithm describes shrinking the interval by a factor of 0.9 with every step, but for visual clarity we use a factor of 0.55 here. } 

\end{figure}

\ifdefined\maindocument
\else
    
    \end{document}
\fi

\ifdefined\maindocument
    \section{Hybrid Classical-Quantum Estimation \label{sec:hybrid}}
\else
    \documentclass[11pt]{article}
    \usepackage[margin=0.4in]{geometry}
    \usepackage[margin=0.4in]{geometry}
\usepackage[]{graphicx}
\usepackage{amssymb}
\usepackage{enumitem}
\usepackage{hyperref}
\usepackage{amsmath}
\usepackage{amsthm}
\usepackage{multicol}
\usepackage{qcircuit}
\usepackage{braket}
\usepackage{caption}
\usepackage{subcaption}
\usepackage{xcolor}
\usepackage{svg}
\usepackage[htt]{hyphenat}

\newcommand{\eps}{\varepsilon}
\newcommand{\erf}{\text{erf}}
\newcommand{\sgn}{\text{sgn}}
\newcommand{\rect}{\text{rect}}
\newcommand{\bigoh}[1]{\mathcal{O}(#1)}

\newtheorem{theorem}{Theorem}
\newtheorem*{theorem*}{Theorem}
\newtheorem{corollary}[theorem]{Corollary}
\newtheorem{definition}[theorem]{Definition}
\newtheorem{lemma}[theorem]{Lemma}
\newtheorem{remark}[theorem]{Remark}
\newtheorem{example}[theorem]{Example}
\newtheorem{proposition}[theorem]{Proposition}
\newtheorem{fact}[theorem]{Fact}
\newtheorem{empiricalclaim}[theorem]{Empirical Claim}
\newtheorem*{empiricalclaim*}{Empirical Claim}
    \begin{document}
    \begin{center}
        {\Large Hybrid Classical-Quantum Estimation}
    \end{center}
\fi

    In the future, we may have access to quantum computers that are capable of running Grover's algorithm but are still limited in the overall circuit depth. Amplitude estimation algorithms may become useful earlier if they can still accelerate estimation to accuracy $\eps$ without requiring depth $O(1/\eps)$. This can be achieved by mixing classical and quantum resources.
    
    In particular, \cite{2012.03348} consider the following version of this problem: for any $0 \leq \beta < 1$ they show that there exists an amplitude estimation algorithm that runs in time $O(1/\eps^{1+\beta})$, but requires a maximum depth of only $O(1/\eps^{1-\beta})$. When $\beta = 0$ this reproduces the `fully quantum' performance of other algorithms, and when $\beta \to 1$ this algorithm has the same asymptotic complexity as classical probability estimation.

    We consider the construction of such hybrid algorithms a fascinating theoretical problem.   \cite{2012.03348} give two methods for achieving this: a `Power law' algorithm relying on maximum likelihood estimation, and a `QoPrime' algorithm relying on the Chinese remainder theorem. In this section we give another method for solving this problem based on quantum signal processing. Depending on the reader's background this new method may be more intuitive.

    However, we remark that we consider it questionable if any of these methods are actually useful in practice. In particular, constructing the algorithms to achieve the nice theoretical scaling of $O(1/\eps^{1\pm\beta})$ undercuts the simplicity of the underlying situation. A depth-limited quantum computer will be able to run some fixed number of Grover rotations $r$. A practical version of hybrid estimation would take as input such an $r$, and use that to extract as much information about $a$ as possible. In the context of quantum signal processing, this $r$ corresponds to some fixed `degree budget'. We could then use machine learning to find whatever polynomial of degree $\leq r$ gives the most information about $a$. That is not to say that \cite{2012.03348} does not consider practicality: they perform a noise analysis, and even implement their algorithm in hardware \cite{2109.09685} albeit with a trivial oracle. We merely invite the possibility of simpler algorithms with less elegant theoretical properties.

    There is a reason to be skeptical of the $\beta \to 1$ limit in particular: we expect early quantum computers to have significantly slower clock speeds than classical computers. Thus, classical oracle evaluation is likely much faster than quantum oracle evaluation. So if one is going to put in all the effort of running the oracle on the quantum computer, one might as well do so as many times as possible. A perhaps more reasonable notation of hybrid quantum-classical is discussed by \cite{2202.11443}: we have two different types of queries, quantum and classical, each with different cost. Unfortunately, \cite{2202.11443} shows that there is no speedup for Grover's search in this setting. Is the same true for estimation?

    Despite these observations, we consider the task of varying the complexity with $\beta$ an excellent opportunity to display the power of the polynomial sampling framework. We now discuss the new algorithm for achieving depth $\mathbf{d} \in O(1/\eps^{1-\beta})$ and overall complexity $\mathbf{D} \in O(1/\eps^{1+\beta})$ based on quantum signal processing. The final result is presented in Theorem~\ref{thm:hybridestimation}. While the polynomial constructions require cumbersome error analysis as usual, the core idea is quite simple. In light of the above discussion, we do not attempt a constant-factor performance analysis on this algorithm.

    We take the same approach as in Theorem~\ref{thm:unbiasedamplitudeestimation}: we maintain an interval $[a^{(t)}_\text{min}, a^{(t)}_\text{max}]$ containing $a$ with high probability, and shrink the interval with every iteration depending on what side of the interval $a$ is on. Letting $\Delta_t := a^{(t)}_\text{max} - a^{(t)}_\text{min}$, we test this by finding a polynomial $P(a)$ that is $\leq \frac{1}{2} - \gamma$ when $a$ is  $\leq a^{(t)}_\text{min} + 0.1 \Delta_t$, and is $\geq \frac{1}{2} + \gamma$ when $a$ is $\geq a^{(t)}_\text{max} - 0.1\Delta_t$. This can be achieved by making $P(a)$ approximate a line with slope $\propto \gamma/\Delta_t$, and the degree is largely determined by the slope. Then, $\propto 1/\gamma^2$ samples from the polynomial suffice to distinguish the two cases.

    In the previous section we selected $\gamma$ to be constant, so the degree is $\propto 1/\Delta_t$. This makes the depth and the total number of queries $\propto 1/\Delta_t$. To achieve hybrid estimation, we simply select $\gamma \propto \Delta_t^{\beta}$. Then the depth is $\propto \Delta_t^\beta / \Delta_t = 1/\Delta_t^{1-\beta}$. But, since we need to take $\propto 1/\gamma^2 \propto 1/\Delta_t^{2\beta}$ samples from the polynomial, the total complexity is $ 1/\Delta_t^{1+\beta} $. The runtime is dominated by the final iterations where $\Delta_t \propto \eps$, so the final depth is $\propto 1/\eps^{1-\beta}$ and the final complexity is $\propto 1/\eps^{1+\beta}$ as desired.

    Letting $a^{(t)}_\text{mid} := a^{(t)}_{\text{min}}+0.5\Delta_t$, we require a polynomial that approximates a line passing through $(a^{(t)}_\text{mid}, 1/2)$ with whatever slope we want. Once such a polynomial has been constructed the remainder of the analysis is extremely simple as shown above. However, such a polynomial is unfortunately not so easy to build because of two complications:
    \begin{itemize}
        \item Jackson's theorem (Theorem~\ref{thm:jackson}) is not precise enough. The line that we desire to approximate only creates a gap $\gamma \propto \Delta^{\beta} \in o(1)$. The approximation error better be be smaller than $\gamma$, but Jackson's theorem requires an additional factor of $1/\gamma$ in the degree to achieve this. This is too slow, so we need a more accurate polynomial approximation of a line. Fortunately, \cite{1707.05391} constructed an extremely accurate approximation of $\erf(kx) \approx kx$ via the Jacobi-Anger expansion.
        \item Semi-Pellian polynomials have fixed parity, so even though $\erf(kx)$ is odd, the shifted version $\erf(k(a-a_\text{mid}))$ has mixed parity. We can fix this by instead approximating the even function $\erf(k(a-a_\text{mid})) + \erf(-k(a+a_\text{mid})) + 1$, which is very similar to $\erf(k(a-a_\text{mid}))$ for $a \geq 0$. But this only works when $a_\text{mid}$ is not too close to the origin - otherwise the two $\erf$ functions interfere. This forces us to adopt special behavior for small $a_\text{mid}$, causing $O(a^{-1})$ terms in the final runtime.  \cite{2012.03348}  treat $a$ and $\beta$ as constants in their analysis, so if we are allowed to do the same then we still reproduce their complexities.
    \end{itemize}

    \begin{remark} As pointed out by \cite{2204.13641}, it is sometimes possible to construct oracles for a shifted amplitude $\frac{a+1}{2}$. For some state $\ket{0}$ say we have a unitary $U$ satisfying $U\ket{0} = \ket{\Psi}$ and another unitary $V$ such that $\Pi = V\ket{0}\bra{0}V^\dagger$. Then redefine:
        \begin{align}
            \ket{\bar\Psi} := (I \oplus V^\dagger U)\ket{+}\ket{0}\\
            \bar\Pi := \ket{+}\bra{+} \otimes \ket{0}\bra{0}
        \end{align}
        Observe that $|\bra{0} V^\dagger U \ket{0}| = |\Pi\ket{\Psi}| = a $, but suppose furthermore that $\bra{0} V^\dagger U \ket{0} = a$. Then, we see that:
        \begin{align}
            |\bar\Pi\ket{\bar\Psi}| &=  \left|  \bra{0}\cdot (\bra{+}\otimes I)(I\oplus V^\dagger U)(\ket{+}\otimes I) \cdot\ket{0}  \right|\\
            &=  \left|  \bra{0} \cdot\frac{I+V^\dagger U}{2} \cdot\ket{0}  \right| = \frac{1+a}{2}
        \end{align}
        If we use this oracle in our algorithm instead, then we have $a \geq 1/2$ so all the $O(a^{-1})$ terms in Theorem~\ref{thm:hybridestimation} become constant.
    \end{remark}

We begin by discussing the approximation of $\erf(kx)$.

    \begin{lemma} \label{lemma:erfpoly} (\cite{1707.05391} Corollary 4). For all $k,\eta>0$ there exists an odd polynomial $P_{\text{erf},k,\eta}$ such that for all $x \in [-2,2]$:
        \begin{align}
            \left| P_{\text{erf},k,\eta}(x)  - \erf(kx) \right| \leq \eta.
        \end{align}
        Furthermore $\text{deg}(P_{\text{erf},k,\eta}) \in O( \sqrt{ (k^2 + \log(1/\eta)) \log(1/\eta)  } )$.
\end{lemma}

    The construction in \cite{1707.05391} actually obtains an approximation for $x\in [-1,1]$, but this is easily expanded to $x \in [-2,2]$ by dividing the input by two and then doubling $k$.

    Now we leverage that $\erf(kx)$ can be used to approximate the line $kx$, and is close to $-1$ or $1$ when far from the origin. These properties are listed in the two facts below and also plotted in Figure~\ref{fig:erfproperties}. Note that the kind of approximation $\erf(kx) \approx kx$ we achieve here is sufficient for interval reduction, but not sufficient for the sampling performed for the unbiased estimation algorithm from the previous section. If this had been the case, we could have exponentially improved the accuracy of the unbiased estimation algorithm. 

    \begin{fact} \label{fact:erfline} For $x \in [-1,0]$ we have $ \erf(x) \leq 0.8x$. For $x \in [0,1]$ we have $0.8x \leq  \erf(x)$.
\end{fact}

    \begin{fact} \label{fact:erfbound} For any $\tau > 0$, let $\kappa(\tau) := \frac{1}{2} \sqrt{ 2 \ln( 2/ (\pi \tau^2)  )  }$. Then:
    \begin{align}
        x \leq -\kappa  \hspace{5mm}&\to\hspace{5mm} -1 \leq \erf(x) \leq -(1 - \tau) \\
        \kappa \leq x \hspace{5mm}&\to\hspace{5mm} 1-\tau \leq \erf(x) \leq 1
    \end{align}
\end{fact}

    \begin{figure}[h]
    \begin{center}
        \includegraphics[width=0.5\textwidth]{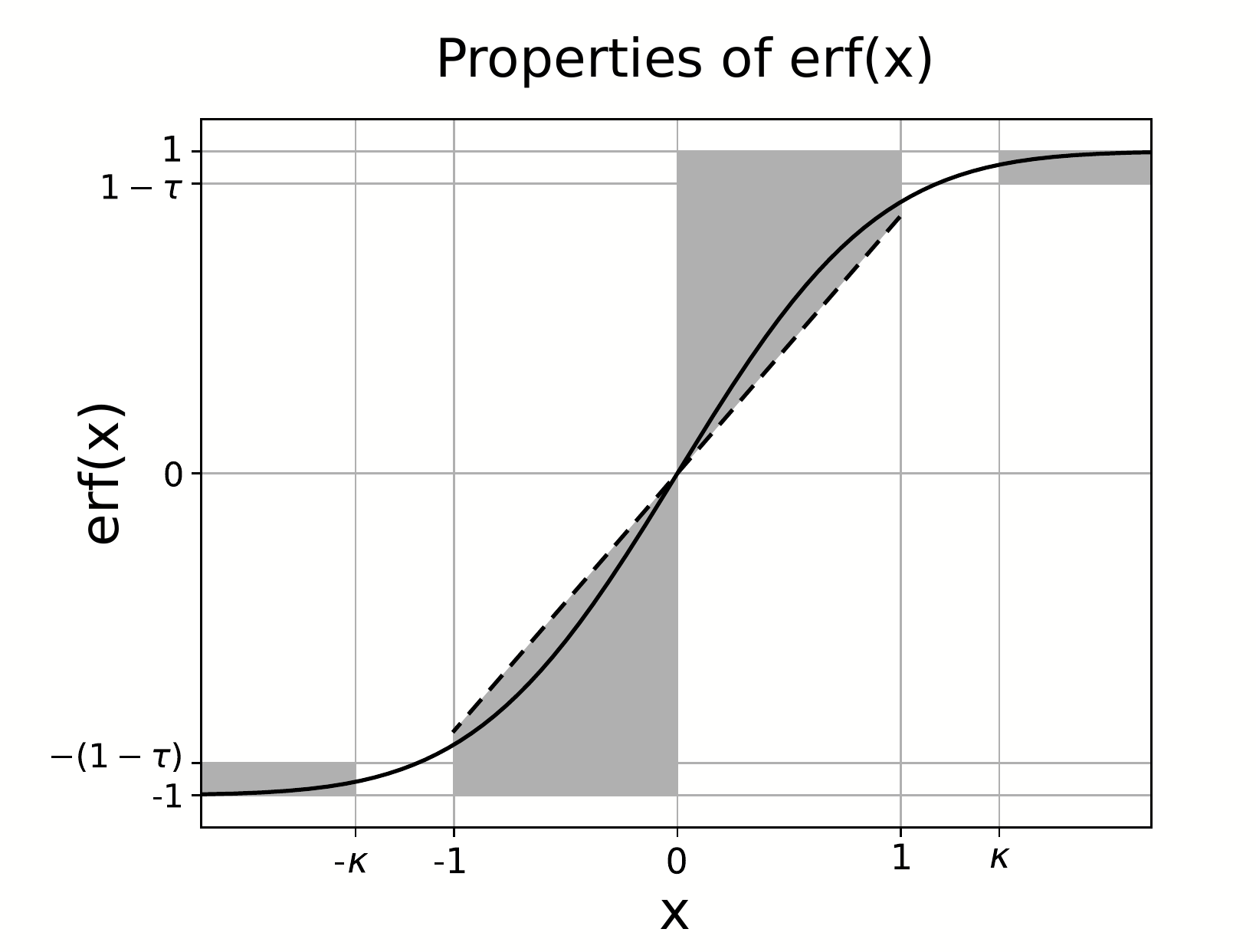}
    \end{center}
    \caption{\label{fig:erfproperties} Sketch of the properties of $\erf(x)$ stated in Facts~\ref{fact:erfline} and~\ref{fact:erfbound}. Here we selected $\tau = 0.1$, making $\kappa(\tau) \approx 1.44$. Fact~\ref{fact:erfline} asserts that $\erf(x)$  is either strictly above or below the dashed line $0.8x$ depending on the sign. Fact~\ref{fact:erfbound} asserts that $\erf(x)$ is within $\tau$ of $\pm 1$ when $x$ is $\kappa$-far from the origin. These bounds correspond to the gray regions. }
\end{figure}

Now we discuss the trick we explained earlier that lets us approximate the mixed-parity function $\erf(k(a-a_\text{mid}))$ with an even function on the interval $a \in [0,1]$. Think of $f(a) \approx \erf(k(a-a_\text{mid})) + 1$, so that $f(a)$ jumps from 0 to 1 on the interval $[a_\text{min},a_\text{max}]$. Then $f_e(a) := f(a) + f(-a)$ approximates $f(a)$. We track the quality of the approximation via the parameter $\lambda$.

    \begin{fact} \label{fact:evencombine} Say there is a function that satisfies $f(a) \leq \lambda$ for $a \in [-1,0]$, $f(a) \leq 1 - \lambda$ for $a\in[0,1]$, and $0 \leq f(a)$ for all $a\in[-1,1]$. Then $f_e(a) := f(a) + f(-a) $ is an even function satisfying $0 \leq f_e(a) \leq 1$ for $a \in [-1,1]$. Furthermore, when $a \in[0,1]$:
        \begin{align}
            f(a) \leq f_e(a) \leq  f(a) + \lambda.
        \end{align}
\end{fact}

We have all the tools in place to construct the family of polynomials for hybrid estimation. A sketch of the construction is shown in Figure~\ref{fig:hybridpoly}, which displays the critical property proven by the following lemma: that $P(a)$ is bounded away from $1/2$ for values close to $a_\text{min}$ and $a_\text{max}$. That bound will later become $\gamma$ from the previous discussion. The criterion $a_\text{mid} \geq \kappa/k$ allows us to determine when $a_\text{mid}$ is too close to the origin for the construction to work.

\begin{figure}[h]
     \centering
     \begin{subfigure}[b]{0.49\textwidth}
         \centering
         \includegraphics[width=0.9\textwidth]{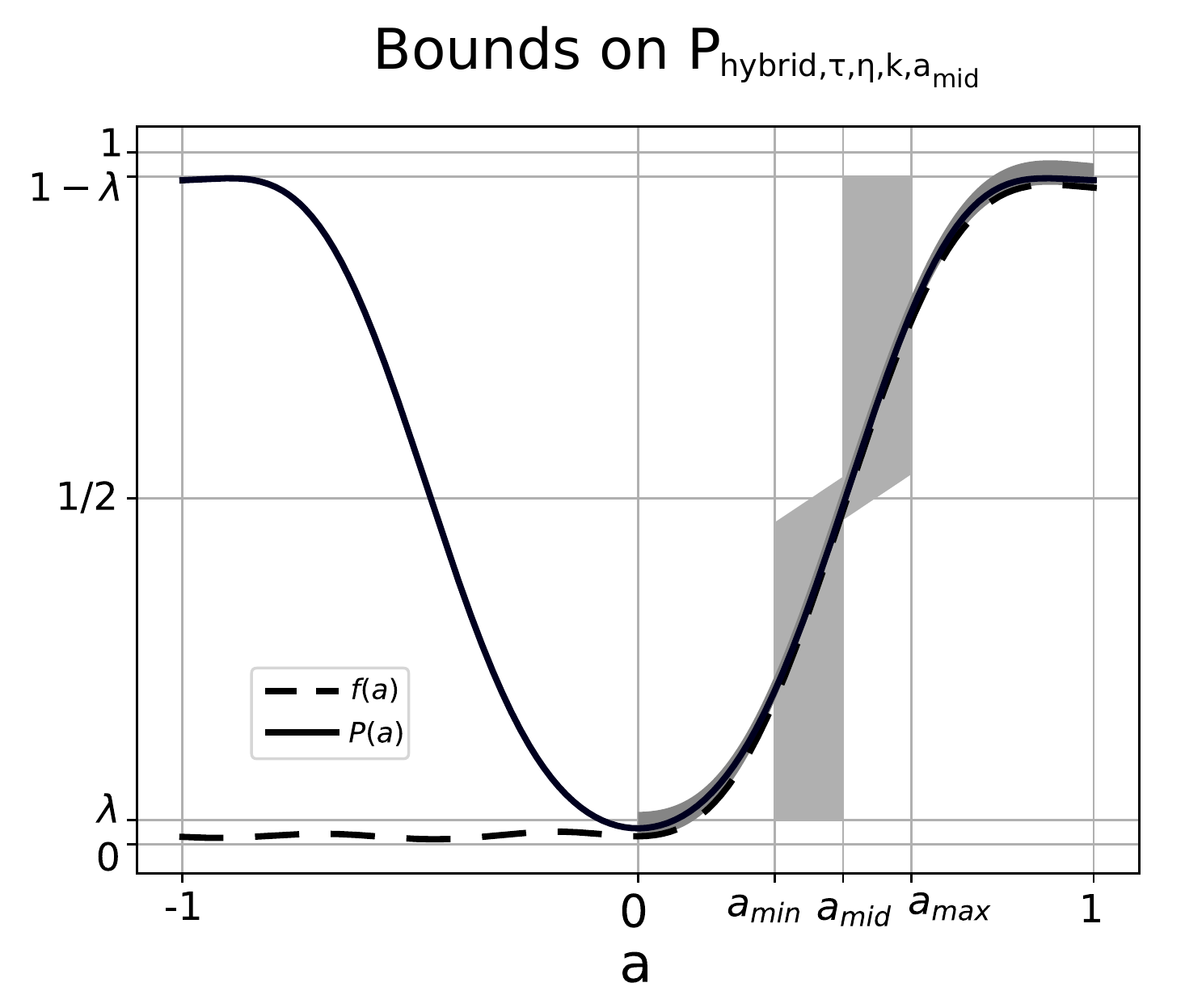}
         \caption{$[a_\text{min}, a_\text{max}] = [0.3,0.6]$}
     \end{subfigure}
     \begin{subfigure}[b]{0.49\textwidth}
         \centering
         \includegraphics[width=0.9\textwidth]{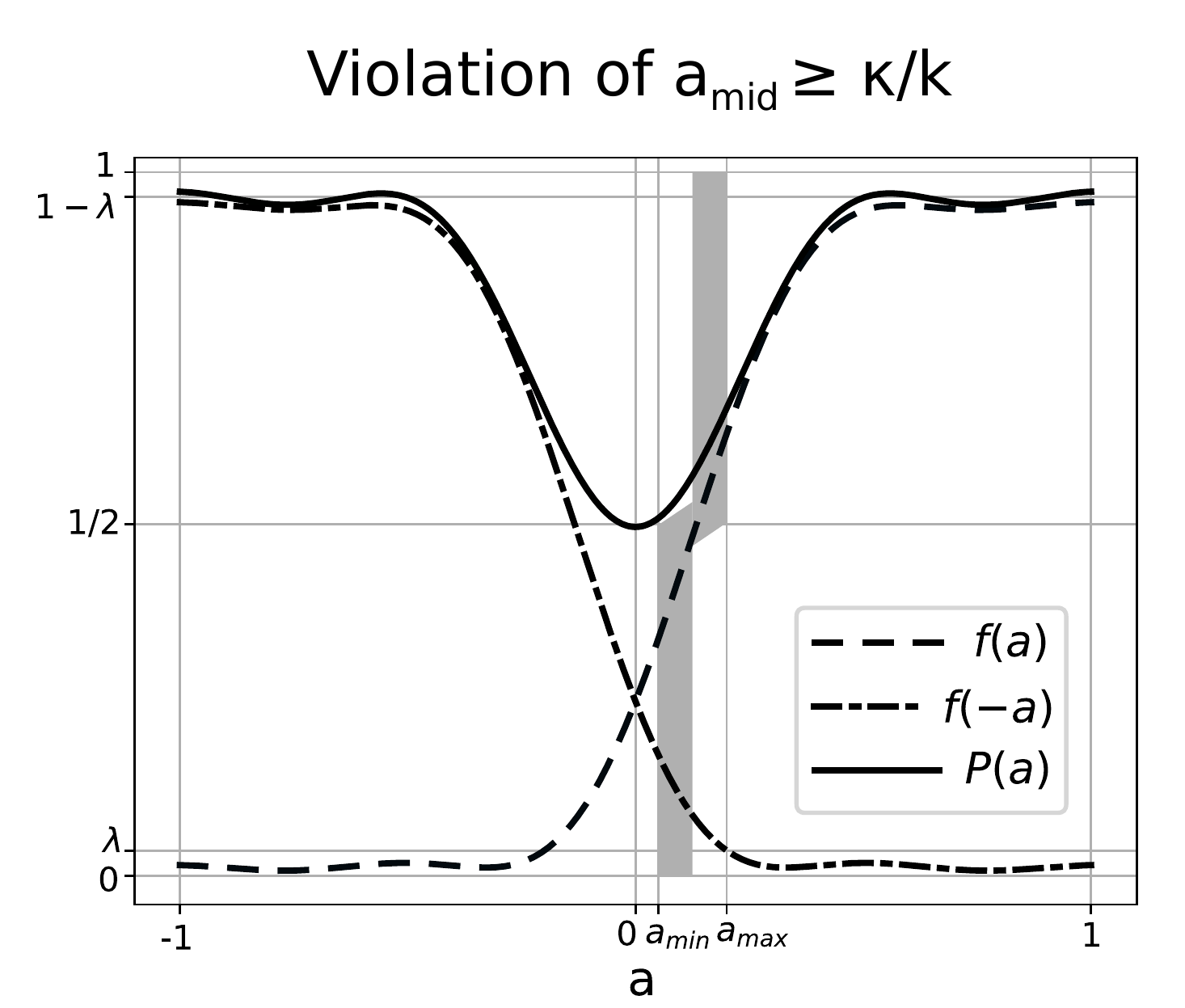}
         \caption{$[a_\text{min}, a_\text{max}] = [0.05,0.2]$}
     \end{subfigure}
     \hfill
        \caption{\label{fig:hybridpoly} Sketch of the construction of $P_{\text{hybrid},\tau,\eta,k,a_\text{mid}}(a)$
    or $P(a)$ for short from Lemma~\ref{lemma:hybridpoly}. Here we selected $\tau = \eta = 0.025$, $k=4$, which make $\kappa \approx 1.87$, $\kappa/k \approx 0.47$ and $\text{deg}(P_{\text{erf},k,\eta}) = 21$. In (a) we have $a_\text{mid} \geq \kappa/k$, so we see that $P(a)$ is contained in the gray regions asserted in $(\ref{eqn:hybridbound1},\ref{eqn:hybridbound2})$. The region at most $\lambda$ above $f(a)$ is highlighted in dark gray. But in (b) where $a_\text{mid} < \kappa/k$, we see that $f(-a)$ is large on the interval $[a_\text{min},a_\text{max}]$ causing the shape of $P(a)$ to be distorted and no longer be contained in the gray regions. }

\end{figure}

\begin{figure}
    \begin{center}
        \includegraphics[width=0.5\textwidth]{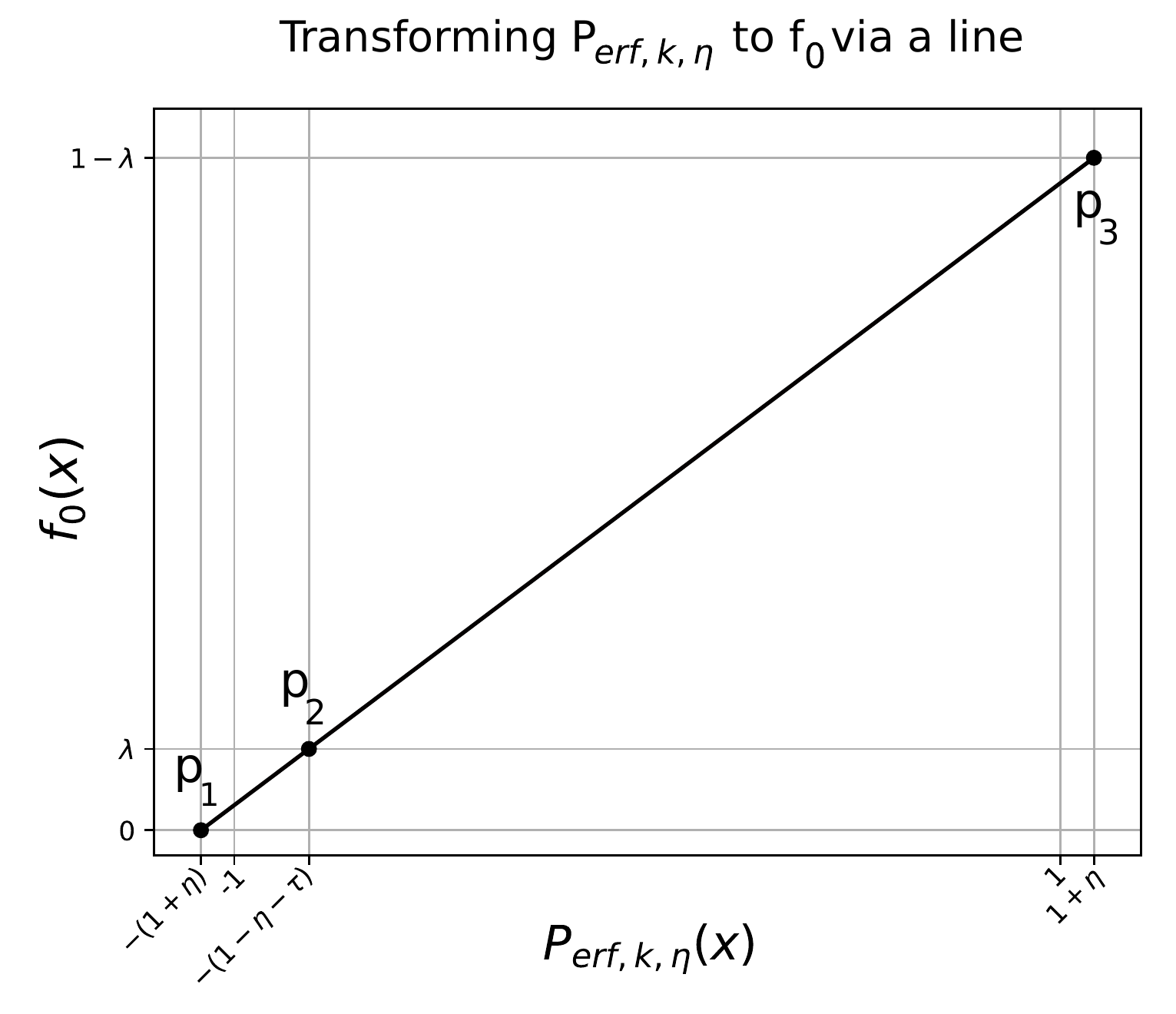}
    \end{center}
    \caption{\label{fig:colinear} Line defined by the points in (\ref{eqn:points}) used to transform $P_{\text{erf},k,\eta}$ to $f_0$ in the construction of $P_{\text{hybrid},\tau,\eta,k,a_\text{mid}}(a)$ from Lemma~\ref{lemma:hybridpoly}.} 
\end{figure}

    \begin{lemma} \label{lemma:hybridpoly} Consider any $\tau,\eta \in (0,1)$, and let $\kappa = \kappa(\tau)$ be as in Fact~\ref{fact:erfbound}. Consider any $a_\text{min},a_\text{max}$ that satisfy $0 < a_\text{min} < a_\text{max}< 1$. Furthermore, let $\Delta := a_\text{max} - a_\text{min}$ and consider any $k$ satisfying $1 \leq k \leq 2/\Delta$. Finally, let $a_\text{mid} := \frac{1}{2}(a_\text{min}+a_\text{max})$ and suppose that $a_\text{mid} \geq \kappa/k$.
 
        Then there exists an even semi-Pellian polynomial $P_{\text{hybrid},\tau,\eta,k,a_\text{mid}}(a)$ or $P(a)$ for short, that satisfies:
        \begin{align}
            a_\text{min}  \leq a \leq a_\text{mid} \hspace{5mm}&\to\hspace{5mm}   P(a) \leq \frac{1}{2} + 0.11 k(a - a_\text{mid}) + \eta + 0.25\tau \label{eqn:hybridbound1} \\
            a_\text{mid}  \leq a \leq a_\text{max} \hspace{5mm}&\to\hspace{5mm}   P(a) \geq \frac{1}{2} + 0.11 k(a - a_\text{mid}) - \eta - 0.25\tau. \label{eqn:hybridbound2} 
        \end{align}
        This polynomial has the same degree as $P_{\text{erf},k,\eta}(x)$  from Lemma~\ref{lemma:erfpoly}.
 \end{lemma}
    \begin{proof} We will obtain $P(x)$ by taking $P_{\text{erf},k,\eta}(x)$ and performing a series of degree-preserving transformations. Our first goal is to take $P_{\text{erf},k,\eta}(x)$, which is between $\pm(1+\eta)$, and $\leq -(1-\eta-\tau)$ for small $x$, and transform it to an $f_0(x)$ which is in $[0,1-\lambda]$, and $\leq \lambda$ for small $x$ as required by Fact~\ref{fact:evencombine}. Let $\lambda$ be chosen such that the following points are colinear:
        \begin{align}
            \vec p_1 :=  (  -(1+\eta) , 0); \hspace{5mm}  \vec p_2 :=(-(1-\eta-\tau), \lambda); \hspace{5mm}  \vec p_2 :=(1+\eta, 1-\lambda)  \label{eqn:points}
        \end{align}
We find this is solved by:
        \begin{align}
            \lambda := \frac{2\eta+\tau}{4\eta + \tau + 2} 
        \end{align}
        Furthermore, we find the line $y = (1 + x + \eta)/(4\eta+\tau+2)$ contains these points. Accordingly, let: 
        \begin{align}
            f_0(x) :=  \frac{1 + P_{\text{erf},k,\eta}(x) + \eta}{4\eta+\tau+2}
        \end{align}

        We want to plug a shifted $f_0$ into Fact~\ref{fact:evencombine}, so we will first prove that $f_0$ meets some requirements for all $x\in[-2,2]$. First, Lemma~\ref{lemma:erfpoly} implies that $P_{\erf,k,\eta}(x) \geq -(1+\eta)$, so since we pass through $\vec p_1$ we have $f_0(x) \geq 0$. Second,  $P_{\erf,k,\eta}(x) \leq 1+\eta$, so since we pass through $\vec p_3$ we also have $f_0(x) \leq 1-\lambda $. Third, Fact~\ref{fact:erfbound} implies that for $kx \leq -\kappa$ we have $\erf(kx) \leq -(1-\tau)$, so $P_{\erf,k,\eta}(x) \leq -(1-\eta-\tau)$, so since we pass through $\vec p_2$ we have $f_0 \leq \lambda$. 

        Now we are ready to use Fact~\ref{fact:evencombine}. Let $f(a) := f_0(a - a_\text{mid})$ with $x = a - a_\text{mid}$. We know that $a \in [0,1]$, but the polynomial construction must work for any $a \in [-1,1]$.  Since $a_\text{mid} \in [0,1]$ it is ensured that $x \in [-2,1] \subset [-2,2]$ so the approximation holds. Since the function was shifted to the right, we have $f(a) \geq 0$ for $a \in [-1,1]$, and we have $f(a) \leq 1-\lambda$ for $a \in [0,1]$. Since we assumed $a_\text{mid} \geq \kappa/k$, we know that for $a \in [-1,0]$ we have $kx = k(a - a_\text{mid}) \leq k a_\text{mid} \leq -\kappa$, which implies $f(a) \leq \lambda$. We meet the requirements of Fact~\ref{fact:evencombine}, so let $P_{\text{hybrid},\tau,\eta,k,a_\text{mid}}(a) := f(a) + f(-a)$.

        $P_{\text{hybrid},\tau,\eta,k,a_\text{mid}}(a)$, or $P(a)$ for short, is an even function by construction. Since it is a linear combination of polynomials it is itself a polynomial, and polynomials that are even functions are even polynomials. From Fact~\ref{fact:evencombine}, we know that $0 \leq P(a) \leq 1$, so $P(a)$ is semi-Pellian. The transformations from $P_{\erf,k,\eta}$ to $P(a)$ were all linear, so the degree is unchanged.

        It remains to show the desired bounds on $P(a)$. These stem from Fact~\ref{fact:erfline}. For $kx \in [-1,0]$ we have $\erf(kx) \leq 0.8kx$, so $P_{\erf,k,\eta}(x) \leq 0.8kx + \eta $, so:
        \begin{align}
            f_0(x) &\leq \frac{1 +(0.8kx +\eta)+\eta}{4\eta +\tau+2} \\
                    &= \frac{1}{2} +  \left[\frac{1 +(0.8kx +\eta)+\eta}{4\eta +\tau+2} - \frac{1}{2} \right] \\
                    &= \frac{1}{2} +  \frac{2(1 + 0.8kx + 2\eta) - (4\eta + \tau + 2)}{2(4\eta +\tau+2)} \\
                    &\leq \frac{1}{2} +  \frac{\tau}{2(4\cdot 0 + 0 +2)} - \frac{ 2\cdot 0.8(-x) }{2(4\cdot 1 +1 +2)} \\
                    &\leq \frac{1}{2} +  0.11 kx  - 0.25\tau\\
             f_0(x) + \lambda  &\leq \frac{1}{2} +  0.11 kx  + \eta + 0.25\tau
        \end{align}
        where we plugged in $0 \leq \eta,\tau \leq 1$ to obtain a very loose but workable bound, and observed $\lambda \leq \eta+\tau/2$. Similarly, for $kx \in [0,1]$ we have $\erf(kx) \geq 0.8kx - \eta$, so we derive:
            \begin{align}
            f_0(x) &\geq \frac{1 +(0.8kx -\eta)+\eta}{4\eta +\tau+2} \\
                    &= \frac{1}{2} +  \frac{2(1 + 0.8kx) - (4\eta + \tau + 2)}{2(4\eta +\tau+2)} \\
                    &\geq \frac{1}{2} +  \frac{2\cdot 0.8kx}{2(4\cdot 1 +1+2)} - \frac{ 4\eta + \tau }{2(4\cdot 0 +0+2)} \\
                    &\geq \frac{1}{2} +  0.11 kx - \eta  - 0.25\tau
        \end{align}
        The transformation $x = a-a_\text{mid}$ makes the intervals $kx\in[-1,0]$ and $kx\in[0,1]$ correspond to $a \in [a_\text{mid} -1/k , a_\text{mid}] $ and $a \in [a_\text{mid} , a_\text{mid} + 1/k] $ respectively. Since we assumed $k \leq 2/\Delta$, the intervals $[a_\text{min}, a_\text{mid}]$ and $[a_\text{mid},a_\text{max}]$ are subintervals. From Fact~\ref{fact:evencombine}, we thus derive the desired bounds via $f_0(x) \leq P(a) \leq f_0(x) + \lambda$.
    \end{proof}

We have yet to take into account the Born rule: when we sample from $P(a)$, we actually toss a coin that comes up heads with probability $|P(a)|^2$. Fortunately, as the following lemma shows, it suffices to bound $x$ away from $1/2$ even when tossing a coin with bias $x^2$ in order to decide if $x$ is big or small.

        \begin{lemma} \label{lemma:amppoly} Say there is a coin that comes up heads with probability $x^2$, and for any $\gamma,\delta>0$ let $A_{t,\gamma,\delta}(x)$ be the probability that more than $\frac{1}{4}+\gamma^2$ of $\lceil \frac{1}{2} \gamma^{-2} \ln( \delta^{-1} )  \rceil$ tosses come up heads. Then:
\begin{align}
    \text{if } x\geq \frac{1}{2} + \gamma \text{ then } A_{\gamma,\delta}(x) &\geq 1-\delta,\\
    \text{if } x\leq \frac{1}{2} - \gamma \text{ then } A_{\gamma,\delta}(x) &\leq\delta.
\end{align}
\end{lemma}
        \begin{proof} We can derive $x \geq \frac{1}{2} + \gamma$ implies $x^2 \geq \frac{1}{4} + \gamma^2  + \gamma $, and $x \leq t - \gamma$ implies $x^2 \leq \frac{1}{4} + \gamma^2 - \gamma$. 
            The claims on $A_{t,\gamma,\delta}$ follow from two applications of the Chernoff-Hoeffding theorem. Let $\mathbf{S}$ be the random variable denoting the number of heads, and let $m$ be the number of tosses.  Then, if $x^2 \geq \frac{1}{4} + \gamma^2 + \gamma$:
        \begin{align}
            \text{Pr}\left[ \mathbf{S} \geq m\left(\frac{1}{4} + \gamma^2\right)  \right] \leq \text{Pr}[ \mathbf{S} \geq \mathbb{E}(\mathbf{S}) + \gamma m  ] \leq e^{ - 2 \gamma^2 m }
        \end{align}
            Demanding $e^{ - 2\gamma^2 m} \leq \delta$ and solving for $m$ completes the proof. A similar argument follows for the bound when $x^2 \leq \frac{1}{4} + \gamma^2 - \gamma$.
\end{proof}

Having constructed the polynomial and shown how to use it to distinguish $a \leq a^{(t)}_\text{min} +0.1\Delta_t$ and $a \geq a^{(t)}_\text{max} -0.1\Delta_t$, all that remains to show is how to deal with the restriction that we must have $a^{(t)}_\text{mid} \geq \kappa/k$ for the polynomial construction to work properly. Recall that we would ideally like to pick the slope $k \propto 1/\Delta^{1-\beta}_t$. Since $\kappa$ is roughly constant, this essentially translates the condition $a^{(t)}_\text{mid} \geq \kappa/k$ to $a^{(t)}_\text{mid} \geq \frac{1}{2} \Delta_t^{1-\beta}$. If this is not the case, then we must go `full quantum' ($\beta = 0$) and set the slope to be $k \propto 1/\Delta_t$ instead. Then the condition is $a^{(t)}_\text{mid} \geq \frac{1}{2} \Delta_t$ which is always true.

We find that the consequences of this behavior are to knock out an extra $O(a^{-1})$ term in the total query complexity and an extra $O(a^{-1/(1-\beta)})$ term in the depth. Intuitively, these arise as follows. At iteration $t$ we have $a \in [a^{(t)}_\text{min}, a^{(t)}_\text{max}]$, so, implying that $a_\text{mid}$ is lower bounded by $a - \Delta_t/2  \geq a - \Delta_t^{1-\beta}/2 $. At a certain time $t^*$, $\Delta_t$ becomes small enough $(\Delta_t^{1-\beta} \leq a)$ so that the bound $a^{(t)}_\text{mid} \geq \frac{1}{2} \Delta_t^{1-\beta}$ is assured. So, the algorithm divides into two phases. In the first phase $(t \leq t^*)$ we have $\Delta_t^{1-\beta} \geq a$ and the algorithm is essentially fully quantum. The complexity and depth are then given by $O(a^{-1})$ and $O(a^{-1/(1-\beta)})$ respectively. In the second phase $(t > t^*)$ we are guaranteed to have $a^{(t)}_\text{mid} \geq \kappa/k$, so the algorithm has the desired complexity and depth of $O(1/\eps^{1+\beta})$ and $O(1/\eps^{1-\beta})$ respectively.

Recall that $\mathbf{d}$ is the maximum degree of any particular polynomial, and that $\mathbf{D}$ is the sum of the degrees of all the polynomials. These correspond to query depth and query complexity respectively.

\begin{figure}
    \begin{center}
        \includegraphics[width=0.5\textwidth]{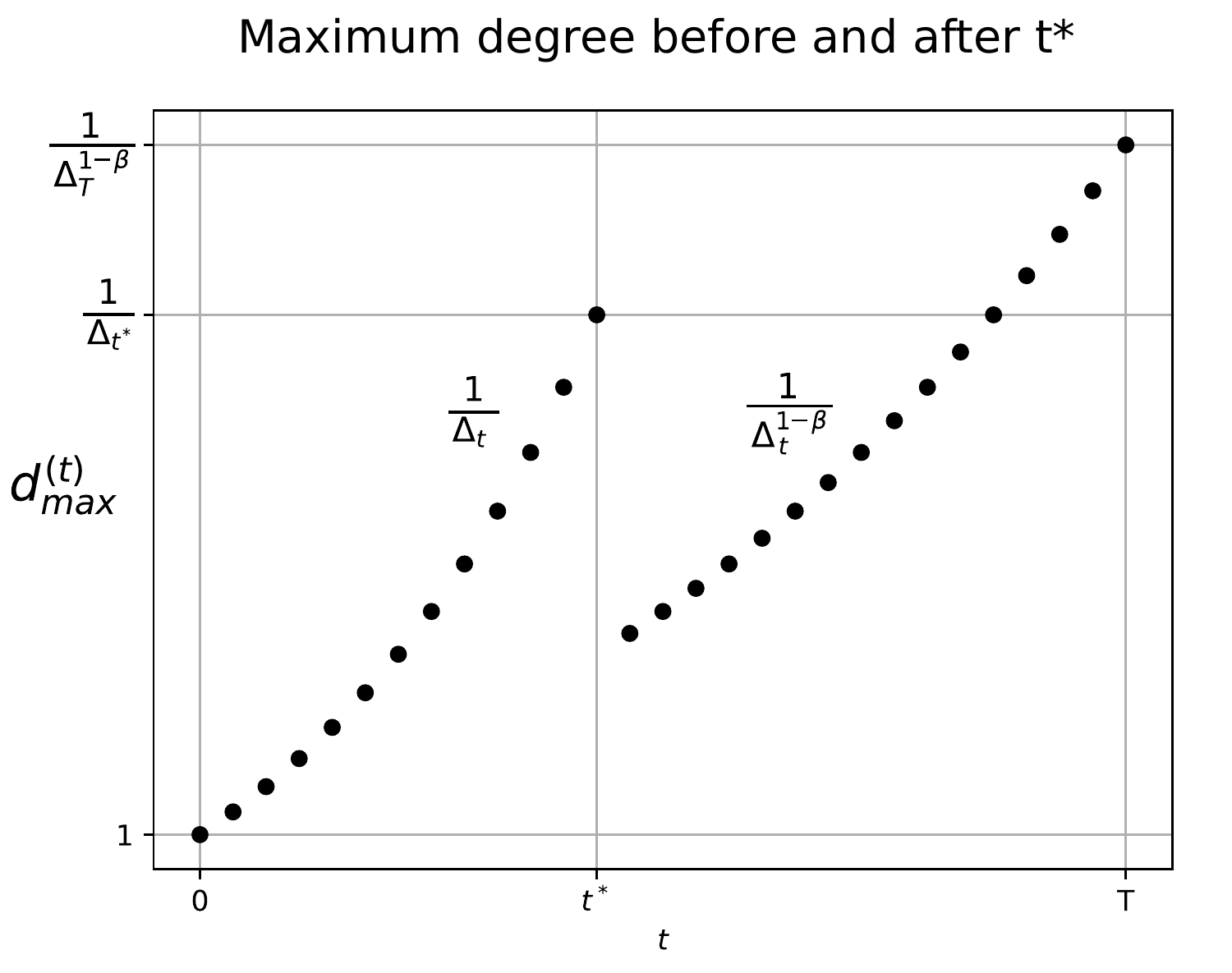}
    \end{center}
    \caption{\label{fig:tstarTmax} Sketch of an argument for bounding $\mathbf{d} := \max_t d^{(t)}_\text{max}$ in Theorem~\ref{thm:hybridestimation} in the case when $T > t^*$. When $t \leq t^*$ then there is a chance that $a^{(t)}_\text{mid} < \frac{1}{2}\Delta_t^{1-\beta}$ so we cannot rule out the more expensive scaling of $d^{(t)}_\text{max} \in O(\Delta^{-1}_t)$. But after $t > t*$, we have the cheaper scaling of $\tilde O( \Delta^{-(1-\beta)}_t )$. So the maximum must occur at either $t=t^*$ or $t=T$.}
\end{figure}

\begin{theorem} \label{thm:hybridestimation} \textbf{Hybrid quantum-classical estimation via polynomial sampling.} For any $\eps,\delta > 0$ and $0 \leq \beta < 1$, there exists an amplitude estimation algorithm satisfying:
     \begin{align}
        \text{Pr}[|\mathbf{\hat a}  - a| \geq \eps] \leq \delta
    \end{align}\\[-1.3cm]
    \begin{align}
        \exists f(\eps,a,\delta) \in \tilde O((a^{-1}  + \eps^{-(1+\beta)} )\log(\delta^{-1}))  \text{ such that } \Pr[ \mathbf{D} \geq f(\eps,a,\delta) ] \leq \delta \\
        \exists g(\eps,a) \in \tilde O( a^{-1/(1-\beta)} + \eps^{-(1-\beta)} )  \text{ such that } \Pr[ \mathbf{d} \geq g(\eps,a) ] \leq \delta 
    \end{align}
    Here, $\tilde O$ hides log factors in $\eps$.
\end{theorem}
    \begin{proof}
        The algorithm proceeds as follows:
        \begin{enumerate}
            \item Initialize $a^{(0)}_\text{min}\leftarrow 0$ and $a^{(0)}_\text{max} \leftarrow 1$. Let $\Delta_t := a^{(t)}_\text{max} - a^{(t)}_\text{min}$ and $a^{(t)}_\text{mid} = a^{(t)}_\text{min} + 0.5\Delta_t$ throughout.
            \item Let $T := \lceil \log_{0.9}(\eps) \rceil$. For $t = 0,1,2,...,T-1$, do:
                \begin{enumerate}
                    \item If $a^{(t)}_\text{mid} \geq \frac{1}{2}\Delta_t^{1-\beta}$, set:
                        \begin{align}
                            \eta_t, \tau_t &:= 0.01 \cdot \Delta_t^{\beta},\\
                             k_t &:=\frac{\kappa(\tau_t)}{2} \frac{1}{\Delta_t^{1-\beta}},\\
                             \gamma_t &:= 0.01 \cdot \Delta_t^{\beta}
                        \end{align}
                        Otherwise, set:
                        \begin{align}
                            \eta_t, \tau_t &:= 0.01,\\
                             k_t &:= \frac{\kappa(\tau_t)}{2}\frac{1}{\Delta_t},\\
                             \gamma_t &:= 0.01
                        \end{align}
                    \item Obtain $P_t(a)$ from Lemma~\ref{lemma:hybridpoly}. Set $\delta_t := \delta/T$ and $m_t := \lceil \frac{1}{2} \gamma_t^{-2}  \ln( \delta_t^{-1} )  \rceil$ and sample from $\mathbf{b}_{P_t}$ a total of $m_t$ times and let $\mathbf{S}^{(t)}$ be the total.
                    \item If $\mathbf{S}^{(t)} > m_t\left( \frac{1}{4} + \gamma_t^2 \right)$, then set:
                                                \begin{align}
                            [a^{(t+1)}_\text{min},\hspace{1mm}a^{(t+1)}_\text{max}] \leftarrow [a^{(t)}_\text{min} + 0.1\Delta_t,\hspace{1mm}a^{(t)}_\text{max}]
                        \end{align}
                             Otherwise, set:
                        \begin{align}
                            [a^{(t+1)}_\text{min},\hspace{1mm}a^{(t+1)}_\text{max}] \leftarrow [a^{(t)}_\text{min},\hspace{1mm}a^{(t)}_\text{max} -  0.1\Delta_t]
                        \end{align}
                \end{enumerate}
            \item Let $\mathbf{\hat a} \leftarrow a^{(T)}_\text{mid}$ (or anything else in the interval $[a^{(T)}_\text{min},a^{(T)}_\text{max}]$).
        \end{enumerate}
       
        We begin by proving correctness. At the beginning we have $\Delta_{0} = 1$, and each of the sub-steps of step 2 sets $\Delta_{t+1} = 0.9\Delta_{t}$, so $\Delta_t = 0.9^t$. At the end we have $\Delta_{T} = 0.9^T \leq \eps$, so, provided $a\in[a^{(T)}_\text{min},a^{(T)}_\text{max}]$, anything inside the interval $[a^{(T)}_\text{min},a^{(T)}_\text{max}]$ is within $\eps$ of $a$. So the algorithm returns correctly if at all iterations of step 2 we have $a\in[a^{(t)}_\text{min},a^{(t)}_\text{max}]$.

        We will show that the probability step 2c) adjusts the interval in a way to make $a\not\in[a^{(t+1)}_\text{min},a^{(t+1)}_\text{max}]$ is at most $\delta_t$. Then, by the union bound, the probability any of the steps fail is $\delta$.

        There are two ways we could fail to have $a\not\in[a^{(t+1)}_\text{min},a^{(t+1)}_\text{max}]$ after step 2c): either we have $a \in [a^{(t)}_\text{min}, a^{(t)}_\text{min}+0.1\Delta_t]$ and $\mathbf{S}^{(t)} > m_t(\frac{1}{4}+\gamma_t^2)$, or we have $a \in [a^{(t)}_\text{max} - 0.1\Delta_t, a^{(t)}_\text{max}]$ and $\mathbf{S}^{(t)} < m_t(\frac{1}{4}+\gamma_t^2)$. To show that these are both unlikely, we invoke Lemma~\ref{lemma:amppoly}. Say $a \in [a^{(t)}_\text{min}, a^{(t)}_\text{min}+0.1\Delta_t]$  implies $P_t(a) \leq \frac{1}{2} - \gamma_t$, and $a \in [a^{(t)}_\text{max} - 0.1\Delta_t, a^{(t)}_\text{max}]$  implies $P_t(a) \geq \frac{1}{2} + \gamma_t$. Then Lemma~\ref{lemma:amppoly} asserts that any of the events causing $a\not\in[a^{(t+1)}_\text{min},a^{(t+1)}_\text{max}]$ happen with probability at most $\delta_t$.

        So all that remains is to show the bounds on $P_t(a)$ in terms of $\gamma_t$. We can show this using Lemma~\ref{lemma:hybridpoly}, provided we meet its requirements. The requirement that $1 \leq k_t \leq 2/\Delta_t$ is always satisfied. But the requirement that $a^{(t)}_\text{mid} \geq \kappa(\tau_t)/k_t$, demanded that in step 2a) we select the parameters conditionally.
    
        We begin with the case where $a^{(t)}_\text{mid} \geq \frac{1}{2} \Delta_t^{1-\beta}$. Here, we see that we selected $k_t := \frac{1}{\Delta_t^{1-\beta}} \frac{\kappa(\tau_t)}{2} $, implying $\kappa(\tau_t)/k_t = \frac{1}{2} \Delta_t^{1-\beta}$. So we have $a^{(t)}_\text{mid} \geq \kappa(\tau_t)/k_t$ as desired. Also observe that $\tau_t < 1/4$, so $\kappa(\tau_t) > 1$, so $k_t \geq \frac{1}{2} \frac{1}{\Delta_t^{1-\beta}}$.  Then, Lemma~\ref{lemma:hybridpoly} demonstrates:
        \begin{align}
            a \in [a^{(t)}_\text{min}, a^{(t)}_\text{min}+0.1\Delta_t] \hspace{5mm}\to\hspace{5mm} P_t(a) &\leq \frac{1}{2} - 0.11 k_t (0.4 \Delta_t) + \eta_t + 0.25 \tau_t \\
            &\leq \frac{1}{2} - \left( 0.11\cdot\frac{1}{2}\cdot 0.4 - 0.01 - 0.25\cdot0.01  \right) \cdot \Delta_t^{\beta} \\
            &\leq \frac{1}{2} - 0.01 \cdot \Delta_t^{\beta} = \frac{1}{2} - \gamma_t\\
            a \in [a^{(t)}_\text{max}-0.1\Delta_t, a^{(t)}_\text{max}] \hspace{5mm}\to\hspace{5mm} P_t(a) &\geq \frac{1}{2} + 0.11 k_t (0.4 \Delta_t) - \eta_t - 0.25 \tau_t \\
            & \geq \frac{1}{2} + \gamma_t
        \end{align}
        In the case where $a^{(t)}_\text{mid} \geq \frac{1}{2} \Delta_t^{1-\beta}$, we instead select $\eta_t,\tau_t = 0.001$, implying $\kappa(\tau_t) > 1$ as before.  That way $\frac{\kappa(\tau_t)}{k_t} \leq \frac{1}{2} \Delta_t $. We always have $\frac{1}{2} \Delta_t \leq a^{(t)}_\text{mid} $, so the conditions of Lemma~\ref{lemma:hybridpoly} are satisfied. So, we derive:
    \begin{align}
        a \in [a^{(t)}_\text{min}, a^{(t)}_\text{min}+0.1\Delta_t] \hspace{5mm}\to\hspace{5mm} P_t(a) &\leq \frac{1}{2} - 0.11 \left(\frac{1}{2}\cdot \frac{1}{\Delta_t}  \right) (0.4 \Delta_t) +\eta_t + 0.25\tau_t \\
        &\leq \frac{1}{2} - \left(0.11 \cdot \frac{1}{2} \cdot 0.4 - 0.01 - 0.25\cdot0.01  \right) \\
        &\leq \frac{1}{2} - 0.01 \geq \frac{1}{2} - \gamma_t\\
        a \in [a^{(t)}_\text{max}-0.1\Delta_t, a^{(t)}_\text{max}] \hspace{5mm}\to\hspace{5mm} P_t(a) &\geq \frac{1}{2} + 0.11 \left(\frac{1}{2}\cdot \frac{1}{\Delta_t}  \right) (0.4 \Delta_t) - \eta_t - 0.25\tau_t \\
        &\geq \frac{1}{2} + 0.01 \geq \frac{1}{2} + \gamma_t
    \end{align}
So in both cases we meet the assumptions of Lemma~\ref{lemma:amppoly}, completing the correctness portion of the proof.

        Now we prove the properties of $\mathbf{D}$ and $\mathbf{d}$. To do so, we bound the maximum degree $d^{(t)}_\text{max}$ and the sum of all the degrees $d^{(t)}$ at each iteration of step 2. The following bounds rely on the guarantee that $a \in [a^{(t)}_\text{min}, a^{(t)}_\text{max}]$ for all $t$, which is only true probabilistically. Hence the bounds are themselves probabilistic.

        We start with $\mathbf{d}$ via $d^{(t)}_\text{max}$. The degree $d^{(t)}_\text{max}$ is given by Lemma~\ref{lemma:erfpoly}, and since we always have $\log(1/\eta_t) \leq k_t$ we have $d^{(t)}_\text{max} \in O(k_t \sqrt{\log(1/\eta_t)})$.   If $a^{(t)}_\text{mid} \geq \frac{1}{2} \Delta_t^x$, then we have $d^{(t)}_\text{max} \in O( \Delta_t^{-(1-\beta)} \sqrt{ \log( \Delta_t^{-2\beta})} \cdot \sqrt{\log( \Delta_t^{-\beta} )}    ) \leq \tilde O( \Delta_t^{-(1-\beta)} )$ where $\tilde O$ hides log factors in $\Delta_t$. Otherwise if $a^{(t)}_\text{mid} < \frac{1}{2} \Delta_t^{1-\beta}$, we have $d^{(t)}_\text{max} \leq O( \Delta_t^{-1} )$.

        We have $\mathbf{d} = \max_t d^{(t)}_\text{max}$. To maximize over all $t$, we must analyze the condition $a_\text{mid} \geq \frac{1}{2} \Delta_t^{1-\beta}$, since when this is the case we have the cheaper degree of $O( \Delta_t^{-(1-\beta)} )$. We argue that there is some $t^*$ such that for $ t > t^*$ the algorithm is guaranteed to satisfy $a^{(t)}_\text{mid} \geq \frac{1}{2} \Delta_t^{1-\beta}$.

        Observe that $a - \Delta_t/2 \leq a^{(t)}_\text{mid} \leq a + \Delta_t/2$, and since $\Delta_t \leq \Delta_t^{1-\beta}$ we also have $a^{(t)}_\text{mid} \geq a - \Delta_t^{1-\beta}/2$. That means that as soon as $\frac{1}{2} \Delta_t^{1-\beta} \geq a - \Delta_t^{1-\beta}/2 $, we are guaranteed that $a^{(t)}_\text{mid} \geq \frac{1}{2} \Delta_t^{1-\beta}$. This happens when $a \geq \Delta_t^{1-\beta}$, that is, when $t > t^* := \lceil\log_{0.9}(a)\rceil$.

       Let's assume $T > t^*$. Before $t^*$ we have $d^{(t)}_\text{max} \in O( \Delta_t^{-1} )$, and afterwards we have $d^{(t)}_\text{max} \in \tilde O( \Delta_t^{-(1-\beta)} )$. Both of these functions are increasing in $t$, so we maximize $d^{(t)}_\text{max}$ over all $t$ by just considering their values at $t^*$ and $T$. See Figure~\ref{fig:tstarTmax}.  We saw above that at $t^*$ we have $\Delta_t \geq a^{1/(1-\beta)}$, and at $T$ we have $\Delta_t \leq \eps$. So we have established: 
        \begin{align}
            \max_t d^{(t)}_\text{max} \in  \max\left[O( a^{-1/(1-\beta)} ), \tilde O(\eps^{-(1-\beta)} )  \right]   = \tilde O\left( a^{-1/(1-\beta)} + \eps^{-(1-\beta)}   \right)
        \end{align}
    where the final $\tilde O$ hides log factors in $\eps$, but not log factors in $a^{1/(1-\beta)}$. 
    
    Otherwise, assume $T \leq t^*$. Then we have $d^{(t)}_\text{max} \in O(\Delta_t^{-1})$ always, and since this function is increasing the maximum occurs at $t = T$. But at $t \leq t^*$ we have $a < \Delta_t^{1-\beta}$, which implies $\Delta_t^{-1} \leq a^{-1/(1-\beta)}$. So in this case we again have $\max_t d^{(t)}_\text{max} \in O(a^{-1/(1-\beta)})$. This establishes the bound on $\mathbf{d}$ for both cases.

        Now we bound the sum of the degrees $\mathbf{D}$ via the sum of all degrees at a particular iteration $d^{(t)}$. We simply have $d^{(t)} = m_t d^{(t)}_\text{max}$, so in the $t \leq t^*$ case we have $d^{(t)} \in O( \Delta_t^{-1} \log(\delta_t^{-1}) )$, and in the $t > t^*$ case we have $d^{(t)} \in \tilde O( \Delta_t^{-(1-\beta)} \cdot \Delta_t^{-2\beta}  \log(\delta_t^{-1}) ) = \tilde O(\Delta_t^{-(1+\beta)}  \log(\delta_t^{-1}) )$. Since $\delta_t$ is independent of $t$, we can just bound the first factor in the products, recalling that $\Delta_t = 0.9^t$:
        \begin{align}
            \sum_{t=0}^{t^*}\Delta_t^{-1} + \sum_{t=t^*+1}^{T-1}\Delta_t^{-(1+\beta)} &\leq \frac{1}{0.9} \sum_{t=0}^{t^*-1}\Delta_t^{-1} + \sum_{t=0}^{T-1}\Delta_t^{-(1+\beta)} \\
            &= \frac{1}{0.9} \frac{ 0.9^{-t*} -1  }{ 0.9^{-1} -1 } + \frac{ (0.9^{T})^{-(1+\beta)} -1  }{0.9^{-(1+\beta)} -1 }\\
            &\in O( a^{-1} + \eps^{-(1+\beta)} )
        \end{align}
Looking at the $\log(\delta_t^{-1})$ term, we have $T \in O( \log(1/\eps) )$, so we have $\log(\delta_t^{-1}) \in \tilde O( \log(\delta^{-1}) )$ where the $\tilde O$ hides log factors in $\eps$. This establishes the bound on $\mathbf{D}$.
    \end{proof}

\ifdefined\maindocument
\else
    
    \end{document}
\fi

\section{Acknowledgements}
The authors are grateful for fruitful discussions with Stefan Woerner, Julien Gacon, Giacomo Nannicini, John Martyn, Nikitas Stamatopoulos, and Pawel Wocjan. Part of this work was completed while PR was supported by Scott Aaronson's Vannevar Bush Faculty Fellowship, although the majority of the work was performed at IBM Quantum.

\appendix

\end{document}